\documentclass{article}
\usepackage{amsmath}
\usepackage{graphicx}
\usepackage{enumerate}
\usepackage{natbib}
\usepackage{graphicx} 
\usepackage{tikz}
\usetikzlibrary{positioning,arrows.meta}
\usepackage[left=1in,top=0.8in,right=1in,bottom=0.7in,head=.2in,nofoot]{geometry}

\usepackage[utf8]{inputenc}
\usepackage[bookmarksopen=true,bookmarks=true,pdfencoding=auto,psdextra,colorlinks]{hyperref}
\hypersetup{ 
    colorlinks,
    linkcolor={blue!80!black},
    citecolor={green!40!black},
    urlcolor={red!50!black}
}
\usepackage{bookmark}
\usepackage{mystyle}

\allowdisplaybreaks
\usepackage{xspace}

\usepackage[svgnames,dvipsnames,table]{xcolor}

\usepackage[font=small,labelfont=bf]{caption}
\usepackage{tabularx}
\usepackage{multirow}
\usepackage{array}
\newcolumntype{Y}{>{\centering\arraybackslash}X}
\newcolumntype{P}{>{\raggedleft\arraybackslash}X}
\newcolumntype{C}{ >{\centering\arraybackslash} m{4cm} }
\newcolumntype{D}{ >{\raggedright\arraybackslash} m{0.72\textwidth} }

\usepackage{algorithm}
\usepackage{algpseudocode}

\usepackage{tikz}

\newcommand{\BC}{\bar c}

\usepackage[normalem]{ulem}
\newcommand{\stkout}[1]{\ifmmode\text{\color{red}\sout{\ensuremath{#1}}}\else\sout{#1}\fi}

\usepackage[
    textsize = scriptsize,
    bordercolor = White,
    textcolor =  Black,
    backgroundcolor = Yellow!10,
]{todonotes}

\title{AI-Augmented Statistical Network Estimation\\
with Proxy Gene Embeddings}
\author{
\begin{tabular}{ccc}
Yan Chen & Weijing Tang & Jin-Hong Du \\
The University of Hong Kong & Carnegie Mellon University &The University of Hong Kong \\
\texttt{yanc25@hku.hk} & \texttt{ weijingt@andrew.cmu.edu} & \texttt{jinhongd@hku.hk}
\end{tabular}
}
\date{\today}

\begin{document}

\maketitle

\begin{abstract}
Gene--gene networks are often observed only on a restricted target set, while modern biomedical foundation models provide proxy gene embeddings over substantially larger gene universes. 
To leverage externally learned representations to improve latent-structure recovery in partially observed target networks, we propose \emph{Proxy-Latent Assisted Network Estimation} (PLANE), an adaptively weighted joint network--embedding latent variable model.
PLANE combines the two sources of information through the common latent positions of the target network and proxy embeddings.
Under mild rank conditions, the target network enables the identification of latent positions and loading of all nodes up to an orthogonal rotation.
We show that zero-order optimality analyses sharply control the weighted reconstruction loss, but are insufficient to identify the optimal weighting.
To understand the network and embedding information trade-off for latent-factor recovery, we analyze blockwise Gram-normalized gradient descent and prove deterministic contraction of aligned, curvature-weighted errors up to an explicit statistical tolerance.
We then specialize the weighted statistical error bound to derive the target-block error bound, yielding an optimal, data-adaptive choice of the network embedding weights.
Simulations and single-cell perturbation analyses show that informative proxy embeddings improve latent recovery, network reconstruction, and imputation beyond the observed target network.
\end{abstract}

\section{Introduction}
    In modern biology, a central goal is to recover interaction structure among molecular entities from noisy, high-dimensional measurements. In genomics, gene--gene networks are widely used to summarize coordinated expression patterns, identify biologically coherent modules, and prioritize disease-relevant mechanisms, as illustrated by methods such as weighted gene co-expression network analysis and network-integrative procedures  \citep{langfelder2008wgcna,liu2015network}. This perspective is especially natural in single-cell and perturbational genomics, where understanding gene--gene relationships is often essential for characterizing cellular programs and regulatory circuits \citep{aibar2017scenic,zhen2025network}.
    At the same time, recent foundation models trained on large biomedical corpora now provide gene- or cell-level embeddings that encode complementary notions of biological similarity, learned from sources such as single-cell atlases, biomedical text, and protein sequences \citep{theodoris2023transfer, cui2024scgpt, luo2022biogpt, lin2023esm2}. These developments raise a basic statistical question: when the target of interest is a latent gene--gene network, can such externally derived embeddings serve as useful proxies for latent structure and thereby improve estimation beyond what is possible from the observed network alone?

Existing methods study covariate-assisted community detection with observed node variables, where covariates are used as auxiliary predictors of communities or edge formation \citep{binkiewicz2017covariate,yan2021covariate,hu2024network}. 
Relatedly, \citet{zhang2022joint} treats embeddings as noisy proxies in a joint latent-space model, assuming that the embeddings are observed on the same nodes in the network. 
However, we consider a different regime where proxy embeddings are available on a set of genes \(P\), while the target network is observed only on a subset \(Q\subseteq P\). 
In this setting, embeddings for the extra genes \(Q^c := P\setminus Q\) are statistically useful: they help estimate the shared latent geometry and enable connectivity imputation beyond the observed target subnetwork. 
This distinguishes our setting from standard covariate-assisted network models and from transfer-learning approaches that assume access to additional observed networks rather than auxiliary embeddings
\citep{jalan2024transfer}. Therefore, the central question is  \emph{how to integrate the observed target network and the proxy embedding matrix to improve latent position estimation, network denoising, community detection, and imputation relative to using either source alone?}

    Formally, let \(Q\) denote the target genes with an observed network and \(P\supseteq Q\) the larger gene set with proxy embeddings. 
    We observe a target adjacency matrix $A_Q \in \RR^{n_Q \times n_Q}$ of $n_Q$ genes in $Q$ and the proxy embedding matrix \(W\in \RR^{n_P\times d}\) of $n_P$ genes in $P$.
    To link the two types of observation data, we consider a low-dimensional latent representation \(U \in \RR^{n_P \times k}\) whose submatrix \(U_Q\in \RR^{n_Q \times k}\) denotes the target genes.
    The latent position $U$ then generates the network and embedding through:
    \begin{align}
        \EE[A_Q \mid U_Q] = U_Q U_Q^\top,
        \qquad
        \EE[W \mid U] = U B^\top,\label{eq:expect-model}
    \end{align}
    where \(B\in\RR^{d\times k}\) maps the latent factors into the embedding space. 
   {This low-rank structure is viewed as a latent-factor approximation to the proxy embeddings, because in biology, the latent factors may be relevant to shared programs such as pathway activity, transcriptional regulation, protein interaction, and functional similarity.}
    Our primary goal is to recover \(U_Q\), enabling target-network denoising, community detection, and link imputation for genes in \(Q^c\). 
    Towards this goal, we propose PLANE, an adaptively weighted joint estimator that combines network and embedding reconstruction losses through tuning parameters \((\lambda_1,\lambda_2)\), as summarized in \Cref{fig:overview}. 
    This formulation generalizes existing joint latent-space models by allowing proxy embeddings on additional nodes outside the observed network \citep{zhang2022joint}, thereby raising the central statistical question of how the two channels should be optimally balanced for latent structure recovery.

    \begin{figure}[t]
        \centering
\includegraphics[width=0.9\linewidth]{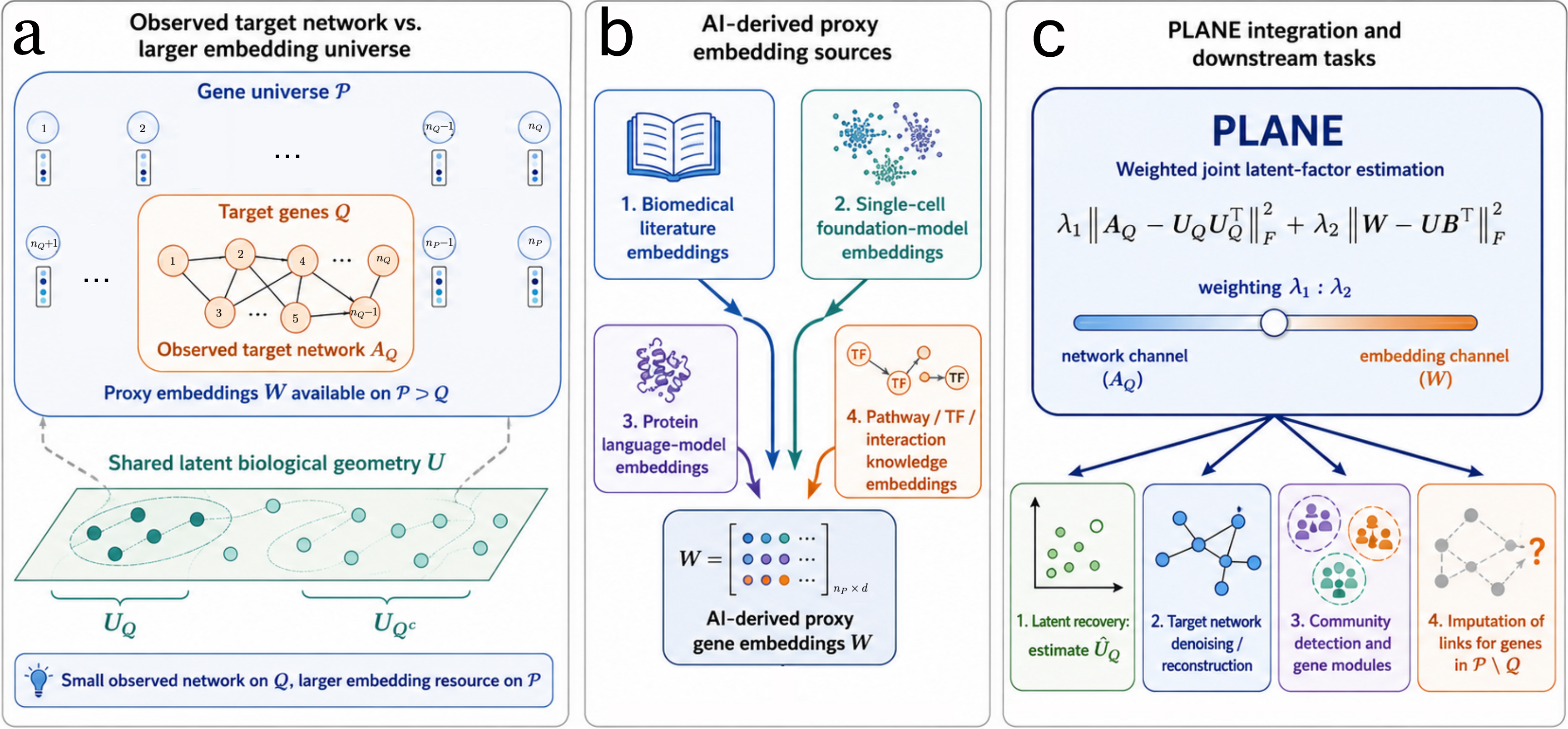}
        \caption{
        Overview of PLANE.
         PLANE models both sources as noisy observations of a shared latent biological geometry \(U\), and jointly balances the network and embedding channels to recover target latent positions, denoise the target network, detect gene modules, and impute links.
        }
        \label{fig:overview}
    \end{figure}

\subsection{Contributions}
\noindent \textbf{Model formulation, identification, and the limitation of zero-order analysis.}
    We formulate the ``proxy embedding'' problem through a joint latent-factor model for $(A_Q,W)$, in which the network block $A_Q$ and the embedding proxy $W$ share the same latent structure but are observed through different channels.
    Under mild rank conditions, we establish identification of $(U,B)$ up to an orthogonal transformation (\Cref{pro:identi}).
    We also show that a zero-order analysis of the weighted objective yields a sharp bound for the weighted loss (\Cref{thm:model}), but this type of argument is inherently insufficient for recovering the statistical trade-off induced by $(\lambda_1,\lambda_2)$ in the latent factors themselves.
    This gap motivates a gradient-based analysis of the estimation dynamics.

\noindent \textbf{Normalized gradient descent with a clean aggregate recursion.}
    We study a blockwise normalized gradient descent scheme (\Cref{alg:joint-UB-nGD}) for jointly estimating $U$ and $B$ under the weighted network--embedding loss.
    Our update is scaled by block-specific Gram matrices, which makes the local geometry more isotropic and leads to a substantially cleaner recursion in an aligned weighted error metric.
    By analyzing the aligned iterates and the aggregate error $e_t$ that couples $(U_Q,U_{Q^c},B)$, we derive a non-asymptotic contraction bound with an explicit statistical tolerance (\Cref{thm:error_algor}).
    This extends gradient-dynamics analysis from a single network objective to a joint network--factor model, while making the roles of $(\lambda_1,\lambda_2)$ and the two noise sources $(D,E)$ explicit.

    To the best of our knowledge, this is the first work to characterize the overall
effect of joint estimation in a factor-assisted network estimation problem. The
main theoretical difficulty is that the estimator couples three interacting
updates: the target-network latent factors, the proxy-embedding loadings, and the
shared latent representation for the enlarged node set. These updates affect one
another across iterations, so the analysis must control both statistical error
propagation and computational convergence simultaneously. An additional challenge
is the rotational non-identifiability of the latent factors, which requires
careful alignment through suitable rotation matrices at each step of the
argument. For this reason, defining an appropriate weighted error metric is
essential: it provides a common scale on which the network and embedding
components can be compared, aligned, and controlled throughout the iterative
analysis.

\noindent \textbf{Trade-off characterization and sharpness of the weighting rule.}
    Our analysis yields an explicit bound for estimating $U_Q$ (\Cref{cor:weighted-error-UQ}) that separates the network channel and the embedding channel and quantifies how  $(\lambda_1,\lambda_2)$ balance the two sources of information. 
    Within the admissible weighting regime, the bound is no worse than the network-only benchmark and, in the unconstrained oracle comparison, no worse than the better single-channel bound.
    We further prove that this dependence on $\lambda_j$'s is sharp, up to constants, in the special case of a weighted least-squares oracle (\Cref{prop:tightness_partial_oracle_short}).

    \noindent\textbf{Empirical validation of adaptive weighting and biological utility.}
    Simulations show that PLANE adaptively upweights informative proxy embeddings, downweights null proxies, and improves target latent recovery when embedding-only nodes are aligned with the target network. 
In the Norman CRISPRa Perturb-seq application, PLANE improves held-out target-network prediction over the network-only baseline and yields biologically interpretable modules, demonstrating the practical value of proxy embeddings beyond the observed target graph.

\subsection{Related work}
\noindent\textbf{Latent network models.}
A large literature estimates network structure through latent-space, low-rank, or graphon-based models.
Classical latent-space models assign each node an unobserved low-dimensional position and relate edge probabilities to distances or similarities between positions. 
A closely related finite-rank subclass includes eigenmodels and random-dot-product graphs, which use bilinear or inner-product latent interactions and therefore induce low-rank conditional mean adjacency matrices \citep{hoff2002latent,hoff2007modeling,sussman2014consistent,rubinDelanchy2022statistical}.
Graphon models provide a nonparametric framework for exchangeable networks and characterize the statistical difficulty of estimating network structure \citep{diaconis2008graph,olhede2014network,gao2015rate}. 
Closest to our setting, \citep{ma2020universal} studies latent-space fitting for large networks with edge covariates, which use covariates as observed modifiers of edge probabilities on the same node set, whereas our proxy embeddings are noisy observations of the latent positions themselves on a possibly larger network.

\noindent\textbf{Covariate-assisted and joint network--feature models.}
Another related line incorporates node covariates into network estimation and community detection. 
Covariate-assisted spectral methods combine graph and feature information to improve clustering, with extensions to sparse graphs, multilayer networks, network-adjusted covariates, and robust high-dimensional settings \citep{binkiewicz2017covariate,yan2021covariate,xu2023covariate,hu2024network,zhao2025robust}. 
Model-based approaches instead incorporate covariates directly into the network-generating mechanism, including pairwise covariate-adjusted block models and joint latent-space models for networks with high-dimensional node variables \citep{huang2024pcabm,zhang2022joint}. 
In contrast, our embeddings are not ordinary covariates with direct effects on edge formation; they are noisy proxy observations of the latent structure and are available on a larger node set than the observed target network.

\noindent\textbf{Transfer learning, multi-view networks, and biological foundation-model embeddings.}
A third related line studies transfer learning and shared latent structure across multiple networks or heterogeneous network views. 
Transfer-learning methods for latent variable network models use an additional observed source network to improve target-network estimation under structural similarity assumptions, while multi-view network methods estimate shared and view-specific latent components across related networks \citep{jalan2024transfer,Tian2024EfficientAO}.
In contrast, we study a single target network observed only on \(Q\), together with proxy embeddings observed on a larger set \(P \supseteq Q\), and focus on how the network and embedding channels should be weighted for latent recovery and network imputation.

\section{Proxy-Assisted Latent Network Model}

In this section, we study the identifiability under model \eqref{eq:expect-model}.
For any matrix $X$ with $n_P$ rows, let $X_Q$ and $X_{Q^c}$ denote the
submatrices formed by the rows indexed by $Q$ and $Q^c$, respectively.
In particular, $U_Q$ consists of the first $n_Q$ rows of $U$, and
$U_{Q^c}$ consists of the remaining $n_P-n_Q$ rows. 
Let $\cO_k$ denote the set of $k\times k$ orthonormal matrices.

   \begin{proposition}[Identification]\label{pro:identi}
    Under data model \eqref{eq:expect-model}, suppose that
    (i) $U_Q \in \mathbb{R}^{n_Q\times k}$ has full column rank,
    (ii) $B \in \mathbb{R}^{d \times k}$ has full column rank with $d\ge k$.
    Then the pair $(U,B)$ is identifiable up to a $k\times k$ orthogonal transformation.
    That is, if $(U,B)$ and $(U^\prime,B^\prime)$ generate the same distribution, then there exists an orthonormal matrix $O\in\cO_k$ such that
    $U^\prime=UO$ and $B^\prime=BO$.
   \end{proposition}

    Proposition~\ref{pro:identi} shows that the network channel is essential for identifying the latent positions beyond the usual factor-model ambiguity.
    Specifically, the factor model on proxy $W$ alone is invariant under any invertible transformation, unless one imposes additional normalization or balanced identification conditions on factors and loadings \citep{Bai2012STATISTICALAO}. 
    The network model removes this general linear indeterminacy, so long as $U_Q$ has full column rank. 
    
    Motivated by this identification structure, we estimate the shared latent factors by a likelihood-based joint criterion.
    For clarity, we first use a Gaussian working model, under which the negative log-likelihood reduces to the weighted
    least-squares criterion:
   \begin{align}\label{eq:loss}
        \argmin_{U,B} \ell(U,B):= & \lambda_1 \left\| A_Q - U_Q  U_Q^\top \right\|_{\fro}^2 + \lambda_2 \left\| W - U B^\top \right\|_{\fro}^2.
    \end{align}
    Because multiplying $(\lambda_1,\lambda_2)$ by a common positive constant does not change the optimizer, we use the normalization $\lambda_1+\lambda_2=1$ with $\lambda_1,\lambda_2\in[0,1]$ when comparing weights.
    Statements that require joint network-assisted identification are understood with $\lambda_1>0$.
This optimization problem can be equivalently represented as an augmented $(n_P + d)\times (n_P + d)$ system. Define the block matrices:
\begin{align*}
Y &=
\begin{pmatrix}
    A_Q &  \star & \multirow{2}{*}{$W$}\\
        \star & \star \\
        \multicolumn{2}{c}{W^{\top}} & \star \\
    \end{pmatrix},\quad
    Z =
    \begin{pmatrix}
       U_Q\\
        U_{Q^c}\\
         B
    \end{pmatrix},\quad
    {F} =
    \begin{pmatrix}
           D & \star & \multirow{2}{*}{$E$}\\
            \star & \star \\
            \multicolumn{2}{c}{E^{\top}}  & \star\\
        \end{pmatrix},
\end{align*}
where $D := A_Q - U_QU_Q^{\top}$ and $E := W - UB^{\top}$ represent residual matrices, and  the symbols ``$\star$'' indicate unspecified entries.
Then the augmented system satisfies the decomposition:
\begin{equation*}\label{eq:augmented_system}
Y = Z Z^\top + F .
\end{equation*}
 Introduce a binary mask matrix $\Lambda\in[0,1]^{(n_P+d)\times (n_P+d)}$ such that
\[
    \Lambda_{ij} =
    \begin{cases}
        \sqrt{\lambda_1}, & \text{if the $(i,j)$-entry of $Y$ corresponds to $A_Q$,}\\[4pt]
        \sqrt{\lambda_2/2}, & \text{if the $(i,j)$-entry of $Y$ corresponds to $W$ or $W^\top$,}\\[4pt]
        0, & \text{otherwise.}
    \end{cases}
\]
    Then the loss function in \eqref{eq:loss} is equivalent to \(\ell( X) = \|\Lambda\odot(Y-X) \|_{\fro}^2\), where $X:=ZZ^{\top}\in \SSS^{\,n_P+d}$ and $\odot$ denotes the Hadamard product.
    Moreover, $\ell$ is Lipschitz smooth and restricted strongly convex over the subspace $\mathcal{X} := \{X \in \SSS^{\,n_P+d} : \operatorname{supp}(X) \subseteq \operatorname{supp}(\Lambda)\}$.

    \begin{assumption}[Gaussian noise]\label{asm:noise}
        The noise entries $D_{ij} $ and $E_{ij} $ are i.i.d., mean-zero, Gaussian random variables with variances $\sigma_1^2$ and $\sigma_2^2$, respectively.
    \end{assumption}

    \begin{theorem}[Estimation error of $\hat X$]\label{thm:model}
        Let $\hat X\in\argmin_{X\in\SSS^{\,n_P+d},\rank(X)\leq k}\ell(X)$ and suppose the true signal
        $X\in\SSS^{\,n_P+d}$ has $\rank(X)\le k$.
        Then the following deterministic inequality holds:
        \[
        \|\Lambda\odot(\hat X-X)\|_{\fro}^2
        \ \le\
        48k\Bigl(\lambda_1\|D\|_{\oper}^2+\lambda_2\|E\|_{\oper}^2\Bigr).
        \]
        Moreover, under \Cref{asm:noise}, for any estimator $\tilde X(Y)$, up to universal constants,
        \begin{align}
            \inf_{\tilde X(Y)}\ \sup_{X\in \SSS^{\,n_P+d}:\,\rank(X)\le k}\
            \EE\Big[\ \|\Lambda\odot(\tilde X-X)\|_{\fro}^2\ \Big]
            \ \gtrsim\
            k\cdot \max\Bigl\{\lambda_1 \sigma_1^2 n_Q,\ \lambda_2 \sigma_2^2 (n_P+d)\Bigr\}.
            \label{eq:lb_thm2_rate}
        \end{align}
   \end{theorem}

Under the Gaussian noise assumption in Theorem~\ref{thm:model}, the minimax lower bound in \eqref{eq:lb_thm2_rate} implies that no estimator can achieve expected weighted error smaller than order \(k\cdot\max\{\lambda_1\sigma_1^2 n_Q,\lambda_2\sigma_2^2(n_P+d)\}\).
On the other hand, the deterministic upper bound yields
\(\|\Lambda\odot(\hat X-X)\|_{\fro}^2 \lesssim k(\lambda_1\|D\|_{\oper}^2+\lambda_2\|E\|_{\oper}^2)\), and standard random-matrix concentration gives \(\EE\|D\|_{\oper}^2\asymp \sigma_1^2 n_Q\) and \(\EE\|E\|_{\oper}^2\asymp \sigma_2^2(n_P+d)\).
Therefore, in expectation and with high probability, the upper bound matches the minimax lower bound up to universal constants (and at most a factor-$2$ gap from \(\max\) versus \(\lambda_1\sigma_1^2 n_Q+\lambda_2\sigma_2^2(n_P+d)\)). 
{Note that the Gaussian assumption is used mainly to obtain the minimax lower bound in \Cref{eq:lb_thm2_rate}, and the upper-bound analysis only requires operator-norm control of the noise matrices \(D\) and \(E\).
}

\Cref{thm:model} only relies on the zero-order optimality condition that the empirical minimizer has no larger objective value than the truth. 
Consequently, it controls only the objective-induced weighted error \(\lambda_1\|\hat U-U\|_{\fro}^2 + \lambda_2\|\hat B-B\|_{\fro}^2\), not the total or componentwise estimation errors. 
Since this weighted loss merges the network and embedding channels through \((\lambda_1,\lambda_2)\), it does not separate the channel contributions or reveal how the tuning weights affect each latent component.
This limitation motivates the gradient-dynamics analysis in the next section.

\section{PLANE: Estimation and Weighting Theory}
\subsection{Normalized gradient estimator}

We use a blockwise normalized gradient scheme to estimate the latent factors in the
joint objective.
Define the blockwise Gram matrices
\[
G_{U_Q}^{(t)} := 4\lambda_1 U_Q^{(t)\top} U_Q^{(t)}
+ 2\lambda_2 B^{(t)\top} B^{(t)},\qquad
G_{B}^{(t)} := 2\lambda_2 U^{(t)\top} U^{(t)},\qquad
G_{U_{Q^c}}^{(t)} := 2\lambda_2 B^{(t)\top} B^{(t)} .
\]
For each block $v\in\mathcal V := \{U_Q,B,U_{Q^c}\}$, let
$S_v^{(t)} := \{G_v^{(t)}\}^{1/2}$ denote the principal square root. 
Algorithm~\ref{alg:joint-UB-nGD} updates each block using the corresponding Gram inverse. 
Thus, unlike vanilla gradient descent, the step is rescaled according to the local curvature of the network--factor loss. 
This blockwise preconditioning reduces sensitivity to heterogeneous singular values and scale imbalance across latent directions, without requiring inversion of the full Hessian.

The statistical analysis must also account for the intrinsic rotational non-identifiability in \Cref{pro:identi}.
At iteration $t$, we align the iterates to the population orbit by the weighted Procrustes rotation $R^{(t)} \in
\argmin_{R\in\mathcal O(k)} \sum_{v \in\mathcal V}
\bigl\|
\bigl(v^{(t)}-v R\bigr) S_{v}^{(t)}
\bigr\|_{\fro}^2$.
Define the aligned deviations
$\Delta_{v}^{(t)}=v^{(t)}-vR^{(t)},{v\in\mathcal V} $.
We measure the algorithmic error in the same local geometry used by the
updates.
Specifically, let
\begin{align}
    e_t := \sum_{v\in\mathcal V} s_v^{(t)},\qquad
    s_v^{(t)} := \|\Delta_v^{(t)}S_v^{(t)}\|_{\fro}^2 ,
    \quad v\in\mathcal V . \label{eq:et}
\end{align}
This curvature-weighted error is the quantity contracted in the sequel; it couples the target latent positions, the embedding loadings, and the embedding-only latent positions in a single recursion.

\begin{algorithm}[!t]
    \caption{PLANE: Proxy-Latent Assisted Network Estimation via  normalized gradient descent}
    \label{alg:joint-UB-nGD}
    \begin{algorithmic}[1]
    \Require Adjacency block $A_Q\in\mathbb{R}^{n_Q\times n_Q}$, covariates $W\in\mathbb{R}^{n_P\times d}$, latent dimension $k$, normalized tuning weights $(\lambda_1,\lambda_2)$;
    initial guesses $U^{(0)}\in\mathbb{R}^{n_P\times k}$, $B^{(0)}\in\mathbb{R}^{ d\times k}$; step sizes $\eta$.
    \Ensure Estimators of $ (U,  B)$.
    \For{$t=0,1,\dots,T-1$}
        \State Split $(U_Q^{(t)}; U_{Q^c}^{(t)}) \gets U^{(t)}$.
        \State Update $U_Q$: \( U_Q^{(t+1)} \gets U_Q^{(t)} - \eta \, \nabla_{U_Q}\ell(U^{(t)}, B^{(t)})\cdot(G_{{U_Q}}^{(t)})^{-1}\).
        \State Update $B$: \( B^{(t+1)} \gets B^{(t)} - \eta \, \nabla_{B}\ell (U^{(t)}, B^{(t)})\cdot(G_{B}^{(t)})^{-1}\).
        \State Update $U_{Q^c}$: \( U_{Q^c}^{(t+1)} \gets U_{Q^c}^{(t)} - \eta \,  \nabla_{U_{Q^c}}\ell(U^{(t)}, B^{(t)}) \cdot (G_{U_{Q^c}}^{(t)})^{-1}\).

        \State Reassemble $U^{(t+1)} \gets \big(U_Q^{(t+1)};\, U_{Q^c}^{(t+1)}\big)$.
    \EndFor
    \State \Return $(U^{(T)}, B^{(T)})$.
    \end{algorithmic}
\end{algorithm}

\subsection{Iterative network-factor weighted errors}
We analyze Algorithm~\ref{alg:joint-UB-nGD} through the curvature-weighted error $e_t$ in \eqref{eq:et}, which jointly tracks the aligned errors of $(U_Q,B,U_{Q^c})$. 
We first state the regularity conditions for the local gradient-dynamics analysis, with full details given in \Cref{sec:asmp}.

\begin{theorem}[Deterministic bounds for iterative errors]\label{thm:error_algor}
    Suppose that $\lambda_1+\lambda_2=1$ and $\lambda_1\ge \underline\lambda$ for some universal constant $\underline\lambda \in(0,1)$.
    Consider the iterative error $\{ e_t\}_{t\ge 0}$ of \Cref{alg:joint-UB-nGD}. 
    Under Assumptions~\ref{asp:row_rank}--\ref{asp:local-well-conditioned}, if the step-size constant $\eta$ is sufficiently small, then there exist constants $\rho\in(0,1)$ and $C>0$ (depending only on the spectrum of $U_Q$, $B$, and $U_{Q^c}$, but not on $\lambda_1$ or $\lambda_2$), such that
    \begin{align*}
    e_{t+1}
        \le
        \left(1-\eta\rho\right)^t\,e_0
        \;+\;
        C\Delta_S,\quad\text{for all }t\ge 0,
    \end{align*}
    where the statistical error term $\Delta_S$ is given by
\begin{align*}
    \Delta_S &= k\left(\frac{\lambda_1^2 \|U_Q\|_{\oper}^2\|D\|_{\oper}^2+\lambda_2^2\|
E_Q\|_{\oper}^2\|B\|_{\oper}^2 }{2\lambda_1\sigma_{\min}^2(U_Q)+\lambda_2\sigma_{\min}^2(B)}  +\frac{\lambda_2\|E_{Q^c}\|_{\oper}^2\| B\|_{\oper}^2 }{\sigma_{\min}^2(B)}+ \frac{\lambda_2\|E\|_{\oper}^2\| U\|_{\oper}^2 }{\sigma_{\min}^2(U)}\right).
\end{align*}
\end{theorem}

\Cref{thm:error_algor} shows that the weighted error can be bounded by the optimization and statistical errors, where the optimization error contracts geometrically at a rate of $(1-\eta\rho)^t$ and the statistical error captures the three statistical error terms arising from the inner products between weighted error and gradient for \(U_Q\), \(U_{Q^c}\), and \(B\), respectively.

\subsection{Target-network error and optimal weighting}
\Cref{thm:error_algor} controls the aggregate curvature-weighted error over $(U_Q,B,U_{Q^c})$, whereas our main statistical target is the latent position matrix $U_Q$. The next corollary characterizes how \((\lambda_1,\lambda_2)\) balances network and embedding information in the estimation of \(U_Q\).

\begin{corollary}[Weighted network--embedding trade-off]
\label{cor:weighted-error-UQ}
    Suppose that the assumptions of \Cref{thm:error_algor} hold and \(\sup_{0\le t}s_B^{(t)}/s_{U_Q}^{(t)}=o(1)\) as $n_P\to\infty$.
   Consider the estimator ($\hat U_Q$, $\hat B$, $\hat{U}_{Q^c}$) returned by \Cref{alg:joint-UB-nGD} when $T\to\infty$.
   There exists a Procrustes alignment \(R\in\mathcal O(k)\) such that
    \begin{align}
        \frac{1}{k}
    \bigl\|\widehat U_Q-U_QR\bigr\|_{\fro}^2
    \;\lesssim\;
    \frac{
    \lambda_1^2\|U_Q\|_{\oper}^2\|D\|_{\oper}^2
    +
    \lambda_2^2\|B\|_{\oper}^2\|E_Q\|_{\oper}^2
    }{
    \bigl(
    \lambda_1\sigma_{\min}^2(U_Q)
    +
    \lambda_2\sigma_{\min}^2(B)
    \bigr)^2
    }:=\cE(\lambda_1,\lambda_2).\label{eq:Eps}
    \end{align}
\end{corollary}
\Cref{cor:weighted-error-UQ} shows that the error for $\widehat U_Q$ has two sources: the target-network noise $\|D\|_{\oper}$ and the proxy-embedding noise $\|E_Q\|_{\oper}$. 
In the network-only case, $(\lambda_1,\lambda_2)=(1,0)$, the bound reduces to the usual latent-network error form, comparable to Theorem~9 of \cite{ma2020universal}. 
When $Q^c=\varnothing$, the model also falls within the joint network--feature setting of \cite{zhang2022joint}. 
The main difference is that \Cref{cor:weighted-error-UQ} keeps the dependence on $(\lambda_1,\lambda_2)$ explicit, showing how the target-block error changes as weight shifts between the network and embedding channels.

To further interpret the weighting trade-off, we consider the partial-oracle update in \Cref{prop:tightness_partial_oracle_short}. 
It can be viewed as an idealized local-gradient step for $U$, where the complementary factors are evaluated at their population values: the right network factor is fixed at $U$, and the embedding loading is fixed at $B$. 
This removes nuisance estimation error while keeping the same network and embedding noise sources. 
The resulting objective in \Cref{prop:tightness_partial_oracle_short} isolates the role of $(\lambda_1,\lambda_2)$ in balancing signal strength and noise for estimating $U_Q$, and thus the dependence in \eqref{eq:Eps} is sharp up to constants.

The trade-off is \emph{scale-invariant}, since
\(\cE(t\lambda_1,t\lambda_2)=\cE(\lambda_1,\lambda_2)\) for any \(t>0\). Thus only the ratio of the two weights matters. Writing
\begin{align*}
    \lambda := \frac{\lambda_1}{\lambda_1+\lambda_2}\in[0,1],
    \qquad
    \cE(\lambda)\equiv \cE(\lambda,1-\lambda)
    =
    \frac{\lambda^2\|U_Q\|_{\oper}^2\|D\|_{\oper}^2+(1-\lambda)^2\|B\|_{\oper}^2\|E_Q\|_{\oper}^2}
    {\bigl(\lambda\sigma_{\min}^2(U_Q)+(1-\lambda)\sigma_{\min}^2(B)\bigr)^2}.\label{eq:cE-lambda}
\end{align*}
In particular, by \Cref{lem:opt}, the error function $\cE(\lambda)$ attains its unique global minimum at
\begin{align*}
\lambda^*
=
\frac{
\|B\|_{\oper}^2\|E_Q\|_{\oper}^2\,\sigma_{\min}^2(U_Q)
}{
\|U_Q\|_{\oper}^2\|D\|_{\oper}^2\,\sigma_{\min}^2(B)
+
\|B\|_{\oper}^2\|E_Q\|_{\oper}^2\,\sigma_{\min}^2(U_Q)
}.
\end{align*}
This unconstrained optimizer is admissible for \Cref{thm:error_algor} only when
\(
\lambda^*\ge \underline\lambda,
\)
or equivalently,
\(
{\|E_Q\|_{\oper}^2}/{\|D\|_{\oper}^2}
\ge
{\underline\lambda}/({1-\underline\lambda})
\cdot
{\kappa_{U_Q}^2}/{\kappa_B^2}.
\)
Thus, the harmonic-mean identity below applies directly in the regime where the unconstrained optimum satisfies the lower-bound condition:
\[
\cE(\lambda^*)
=
\frac{\cE(1)\cE(0)}{\cE(1)+\cE(0)}
=
\bigl(\cE(1)^{-1}+\cE(0)^{-1}\bigr)^{-1}
\le \min\{\cE(1),\cE(0)\}.
\]
If \(\lambda^*<\underline\lambda\), the unconstrained optimum lies outside the
region covered by the contraction theory. In that case, the relevant admissible
choice is the projected weight
\(
\lambda^\dagger
:=
\argmin_{\lambda\in[\underline\lambda,1]} \cE(\lambda)
=
\max\{\lambda^*,\underline\lambda\}.
\)
Since \(1\in[\underline\lambda,1]\), this admissible choice still satisfies
\(
\cE(\lambda^\dagger)\le \cE(1),
\)
so the theory guarantees improvement over, or at least competitiveness with, the network-only bound. 
The sharp harmonic-mean comparison should be interpreted as applying to the unconstrained oracle weighting, or to the admissible regime \(\lambda^*\ge\underline\lambda\).

\section{Simulation Studies}
In this section, we perform simulations to evaluate two questions suggested by the theory. 
First, we examine whether PLANE selects larger embedding weights when the proxy channel is more informative and smaller weights when the target network is more informative. 
Second, we study whether embeddings on the additional nodes \(Q^c\) improve recovery of \(U_Q\), and whether held-out validation downweights uninformative proxies. {Details on the data-generating mechanisms and additional simulation experiments
are provided in \Cref{app:sec:simulation}, including robustness to model
misspecification (\Cref{fig:simulation-missspecification}), baseline comparisons
(\Cref{fig:simulation-baseline}), optimization performance
(\Cref{fig:simulation-1-convergence}), downstream network analysis
(\Cref{fig:simulation-study-downstream}), sample-size scaling
(\Cref{fig:appendix-simulation-sample-size}), and robustness to binary target
networks (\Cref{fig:appendix-simulation-binary}).}

Our main experiments use the Gaussian model
$ A_Q = U_Q U_Q^\top + \sigma_A D,\;
    W = U B^\top + \sigma_W E$,
with rank \(k=3\) and independent Gaussian noise. We tune the relative weight by setting \(\lambda_1=1-\lambda_2\) over a fixed grid of \(\lambda_2\) values. PLANE-oracle selects \(\lambda_2\) using the true latent error of \(\widehat U_Q\), so it serves only as a benchmark. PLANE-CV selects \(\lambda_2\) using held-out target-network entries, which is the practical tuning rule.

The first study varies the relative signal-to-noise ratios of the network and embedding channels and tracks the normalized gradient iterates.
\Cref{fig:simulation-study-embedding-nodes} A--C show that the oracle choice of $\lambda_2$ moves from $0.03$ when the proxy is weak ($\sigma_A\ll\sigma_W$), to $0.60$ in the balanced case, and to $0.80$ when the proxy is strong ($\sigma_A\gg\sigma_W$). 
Thus, the selected embedding weight increases when the proxy becomes more useful and decreases when the network is more reliable, matching the trade-off in \Cref{cor:weighted-error-UQ}.

The second study considers the setting where embeddings are observed on both \(Q\) and \(Q^c\), while the target network is observed only on \(Q\).
When \(|Q^c|=0\), there are no embedding-only nodes, and the setup reduces to the same-node network--feature setting of \cite{zhang2022joint}.
As shown in \Cref{fig:simulation-study-embedding-nodes} D--E, as \(|Q^c|\) increases, informative embedding-only nodes improve PLANE-CV.
The gain then levels off, consistent with cross-validation, lowering the embedding weight when extra proxy information no longer improves held-out target-network prediction. 
Null proxies provide little gain and remain close to the network-only baseline.
\Cref{fig:simulation-study-embedding-nodes} F shows that cross-validation assigns moderate weights to informative proxies but near-zero weights to null proxies.
Thus, informative \(Q^c\) nodes can improve recovery of \(U_Q\), while PLANE-CV adapts to different embedding setups.

\begin{figure}[!ht]
    \centering
    \includegraphics[width=1\textwidth]{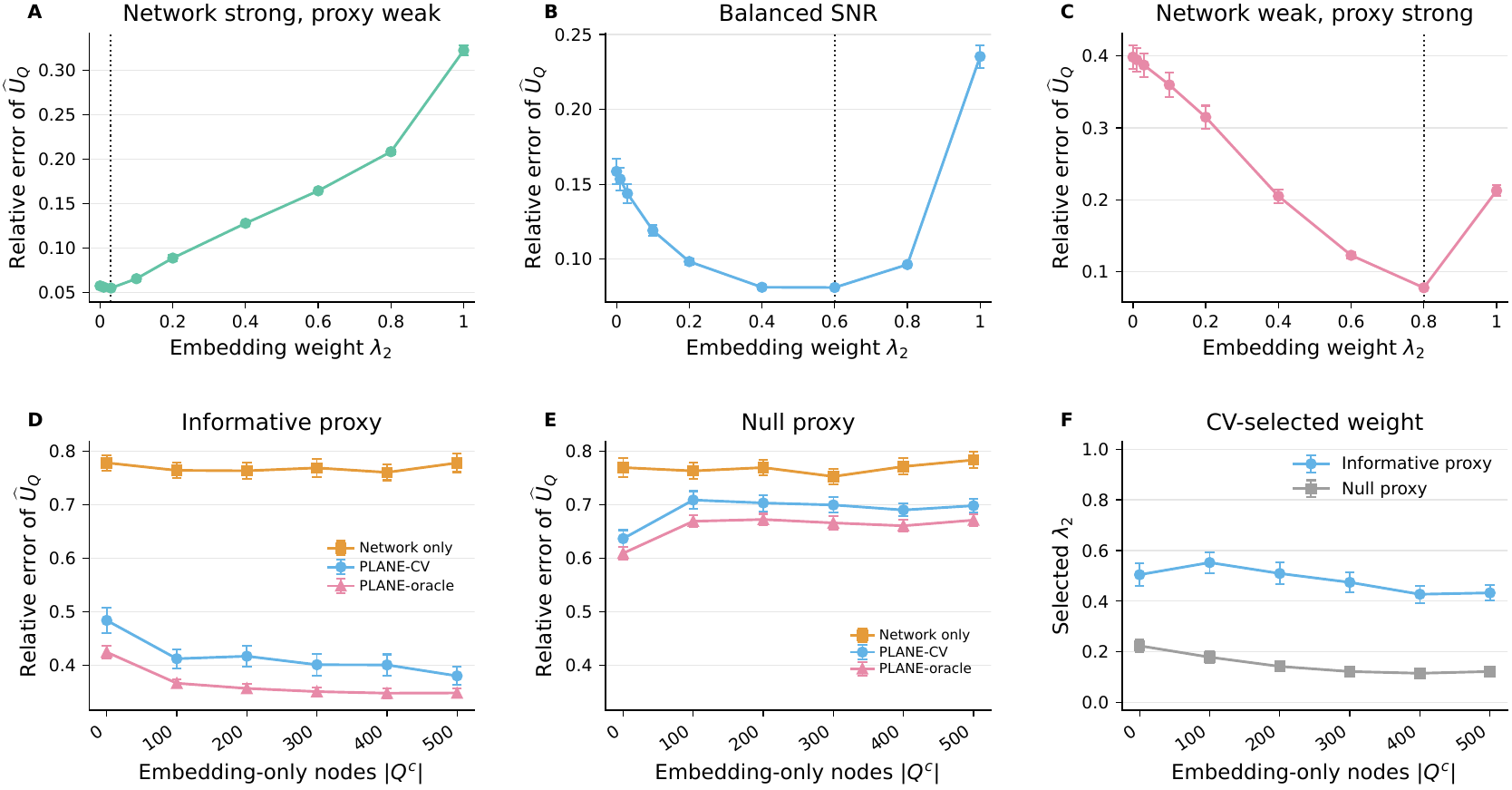}
    \caption{{
    Panels A--C show the mean relative Procrustes error of \(\widehat U_Q\) across weights \(\lambda_2\), with dotted lines indicating oracle choices in study 1.
    The endpoints correspond to two baseline methods: network-only model \((\lambda_2=0)\) and an embedding-only model \((\lambda_2=1)\).
    In study 2, Panels D--E compare informative and null proxy embeddings as \(|Q^c|\) increases. 
    Panel F reports the mean cross-validated embedding weight.
    When \(Q^c=\emptyset\), it reduces to the baseline of the network-factor weighted method studied by \cite{zhang2022joint} without extra nodes.
    }
   }
    \label{fig:simulation-study-embedding-nodes}
\end{figure}

\section{Gene Network Inference with External Gene Embeddings}\label{sec:application}

Single-cell CRISPR screens perturb selected genes across individual cells and measure the resulting transcriptomic responses, creating perturbation-response data for studying downstream regulatory programs \cite{du2025causal,du2026uncovering}. 
Here we use these data for a network-level task: recovering a gene--gene similarity structure in which two genes are close if their perturbations induce similar expression signatures.
We apply PLANE to the CRISPR activation (CRISPRa) Perturb-seq experiment of \citet{norman2019exploring}. 
CRISPRa uses guide RNAs to increase, rather than disrupt, the expression of selected target genes, and Perturb-seq records the resulting single-cell transcriptomic responses. 
This is a natural application for the proposed model: the target network is noisy and available only on a selected perturbation set \(Q\), whereas external gene embeddings are available on a larger gene universe \(P\supseteq Q\). 
{Details on the real-data construction and additional application analyses are
provided in \Cref{app:sec:application}, including eigengap diagnostics for the
concatenated proxy embedding (\Cref{fig:eigengap-total}), source-wise eigengap
diagnostics (\Cref{fig:eigengap}), sensitivity to the weighted target-network
construction (\Cref{fig:appendix-application-weighted-network}), sensitivity to the ordering of embedding-only genes (\Cref{fig:appendix-application-ordering}), optimization and selection diagnostics (\Cref{fig:appendix-application-diagnostics}), supplementary fitted-network summaries (\Cref{fig:appendix-application-network}), and GO enrichment analysis (\Cref{fig:appendix-application-go}).}

We construct a binary target graph \(A_Q\) (\Cref{fig:application-biology} (A)) from signed differential-expression (DE) signatures.
The proxy matrix \(W\) is built from the embedding collection of PRESAGE \citep{littman2025presage}, concatenating GO biological-process pathway and transcription-factor-target in MSigDB \cite{subramanian2005gene,liberzon2015molecular}, STRING \cite{Szklarczyk2025STRING}, GenePT text \cite{chen2025simple}, ESM protein-language \citep{lin2023esm2}, and GenePT protein embeddings \cite{chen2025simple} after column standardization and rescaling individually.
After target-network and embedding alignment, the analysis contains \(n_Q=99\) target perturbation genes, \(n_P=1099\) genes with embeddings of dimension \(d=10112\).
We fit the PLANE estimator over a grid of $\lambda_2$ and set $\lambda_1=1-\lambda_2$.

\begin{figure}[!ht]
    \centering
    \includegraphics[width=0.9\textwidth]{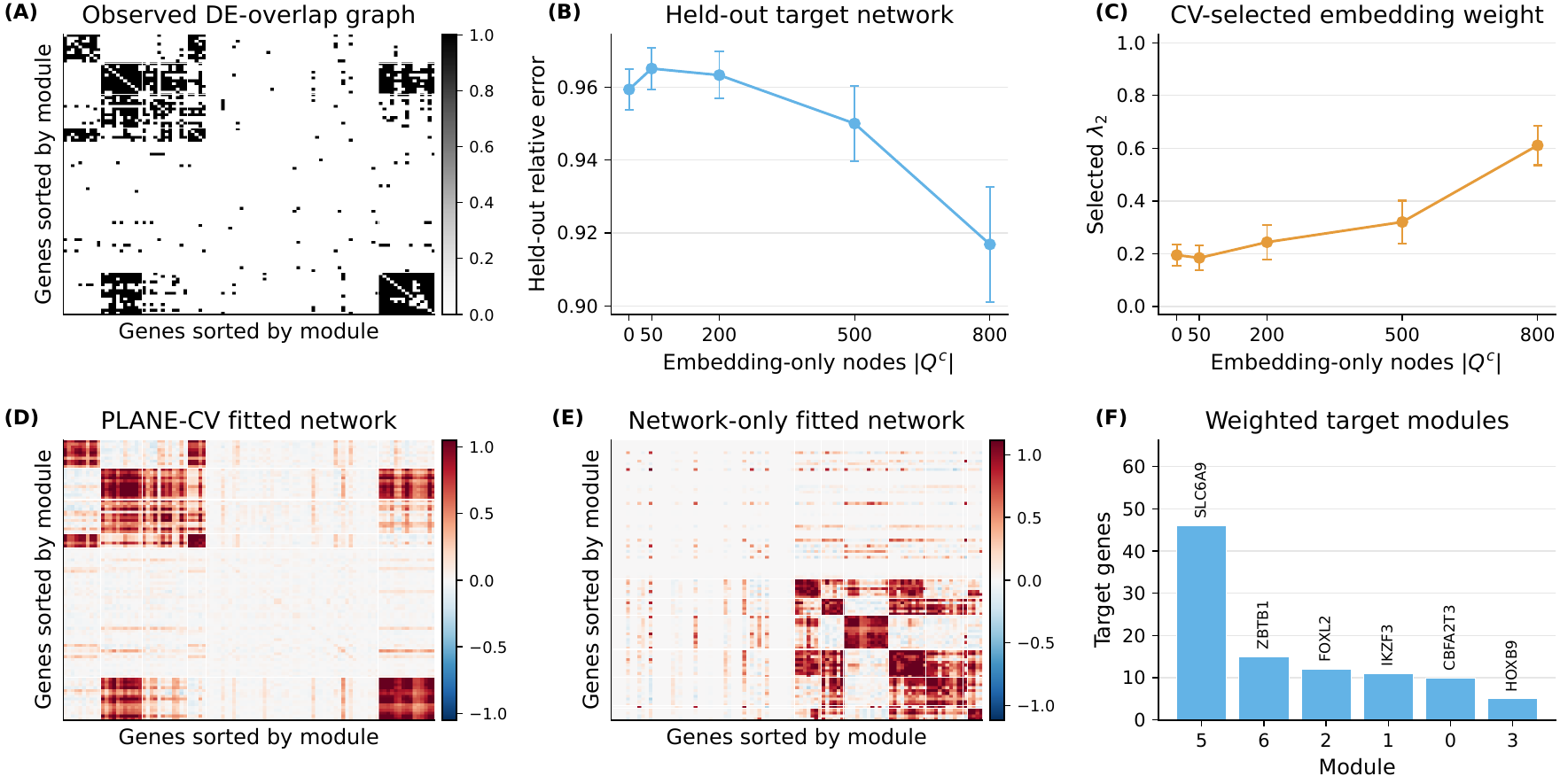}
    \caption{Results on Norman CRISPR data analysis.
    (A) The DE-overlap graph among target genes.
    (B) The held-out relative error on the Norman target graph.
    (C) The selected \(\lambda_2\) values across repetitions.
    (D) The estimated target network by PLANE-CV.
    (E) The estimated target network by network-only PLANE ($\lambda_1=1$).
    (F) The module sizes and top hubs of PLANE-CV's target network.}\label{fig:application-biology}
\end{figure}

{
\Cref{fig:application-biology} (B)--(C) show the effect of external embedding size. 
The candidate embedding-only genes were ranked by their maximum cosine similarity to the target-gene embeddings in the standardized proxy-embedding space, and \(Q^c\) was formed by taking the top \(m\) ranked genes for each \(m\in\{0,50,200,500,800\}\). For each \(m\), \(\lambda_2\) was selected by held-out target-network MSE over 20 random validation splits, and we report the mean and standard error across splits.}
Across held-out target-node splits, the selected PLANE fit reduced held-out relative error compared with the network-only fit, with the selected weight concentrated at \(\lambda_2=0.6\).
We then refit the selected weighted model ($\lambda_2=0.6$) using all target-network entries and apply \(K\)-means with \(K=8\) to the estimated latent positions from
PLANE-CV and the network-only baseline; the corresponding module-ordered
heatmaps are shown in \Cref{fig:application-biology} (D)--(E). We clustered the fitted latent positions on the full proxy-embedding gene set into \(8\) modules; target-gene visualizations use the induced ordering restricted to target genes, so modules containing only embedding-only genes do not appear in the target-gene heatmaps. The PLANE-CV fit yields a refined target network that retains the main block patterns of the observed graph. By contrast, the network-only fit places a concentrated bottom-right block, and appears less consistent with the observed graph.

\Cref{fig:application-biology} (F) shows that the fitted graph yields interpretable target modules. The largest module contains 46 genes and is centered around \textit{SLC6A9}, which encodes glycine transporter 1, a regulator of extracellular glycine in the central nervous system \citep{marques2020neurobiology,ncbiSLC6A9}. 
Since glycine participates in glycinergic and NMDA-receptor signaling \citep{harvey2008critical,cavalcante2026glyt1,li2025modulation}, this module points to perturbation-response signals related to neurotransmitter transport. 
Other modules are represented by hubs, including \textit{ZBTB1}, \textit{FOXL2}, \textit{IKZF3}, \textit{CBFA2T3}, and \textit{HOXB9}. These genes point to several complementary regulatory programs: \textit{ZBTB1} and \textit{IKZF3} are linked to immune-cell development and lymphocyte regulation \citep{siggs2012zbtb1,punwani2012zbtb1,shi2024ikzf3}, \textit{CBFA2T3} is involved in hematopoietic and myeloid transcriptional control \citep{
steinauer2020cbfa2t3}, and \textit{FOXL2} and \textit{HOXB9} represent broader developmental transcription-factor programs
\citep{caburet2012foxl2,shrestha2012hoxb9}. 
These annotations suggest that the smaller modules are enriched for immune, hematopoietic, and developmental regulatory programs, supporting biological interpretability of the fitted modules. GO and regulatory-overlap analyses in
\Cref{fig:appendix-application-go} further support the biological
interpretability of the fitted modules.

{
To evaluate the contribution of each proxy-embedding source, we conducted a leave-one-source-out ablation study. 
For each ablated embedding, we refit PLANE using the same target network, rank, optimization settings, and held-out validation procedure. 
The embedding weight \(\lambda_2\) was selected by the mean held-out target-network MSE over 10 validation repetitions, excluding the embedding-only endpoint \(\lambda_2=1\). 
As shown in \Cref{fig:appendix-application-source-sig}, the full proxy embedding
improves target-network estimation relative to the network-only baseline across
all source-removal settings. 
In particular, the held-out MSE curves stay below or close to the network-only fit over a range of intermediate \(\lambda_2\) values, indicating that the benefit of proxy embeddings is not driven by any single source alone. 
Among the ablations, GenePT text embedding contributes the most distinctive information, while ESM protein embedding is the least critical source in this analysis. 
The fitted networks in Panel B further preserve a similar block-structured pattern under the common gene ordering, showing that the main recovered modules are stable to the removal of individual
sources. Overall, PLANE benefits from integrating complementary proxy
embeddings rather than relying on one isolated source.

\begin{figure}[!ht]
    \centering
    \includegraphics[width=1\textwidth]{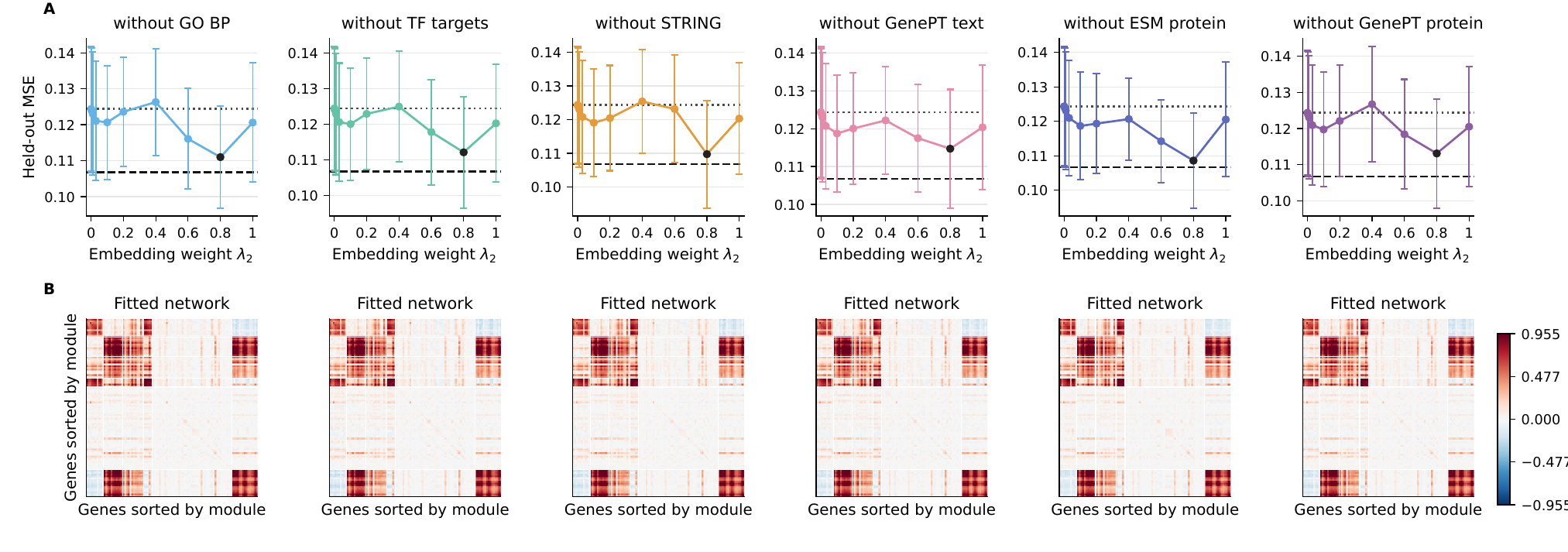}
\caption{{Leave-one-source-out ablation of the real-data proxy embedding.
Each column removes one embedding source from the full concatenated proxy embedding.
(A) Held-out target-network MSE across the \(\lambda_2\) grid; the dotted and dashed horizontal lines denote the network-only and full-embedding references, respectively.
(B) Fitted target-network heatmaps from the selected PLANE-CV fits, with genes ordered by the final target-module labels.}}
    \label{fig:appendix-application-source-sig}
\end{figure}}

\section{Discussion}\label{sec:disc}
This paper studies how externally learned gene embeddings can help estimate a target gene network that is observed only on a smaller set of genes. The proposed PLANE estimator treats the network and the embeddings as two noisy views of the same latent biological structure, and learns how much weight to give each source. Our theory shows when this joint use of information improves latent-position recovery, while the simulations and CRISPR perturbation analysis show that informative embeddings can improve held-out network prediction, module discovery, and link imputation beyond the observed target graph. 
The method is most useful when the embeddings carry biological signals related to the target network; when they are weak or unrelated, cross-validation downweights them. {A natural direction for future work is to introduce link functions $h_A,h_W$  between the low-rank predictors and the observed responses,  in the spirit
of generalized low-rank network and joint network-covariate models
\cite{ma2020universal,zhang2022joint}, e.g.,
\(h_A(\mathbb{E}(A_{Q,ij}\mid U_Q))=(U_QU_Q^\top)_{ij}\) and
\(h_W(\mathbb{E}(W_{ij}\mid U,B))=(UB^\top)_{ij}\). 
Such generalized PLANE models could better accommodate binary, signed, thresholded, or heterogeneous biological observations while preserving the same latent-factor weighting principle.}

\bibliographystyle{apalike}
\bibliography{main}

\clearpage

\appendix
\counterwithin{theorem}{section}
\numberwithin{equation}{section}

\renewcommand{\thesection}{\Alph{section}}
\setcounter{table}{0}
\renewcommand{\thetable}{\thesection\arabic{table}}
\renewcommand{\theHtable}{appendix.\thesection.\arabic{table}}
\setcounter{figure}{0}
\renewcommand\thefigure{\thesection\arabic{figure}}
\renewcommand{\theHfigure}{appendix.\thesection.\arabic{figure}}
\setcounter{algorithm}{0}
\renewcommand{\thealgorithm}{\thesection.\arabic{algorithm}}
\renewcommand{\theHalgorithm}{\thesection.\arabic{algorithm}}

\begin{center}
\Large
{\bf Appendix}
\end{center}

This serves as an appendix to the main paper.
Below, we provide an outline for the appendix along with a summary of the notation used in the main paper and the appendix.

\paragraph{Organization.}
The content of the appendix is organized as follows.

\begin{table}[!ht]
\centering \small
\begin{tabularx}{0.98\textwidth}{l l D}
    \toprule
    \multicolumn{2}{c}{\textbf{Appendix}} & \textbf{Content} \\
    \midrule \addlinespace[0.5ex]
    \multirow{3}{*}{\Cref{app:sec:proof}} & \ref{app:subsec:ident} & Proof of \Cref{pro:identi} for identification condition. \\
    & \ref{app:subsec:opt} & Proof of \Cref{thm:model} for optimality condition.
    \\ & \ref{subsec:discussion} & Special cases of \Cref{thm:model}.\\\addlinespace[0.5ex] \cmidrule(l){1-3}\addlinespace[0.5ex]
    \multirow{6}{*}{\Cref{app:sec:proof-algo}} & \ref{app:subsec:pre} & Preliminary analysis for the main technical proof related to \Cref{alg:joint-UB-nGD}. \\
    & \ref{sec:asmp} & Algorithmic Assumptions.\\
     & \ref{app:subsec:err-weighted} & Proof of \Cref{thm:error_algor} for weighted estimation error. \\     
           & \ref{app:subsec:tradeoff} &  Proof of \Cref{cor:weighted-error-UQ}. \\    
           & \ref{subsec:dis_error} &  {Discussion of \Cref{cor:weighted-error-UQ} and error bound for \(Q^c\)} \\  
      & \ref{subsec:Poa} &  Partial-oracle analysis of the trade-off. \\
     \addlinespace[0.5ex] \cmidrule(l){1-3}\addlinespace[0.5ex]
    \multirow{4}{*}{\Cref{app:sec:lemma}} & \ref{app:subsec:grad} & Deterministic gradient lemmas: \Cref{lem:inner-product-eq,lem:gradient,lem:DeltaS,lem:grad-DeltaS,lem:cross-term-blockwise}. \\
    & \ref{app:subsec:recusion} & Error recursion lemmas: one-step update for the weighted aggregate error (\Cref{lem:iter_all}) and componentwise errors (\Cref{lem:norm}).
    \\
     & \ref{app:subsec:ineq} & Block matrix inequality (\Cref{lm:rank}).  
    \\\addlinespace[0.5ex] \cmidrule(l){1-3}\addlinespace[0.5ex]
    \Cref{app:sec:simulation} & -- & Simulation design, tuning rules, and supplementary simulation results. \\\addlinespace[0.5ex] \cmidrule(l){1-3}\addlinespace[0.5ex]
    \Cref{app:sec:application} & -- & Real-data application preprocessing, validation design, optimization diagnostics, and supplementary fitted-network displays. \\
    \bottomrule
\end{tabularx}
\end{table}
\addcontentsline{toc}{part}{\appendixname}

\paragraph{Notation}
An overview of some general notation used in the main paper and the appendix is as follows.

For a matrix \( A \in \mathbb{R}^{m \times n} \), we denote by \( \|A\|_{\fro} = (\sum_{i,j} A_{ij}^2)^{1/2} \) the \emph{Frobenius norm}, by \( \|A\|_{\oper} = \sup_{\|x\|_2 = 1} \|A x\|_2 \) the \emph{operator (spectral) norm}, and by \( \|A\|_{*} = \sum_i \sigma_i(A) \) the \emph{nuclear norm},  where \( \sigma_i(A) \) denotes the \( i \)-th singular value of \( A \).
The \emph{trace} of \( A \) is written as \( \tr(A) = \sum_i A_{ii} \).
We denote the \emph{smallest singular value} of \( A \) by \( \sigma_{\min}(A) \), and the \emph{condition number} of a full-rank matrix \( A \) by \( \kappa_A = \|A\|_{\oper} / \sigma_{\min}(A) \).
The set of symmetric matrices in $\RR^{n\times n}$ is denoted by $\SSS^n$.

For any two matrices $A$ and $B$, the matrix Hadamard product (element-wise product) and face-splitting product (row-by-row Kronecker product) are denoted by $A \odot B$ and $A \bullet B$, respectively, and the Frobenius inner product is denoted by $\langle A, B\rangle = \tr(A^\top B) $.
Matrix Hadamard power is denoted by $A^{\odot c}$ for any constant $c$.

We use ``$o$'' and ``$\cO$'' to denote the little-o and big-O notations; ``$\op$'' and ``$\Op$'' are their probabilistic counterparts. For sequences $\{a_n\}$ and $\{b_n\}$, we write $a_n\lesssim b_n$ if $a_n=\cO(b_n)$; and $a_n\asymp b_n$ if $a_n=\cO(b_n)$ and $b_n=\cO(a_n)$.
Convergence in distribution is denoted by ``$\dto$''. In what follows, we use $\BC$ with subscripts to denote generic positive constants. And we use $C$ to denote a generic constant that may vary from line to line. For any $Z_1, Z_2 \in \mathbb{R}^{n \times k}$, define
\[
    \dist(Z_1, Z_2) \coloneqq \min_{R \in \cO_k} \| Z_1 - Z_2 R \|_{\fro},
\]
where $\cO_k \coloneqq \{ R \in \mathbb{R}^{k \times k} : R^\top R = I_k \}$ denotes the set of $k \times k$ orthogonal matrices.

\clearpage

\section{Proofs for identification and optimality conditions}\label{app:sec:proof}
\subsection[Proof of identification proposition]{Proof of \Cref{pro:identi}}\label{app:subsec:ident}
\begin{proof}[Proof of \Cref{pro:identi}]
    By assumption (i), the first $n_Q$ rows of $U$ have full column rank $k$. Hence the full matrix $U\in\mathbb R^{n_P\times k}$ has at least rank $k$, that is, $\operatorname{rank}(U)=k$.

    Suppose two parameter pairs $(U,B)$ and $(U^\prime,B^\prime)$ generate the same distribution of the observables. In the conditional-expectation setting, this implies that
    \[
      \EE[W\mid U]= U B^\top = U^\prime {B^\prime}^\top.
    \]
    Because $\operatorname{rank}(U)=\operatorname{rank}(U^\prime)=k$ and $B,B^\prime$ have full column rank $k$, the column spaces of $U$ and $U^\prime$ coincide. Hence there exists a $k\times k$ invertible matrix $O$ such that
    \[
        U^\prime = U O.
    \]
    Substituting into $UB^\top = U^\prime {B^\prime}^\top$ gives $UB^\top = U O {B^\prime}^\top$, and since $U$ has full column rank we obtain the relation
    \[
    B = B^\prime O^{-\top} .
    \]
    Thus, the factorization $UB^\top$ is unique up to the invertible linear transform $O$.

    From the model, we also observe the adjacency-type matrix on the target nodes,
    \[
    \EE[A_Q \mid U_Q] = U_Q U_Q^\top = U^\prime_Q U^\prime_Q{}^\top.
    \]
    Using $U^\prime = U O$, we have $U^\prime_Q = U_Q O$, so
    \[
    U_Q U_Q^\top = U_Q O O^\top U_Q^\top.
    \]
    Left- and right-multiplying the above identity by $U_Q^\top$ and $U_Q$ respectively, yields
    \[
    U_Q^\top U_Q U_Q^\top U_Q = U_Q^\top U_Q  O O^\top U_Q^\top U_Q.
    \]
    Since $U_Q$ has full column rank, the $k\times k$ matrix $U_Q^\top U_Q$ is positive definite and therefore invertible. Multiplying on the left and right side by $(U_Q^\top U_Q)^{-1}$ gives
    \[
    I_k = O O^\top,
    \]
    so $O O^\top = I_k$. Consequently, $O$ is an orthonormal matrix.
    Thus, we find an orthonormal matrix $O$ such that
    \[
    U^\prime = U O,\qquad B^\prime = B O,
    \]
    Hence $(U,B)$ is identifiable up to a $k\times k$ orthonormal transformation.
\end{proof}
\subsection[Proof of weighted objective theorem]{Proof of \Cref{thm:model}}\label{app:subsec:opt}
\begin{proof}[Proof of \Cref{thm:model}]
    We split the proof into two parts.

    \paragraph{Part 1: upper bound.}
    Recall $\hat X$ is a solution to the optimization problem \eqref{eq:loss} and $X$ is the truth.
    Define
    \[
    \hat{Z} = \begin{pmatrix}
                \widehat{U}_Q\\
               \widehat{U}_{Q^c}\\
                \widehat{B}
            \end{pmatrix}, \quad
    Z = \begin{pmatrix}
              {U}_Q\\
              {U}_{Q^c}\\
                {B}
            \end{pmatrix},
    \]
    such that $\widehat{X} = \hat{Z}\hat{Z}^\top$ and ${X} = ZZ^\top$.
    Because $\hat{Z}$ minimizes \eqref{eq:loss}, by the optimality condition, we have:
    \begin{align}
    \ell(\widehat{X}) - \ell(X) \leq 0.\label{eq:opt-cond}
    \end{align}
    Applying a second-order Taylor expansion of $f$ restricted to the subspace $\mathcal{X}$,
    \begin{align}
         \ell(\hat{X})-\ell(X) & = \|\Lambda\odot(Y-\hat{X}) \|_{\fro}^2-\|\Lambda\odot(Y-X) \|_{\fro}^2 \notag\\
         & =-2\langle \Lambda\odot(Y-X), \Lambda\odot(\hat{X}-X) \rangle+\|\Lambda\odot(\hat{X}-X)\|_{\fro}^2 \label{eq:first-order}.
    \end{align}
    Combining \eqref{eq:opt-cond} with \eqref{eq:first-order} yields that
    \begin{align*}
        \|\Lambda\odot(\hat{X}-X)\|_{\fro}^{2} &\le 2\langle \Lambda\odot(Y-X),\, \Lambda\odot(\hat{X}-X) \rangle\\
        & \leq 2 \sqrt{\rank(\Lambda\odot(\hat{X}-X))}\|\Lambda\odot(Y-X)\|_{\oper}
       \,\|\Lambda\odot(\hat{X}-X)\|_{\fro},
    \end{align*}
    where the last inequality is due to the trace inequality $\langle A, B \rangle \le \|A\|_{\oper} \, \|B\|_*\le \sqrt{\rank(B)}\|A\|_{\oper} \, \|B\|_{\fro}$.
    This implies that
    \begin{align*}
        \|\Lambda\odot(\hat{X}-X)\|_{\fro}
        &\le  2{\sqrt{\rank(\Lambda\odot(\hat{X}-X))}}\| \Lambda\odot(Y-X)\|_{\oper} \\
        & \leq 2{{\sqrt{\rank(\Lambda\odot(\hat{X}-X))}}} \left(\sqrt{\lambda_1}\|A_Q-U_QU_Q^\top\|_{\oper}+\sqrt{\lambda_2}\|W-UB^\top\|_{\oper}\right)\\
        &\leq 2{{\sqrt{6k}}} \left(\sqrt{\lambda_1}\|A_Q-U_QU_Q^\top\|_{\oper}+\sqrt{\lambda_2}\|W-UB^\top\|_{\oper}\right),
    \end{align*}
   where the second inequality is due to the triangle inequality of the operator norm,
    \begin{align*}
        \| \Lambda\odot(Y-X)\|_{\oper} &\overset{(\text{a})}{ \leq}  \sqrt{\lambda_1}\|A_Q-U_QU_Q^\top\|_{\oper}+\sqrt{\lambda_2/2} \left(\|W-UB^\top\|_{\oper}\right)\\& \leq  \sqrt{\lambda_1}\|A_Q-U_QU_Q^\top\|_{\oper}+\sqrt{\lambda_2} \left(\|W-UB^\top\|_{\oper}\right),
    \end{align*}
     where (a) follows from Lemma \ref{lm:rank}~(i). And the last inequality is because
    \begin{align*}\rank(\Lambda\odot(\hat{X}-X))& \overset{(\text{b})}{\leq}\rank(\widehat U_Q \widehat U_Q^\top- U_Q  U_Q^\top) +2\rank(\widehat U \widehat B^\top- U B^\top)\\ & \leq \rank(\widehat U_Q \widehat U_Q^\top) +\rank( U_Q  U_Q^\top)+2\rank(\widehat U \widehat B^\top)+2\rank( U B^\top)\\ & \leq \rank( \widehat U_Q) +\rank( U_Q  )+2\min\{\rank(\widehat U ),\rank( \widehat B^\top)\}+2\min\{\rank( U ),\rank( B^\top)\}\\ & \leq 6k,
    \end{align*}
    where (b) follows from Lemma \ref{lm:rank}~(ii).
    This completes the proof of the upper bound.

    \paragraph{Part 2: lower bound.}
    We reduce to two classical Gaussian matrix denoising subproblems and invoke the global minimax lower bound of \cite[Theorem~9(1)]{donoho2014minimax}, whose results are stated for a normalized risk with a specific noise scaling of 1.

    For the network-only subproblem, restrict to the subclass where the factor component is absent (so the $W$ block is identically zero and carries no information about the network block).
    Then, estimating $X$ under the weighted loss reduces to estimating an
    $n_Q\times n_Q$ symmetric PSD matrix $M_1$ of rank at most $k$ from
    \(
    A_Q=M_1+D
    \)
    with $D_{ij}\sim \cN(0,\sigma_1^2)$.
    Because the loss weights the $A_Q$ block by $\lambda_1$, we have
    \[
    \|\Lambda\odot(\tilde X-X)\|_{\fro}^2\ \ge\ \lambda_1\|\tilde M_1-M_1\|_{\fro}^2,
    \]
    where $\tilde M_1$ is the corresponding block extracted from $\tilde X(Y)$.
    Applying \citet[Theorem~9(1), case Sym with $m=n=n_Q$]{donoho2014minimax}
    and undoing their normalization and noise scaling yields
    \[
    \inf_{\tilde X(Y)}\sup_{\rank(M_1)\le k}\lambda_1\EE\|\tilde M_1-M_1\|_{\fro}^2
    \ \ge\
    \lambda_1\sigma_1^2\bigl(2kn_Q-(k^2+k)\bigr).
    \]

    For the factor-only subproblem, restrict to the subclass where only the rectangular block $M_2\in\R^{n_P\times d}$ may vary and the remaining blocks are chosen so that $X\in\SSS^{n_P+d}$ and
    $\rank(X)\le k$ (this is always possible by a rank-$k$ factorization of $M_2$).
    Then the data includes
    \(
    W=M_2+E
    \)
    with $E_{ij}\sim \cN(0,\sigma_2^2)$ i.i.d.
    Since the mask places weights $\sqrt{\lambda_2/2}$ on both $W$ and $W^\top$, the weighted Frobenius loss satisfies
    \(
    \|\Lambda\odot(\tilde X-X)\|_{\fro}^2\ge \lambda_2\|\tilde M_2-M_2\|_{\fro}^2.
    \)
    Applying \citet[Theorem~9(1), case Mat with dimensions $(m,n)=(\min\{n_P,d\},\max\{n_P,d\})$]{donoho2014minimax}
    and undoing normalization/scaling gives
    \[
    \inf_{\tilde X(Y)}\sup_{\rank(M_2)\le k}\lambda_2\EE\|\tilde M_2-M_2\|_{\fro}^2
    \ \ge\
    \lambda_2\sigma_2^2\bigl(k(n_P+d)-(k^2+k)\bigr).
    \]

    The full parameter class contains both subclasses considered above, hence the
    minimax risk over $X\in\SSS^{n_P+d}$ is at least the maximum of the two lower
    bounds:
    \begin{align*}
        & \inf_{\tilde X(Y)}\ \sup_{X\in \SSS^{\,n_P+d}:\,\rank(X)\le k}\
        \EE\Big[\ \|\Lambda\odot(\tilde X-X)\|_{\fro}^2\ \Big]\\
        \ \ge\ &
        \max\Bigl\{
        \lambda_1\sigma_1^2\bigl(2kn_Q-k^2-k\bigr),\
        \lambda_2\sigma_2^2\bigl(k(n_P+d)-k^2-k\bigr)
        \Bigr\}. 
    \end{align*}
    This finishes the proof.
\end{proof}

\subsection{Special cases of \Cref{thm:model}}\label{subsec:discussion}
\begin{remark}[Special cases]
    Suppose $D$ and $E$ have i.i.d. 1-sub-Gaussian entries.
    Then we have $\|D\|_{\oper}\asymp \sqrt{n_Q}$ and $\|E\|_{\oper}\asymp \sqrt{n_P+d}$, and the following holds:
    \begin{enumerate}
      \item[(i)] \textbf{Factor regime:} If $\lambda_1=0$, then
        \begin{align*}
            \|\Lambda\odot(\hat{X}-X)\|_{\fro}^2=\lambda_2\|\widehat U \widehat B^\top-UB^\top\|_{\fro}^2 \lesssim & k\lambda_2(n_P+d).
        \end{align*}
        The scaling matches with classical factor analysis when additional identification conditions are imposed \citep{Bai2012STATISTICALAO}.

        \item[(ii)] \textbf{Network regime:} If $\lambda_2=0$, then
        \begin{align*}
             \|\Lambda\odot(\hat{X}-X)\|_{\fro}^2 =\lambda_1\|\widehat U_Q \widehat U_Q^\top-U_QU_Q^\top\|_{\fro}^2\lesssim & {k{\lambda_1}n_Q}.
        \end{align*}
        The scaling matches with Theorem 5 of \cite{ma2020universal} with $\lambda_1=1$.

        \item[(iii)] \textbf{No covariates on $Q^c$:} If $Q^c=\varnothing$, then
        \begin{align*}
            \|\Lambda\odot(\hat{X}-X)\|_{\fro}^2 &= \lambda_1\|A-UU^{\top}\|_{\fro}^2 + \lambda_2\|W-UB^{\top}\|_{\fro}^2 \lesssim   k{\lambda_1}n_Q + k\lambda_2(n_Q+d).
        \end{align*}
        which gives an error bound for two squared errors weighted by $\lambda_1$ and $\lambda_2$.
        Note that \(\min\{\lambda_1,\lambda_2\} (\|A-UU^{\top}\|_{\fro}^2 + \|W-UB^{\top}\|_{\fro}^2) \leq  \|\Lambda\odot(\hat{X}-X)\|_{\fro}^2\).
        One also has a rough bound
        \begin{align*}
            \frac{1}{n_Q(n_Q+d)}(\|A-UU^{\top}\|_{\fro}^2 + \|W-UB^{\top}\|_{\fro}^2) \lesssim \frac{k \max\{\lambda_1,\lambda_2\}}{\min\{\lambda_1,\lambda_2\}} \cdot \frac{1}{n_Q},
        \end{align*}
        which coincides with Theorem 1 of \citet{zhang2022joint} with Gaussian likelihoods and $(\lambda_1,\lambda_2)=(1,\lambda)$.
    \end{enumerate}
\end{remark}

\clearpage
\section{Theoretical analysis of \Cref{alg:joint-UB-nGD}}\label{app:sec:proof-algo}
\subsection{Preliminary}\label{app:subsec:pre}
    Before presenting our proof, we introduce necessary notations.
     The gradients of $\ell(U,B)$ are given by
    \begin{align*}
        \nabla_{B} \ell(U,B) &= 2\lambda_2  (U B^\top - W)^\top U \\
        \nabla_{U_Q} \ell(U,B) &= 4\lambda_1 (  U_Q U_Q^\top U_Q-A_Q U_Q ) + 2\lambda_2 (U_Q B^\top - W_Q) B \\
        \nabla_{U_{Q^c}}\ell(U,B) &= 2\lambda_2 (U_{Q^c} B^\top - W_{Q^c}) B.
    \end{align*}
    For convenience, let
$\mathcal V := \{U_Q,\, B,\, U_{Q^c}\}$.
Define the blockwise Gram matrices
\[
G_{U_Q}^{(t)} := 4\lambda_1 U_Q^{(t)\top} U_Q^{(t)} + 2\lambda_2 B^{(t)\top} B^{(t)},\qquad
G_{B}^{(t)} := 2\lambda_2 U^{(t)\top} U^{(t)},\qquad
G_{U_{Q^c}}^{(t)} := 2\lambda_2 B^{(t)\top} B^{(t)} .
\]
For each block $v\in\mathcal V$, let 
\[
S_v^{(t)} := \bigl(G_v^{(t)}\bigr)^{1/2}
\]
denote the principal matrix square root of $G_v^{(t)}$.
  Then, we define the orthogonal alignment matrices with respect to joint optimization, $U_Q$, $U_{Q^c}$, and $B$, respectively:    
    \begin{align*}
       R^{(t)} \in 
\argmin_{R\in\mathcal O(k)}
\sum_{v\in\mathcal V}
\bigl\|
\bigl(v^{(t)}-v R\bigr) S_v^{(t)}
\bigr\|_{\fro}^2,
    \end{align*}
    where $\cO_k$ contains all $k\times k$ orthonormal matrices.
    Then, we denote the deviations as
    \begin{align}
        \Delta_{U_Q}^{(t)} &= U_Q^{(t)} - U_Q R^{(t)}, 
        \qquad 
        \Delta_{U_{Q^c}}^{(t)} = U_{Q^c}^{(t)} - U_{Q^c} R^{(t)},
        \qquad 
\Delta_B^{(t)} =  B^{(t)} - B{R^{(t)}}.\label{eq:tilde-Delta}
    \end{align}
    and define the weighted estimation error as
    \begin{align*}
        e_t &:=   \sum_{v\in\mathcal{V}}s_v^{(t)},
            \end{align*}
            with
         \begin{align*} s_{U_Q}^{(t)}
&=
4\lambda_1
\bigl\|
\Delta_{U_Q}^{(t)} (U_Q^{(t)})^\top
\bigr\|_{\fro}^2
+
2\lambda_2
\bigl\|
\Delta_{U_Q}^{(t)} (B^{(t)})^\top
\bigr\|_{\fro}^2,\\
s_{B}^{(t)}
&=
2\lambda_2
\bigl\|
\Delta_{B}^{(t)} (U^{(t)})^\top
\bigr\|_{\fro}^2,\\
s_{U_{Q^c}}^{(t)}
&=
2\lambda_2
\bigl\|
\Delta_{U_{Q^c}}^{(t)} (B^{(t)})^\top
\bigr\|_{\fro}^2.
\end{align*}

    The analysis of \Cref{alg:joint-UB-nGD} proceeds via an induction argument. 
    To avoid circularity in the analysis, it is helpful to clearly separate the \emph{local} (one-step) deterministic recursion inequalities from the \emph{global} accumulated inequalities that show these recursion inequalities hold along the whole trajectory.
    For the ease of presentation, we introduce two events at iteration $t$.
     Define the ``good event'' in which the perturbations are controlled:
     \begin{align}
   {\mathcal{G}}_t := \bigl\{
        s_{U_Q}^{(t)}+s_{B}^{(t)}+s_{U_{Q^c}}^{(t)} \leq \tau,
        \text{ hold for all } s\le t\bigr\}.  \label{event:norm-bound}
     \end{align}
     Here
\begin{align*}
\tau
:=
&\frac{4c_1\lambda_1}{\kappa_{U_Q}^4}
\|U_Q\|_{\oper}^4
+
\frac{2c_1\lambda_2}{\kappa_{U_Q}^4}
\|U_Q\|_{\oper}^2\|B\|_{\oper}^2
+
\frac{2c_2\lambda_2}{\kappa_{U_{Q^c}}^4}
\|U_{Q^c}\|_{\oper}^2\|B\|_{\oper}^2
+
\frac{2c_3\lambda_2}{\kappa_{B}^4}
\|B\|_{\oper}^2\|U\|_{\oper}^2,
\end{align*}
where \(c_1,c_2,c_3>0\) are sufficiently small constants, and
\(\kappa_{U_Q}\), \(\kappa_{U_{Q^c}}\), and \(\kappa_B\) denote the
corresponding condition numbers.

    Define the event in which the estimation error recursion holds:
    \begin{align}
      {\cF}_t = \{  e_{s} 
        \le (1-\eta\rho)\, e_{s-1} + C\eta\rho\Delta_S\text{ holds for all }s\leq t \}
        \label{event:iter}
    \end{align}
    for some fixed constants $\rho\in(0,1]$ and $C$ and statistical error term $\Delta_S$ independent of $t$.
    For simplicity, we set $\cG_0 = \cF_0 = \varnothing$.
    In particular, we will show the proof with the following choices of parameters:
    \begin{align}
        \rho &=  1-\sqrt\gamma , \label{eq:rho}\\
       \rho \Delta_S &=  k \left\{\frac{\lambda_1^2 \|U_Q\|_{\oper}^2\|D\|_{\oper}^2+\lambda_2^2\|
E_Q\|_{\oper}^2\|B\|_{\oper}^2 }{2\lambda_1\sigma_{\min}^2(U_Q)+\lambda_2\sigma_{\min}^2(B)}  +\frac{\lambda_2\|E_{Q^c}\|_{\oper}^2\| B\|_{\oper}^2 }{\sigma_{\min}^2(B)}+ \frac{\lambda_2\|E\|_{\oper}^2\| U\|_{\oper}^2 }{\sigma_{\min}^2(U)}\right\}. \label{eq:Delta_S}
    \end{align}
    
    \paragraph{Step 1 (One-step deterministic bounds: \Cref{lem:inner-product-eq,lem:gradient,lem:DeltaS,lem:grad-DeltaS,lem:cross-term-blockwise,lem:ratio-U}).}
        In \Cref{app:subsec:grad}, we first prove directional derivatives and gradient-related bounds for the gradients under the localization condition~\eqref{event:norm-bound}.
        More specifically, \Cref{lem:inner-product-eq} derive exact decomposition of directional derivatives;
        \Cref{lem:gradient,lem:DeltaS,lem:grad-DeltaS,lem:cross-term-blockwise} upper bounds the higher order terms appearing in the recursion formula.
        These lemmas are purely deterministic and are valid \emph{conditionally} on being inside the region~\eqref{event:norm-bound}. \Cref{lem:ratio-U} provides a bound on the interaction ratio, which in turn affects the contraction coefficient.

        \paragraph{Step 2 (Local iteration inequalities: \Cref{lem:iter_all}).}
        Using \Cref{lem:iter_all} and events $\cG_t$ and $\cF_t$, we derive one-step recursion inequalities as defined by the event $\cF_{t+1}$ in \Cref{lem:iter_all}.

        \paragraph{Step 3 (Contraction: \Cref{lem:norm}).}
        Having a local contraction of the form~\eqref{event:iter}, \Cref{lem:norm} shows that the condition \eqref{event:norm-bound} is actually preserved for the next iteration.
        
    \paragraph{Step 4 (Statistical bounds and rates: \Cref{thm:error_algor}).}
        By iterating over Step 2 (\Cref{lem:iter_all}) and Step 3 (\Cref{lem:norm}), we can sum the recursion~\eqref{event:iter} to obtain the non-asymptotic bound on $e_t$ and its stationary value $e_\infty$ in \Cref{thm:error_algor}.

\subsection{Algorithmic Assumptions}\label{sec:asmp}
We denote the condition numbers for $U_Q$, $U_{Q^c}$, and $B$ by $\kappa_{U_Q}$, $\kappa_{U_{Q^c}}$, and $\kappa_{B}$, respectively, where the condition number for a matrix $A$ is defined as $\kappa_A = \|A\|_{\oper}/\sigma_{\min}(A)$.

\begin{assumption}[Rank]\label{asp:row_rank}
    The matrices $U_Q, B$ have full rank, with finite condition numbers $\kappa_{U_Q}, \kappa_{B}$.
    If \(Q^c\neq\varnothing\), assume \(U_{Q^c}\) has full column rank with a finite condition number \(\kappa_{U_{Q^c}}\); otherwise the \(U_{Q^c}\)-block is omitted from the algorithm and analysis.
\end{assumption}

\begin{assumption}[Initialization]\label{asp:init}
   The initial estimator $(U^{(0)},B^{(0})$ satisfies \(e_0\leq \tau,\) where
    \[
    \tau :=
    \frac{4\lambda_1 c_1}{\kappa_{U_Q}^4}\|U_Q\|_{\oper}^4
    + \frac{2 \lambda_2 c_1}{\kappa_{U_Q}^4}\|U_Q\|_{\oper}^2 \|B\|_{\oper}^2
    + \frac{2 \lambda_2 c_2}{\kappa_{U_{Q^c}}^4}\|U_{Q^c}\|_{\oper}^2 \|B\|_{\oper}^2
    + \frac{2 \lambda_2 c_3}{\kappa_{B}^4}\|B\|_{\oper}^2 \|U\|_{\oper}^2.
    \]
    Here \(c_1, c_2, c_3 > 0\) are sufficiently small constants, and
    \(\kappa_{U_Q}\), \(\kappa_{U_{Q^c}}\), and \(\kappa_{B}\) denote the corresponding condition numbers.
\end{assumption}
\begin{assumption}[Signal noise ratio]\label{asp:operator}
    Assume that  \begin{align*}
     \tau \geq  C_1\left\{\frac{\lambda_1^2 \|U_Q\|_{\oper}^2\|D\|_{\oper}^2+\lambda_2^2\|
E_Q\|_{\oper}^2\|B\|_{\oper}^2 }{2\lambda_1\sigma_{\min}^2(U_Q)+\lambda_2\sigma_{\min}^2(B)}  +\frac{\lambda_2\|E_{Q^c}\|_{\oper}^2\| B\|_{\oper}^2 }{\sigma_{\min}^2(B)}+ \frac{\lambda_2\|E\|_{\oper}^2\| U\|_{\oper}^2 }{\sigma_{\min}^2(U)}\right\},
 \end{align*} with a sufficiently large constant $C_1$.
\end{assumption}
\begin{assumption}[Local well-conditioning of the iterates]
\label{asp:local-well-conditioned}
There exist fixed constants
\[
\underline l_{U_Q},\underline l_B,\underline l_{U_{Q^c}}>0,
\qquad
\overline l_{U_Q},\overline l_B,\overline l_{U_{Q^c}}<\infty,
\]
such that, for all iterations under consideration, the iterates satisfy
\[
\sigma_{\min}(U_Q^{(t)})
\ge
\underline l_{U_Q}\sigma_{\min}(U_Q),
\qquad
\sigma_{\min}(B^{(t)})
\ge
\underline l_B\sigma_{\min}(B),
\qquad
\sigma_{\min}(U_{Q^c}^{(t)})
\ge
\underline l_{U_{Q^c}}\sigma_{\min}(U_{Q^c}),
\]
and
\[
\|U_Q^{(t)}\|_{\oper}
\le
\overline l_{U_Q}\|U_Q\|_{\oper},
\qquad
\|B^{(t)}\|_{\oper}
\le
\overline l_B\|B\|_{\oper},
\qquad
\|U_{Q^c}^{(t)}\|_{\oper}
\le
\overline l_{U_{Q^c}}\|U_{Q^c}\|_{\oper}.
\]
\end{assumption}

\Cref{asp:row_rank} requires that the ground truth matrices are well-conditioned to ensure the convergence of the iterative procedure.
\Cref{asp:init} requires a sufficiently accurate initialization, which ensures
local convergence to the global optimum under the stated regularity conditions.
In particular, this condition is satisfied if
\[
\|\Delta_{U_Q}^{(0)}\|_{\fro} \leq \frac{c_1 \|U_Q\|_{\oper}}{\kappa_{U_Q}^{2}},
\qquad
\|\Delta_{U_{Q^c}}^{(0)}\|_{\fro} \leq \frac{c_2 \|U_{Q^c}\|_{\oper}}{\kappa_{U_{Q^c}}^{2}},\qquad
\|\Delta_{B}^{(0)}\|_{\fro} \leq \frac{c_3 \|B\|_{\oper}}{\kappa_{B}^{2}},
\]
and $U_Q^{(0)}$, $B^{(0)}$, and $U_{Q^c}^{(0)}$ are sufficiently close to $U_Q$, $B$, and $U_{Q^c}$, respectively.
\Cref{asp:init} can potentially be relaxed by adopting perturbed gradient descent schemes \citep{jin2017escape}, which are known to converge to a global minimum from random initialization at a linear rate under suitable assumptions. \Cref{asp:operator} requires that the signal level dominates the noise level, that is, the initialization strength $\tau$ is sufficiently larger than the operator norms of the noise matrices $D$ and $E$. This condition is reasonable, as it ensures that the underlying signal can be reliably distinguished from stochastic fluctuations, which is necessary for stable estimation and convergence.
For our theoretical analysis, we focus on normalized gradient descent to study its local convergence properties and the induced statistical errors.

Assumption~\ref{asp:local-well-conditioned} is a local regularity condition.
It holds whenever the initialization is sufficiently close to the truth and the
step size is small enough so that the iterates remain in this local
neighborhood. By Weyl's inequality and the triangle inequality, closeness to
the truth directly implies the stated lower bounds on the minimum singular
values and upper bounds on the operator norms. Thus, the assumption mainly
rules out degenerate iterates and ensures that the blockwise normalizers used
in Algorithm~\ref{alg:joint-UB-nGD} are well conditioned.

Finally, it is worth noting that, because our analysis is based on a Gaussian likelihood objective, which leads to a quadratic loss with well-behaved curvature properties, admitting a globally Lipschitz-continuous gradient and a constant Hessian matrix. Consequently, the boundedness assumptions commonly required in generalized linear model analyses (e.g., \cite{ma2020universal, zhang2022joint}) are not needed for establishing the convergence rate of our algorithm.

\subsection{Proof of \Cref{thm:error_algor}}\label{app:subsec:err-weighted}
\begin{proof}

    We prove by induction.
    
    For $t=0$, \Cref{asp:init} ensures that \eqref{event:norm-bound} holds; i.e., $ {\mathcal{G}}_0$ holds.     
    
    For $t>0$, suppose $\cG_s$ and $\cF_s$ hold for all $s\leq t-1$.
    By \Cref{lem:iter_all}, the event $ {\cF}_t = \{  e_{s} \le (1-\eta\rho)\, e_{s-1} + \eta\rho C \Delta_S\text{ holds for all }s\leq t \}$ holds, with contraction parameter $\rho$ and statistical error $\Delta_S$ given by \eqref{eq:rho} and \eqref{eq:Delta_S}, respectively.
    By \Cref{lem:norm}, we further have that $\cG_t$ also holds.
    In summary, ${\cG}_t$ and ${\cF}_t$ hold.    
    By induction, events \eqref{event:norm-bound} and \eqref{event:iter} hold for all $t\geq 0$.
    
    Under event $ {\cF}_t $, by recursion, we obtain
    \begin{align*}
          e_{t} 
            \le (1-\eta\rho)^t\, e_{0} + \frac{1-(1-\eta\rho)^t}{1- (1-\eta\rho)}\eta\rho C \Delta_S \leq (1-\eta\rho)^t\, e_{0} +  C(\gamma)\Delta_S.
    \end{align*}
    where $0 < 1-\eta\rho<1$. This completes the proof.
\end{proof}
\subsection{Proof of \Cref{cor:weighted-error-UQ}}\label{app:subsec:tradeoff}
\begin{proof}
By the assumption of  \Cref{cor:weighted-error-UQ},
\[
\sup_{0\le t\le T}\frac{s_B^{(t)}}{s_{U_Q}^{(t)}}=o(1),
\qquad n_P\to\infty .
\]
Then the \(B\)-block iteration is asymptotically dominated by the
\(U_Q\)-block.
Specifically, define
\[
\gamma_1^{(t)}
:=
\frac{
\sqrt{2\lambda_2\|B^{(t)}(\Delta_{U_Q}^{(t)})^\top\|_{\fro}^2}
}{
\sqrt{s_{U_Q}^{(t)}}
}
\frac{
\sqrt{2\lambda_2}\|\Delta_B^{(t)} {U_Q^{(t)}}^\top\|_{\fro}
}{
\sqrt{s_B^{(t)}}
}.
\]
By \Cref{lem:inner-product-eq}~(\ref{overcross}) and \Cref{lem:iter_all}, there exists a constant
\[
\widetilde\rho:=1-\sup_{0\le t\le T}\gamma_1^{(t)}\in(0,1)
\]
such that the iterates satisfy
\[
s_{U_Q}^{(t+1)}
\le
(1-\eta\widetilde\rho)s_{U_Q}^{(t)}
+
C\eta k
\frac{
\lambda_1^2 \|U_Q\|_{\oper}^2\|D\|_{\oper}^2
+
\lambda_2^2 \|E_Q\|_{\oper}^2\|B\|_{\oper}^2
}{
2\lambda_1\sigma_{\min}^2(U_Q)
+
\lambda_2\sigma_{\min}^2(B)
}
+
o\!\left(\eta s_{U_Q}^{(t)}\right).
\]
Letting \(T\to\infty\) eliminates the
optimization error, so that the remaining error is dominated by the
statistical terms. In particular, the statistical error  of \(U_Q\) is
\[
\frac{k(
\lambda_1^2 \|U_Q\|_{\oper}^2\|D\|_{\oper}^2
+
\lambda_2^2 \|E_Q\|_{\oper}^2\|B\|_{\oper}^2)
}{
2\lambda_1\sigma_{\min}^2(U_Q)
+
\lambda_2\sigma_{\min}^2(B)
}.
\]
Furthermore,  for the corresponding Procrustes rotation \(R\),
\begin{align*}
      \frac{1}{k}
    \bigl\|\widehat U_Q-U_QR\bigr\|_{\fro}^2
    \;\lesssim\;
    \frac{
    \lambda_1^2\|U_Q\|_{\oper}^2\|D\|_{\oper}^2
    +
    \lambda_2^2\|B\|_{\oper}^2\|E_Q\|_{\oper}^2
    }{
    \bigl(
    \lambda_1\sigma_{\min}^2(U_Q)
    +
    \lambda_2\sigma_{\min}^2(B)
    \bigr)^2
    }.
\end{align*}
This proves the claim.
\end{proof}

\subsection[Discussion of Corollary \ref{cor:weighted-error-UQ} and Error Bound for Qc]{Discussion of \Cref{cor:weighted-error-UQ} and Error Bound for \(Q^c\)}\label{subsec:dis_error}
\begin{remark}[Interpretation of the block-dominance condition]
The condition
$\sup_{0\le t\le T}{s_B^{(t)}}/{s_{U_Q}^{(t)}}=o(1),
\; n_P\to\infty$,
is stated in terms of the curvature-weighted errors.  If $\lambda_2=0$, then ${
\|\Delta_B^{(t)}\|_{\fro}^2
}/{
\|\Delta_{U_Q}^{(t)}\|_{\fro}^2
}=0$, if $\lambda_2\neq0$, it also implies the corresponding dominance
in the ordinary Frobenius norm. Specifically, 
by the definitions of \(s_{U_Q}^{(t)}\) and \(s_B^{(t)}\), we have
\[
s_{U_Q}^{(t)}
\le
\|S_{U_Q}^{(t)}\|_{\oper}^2
\|\Delta_{U_Q}^{(t)}\|_{\fro}^2,\qquad
s_B^{(t)}
\ge
\sigma_{\min}^2(S_B^{(t)})
\|\Delta_B^{(t)}\|_{\fro}^2.
\]
Combining the two inequalities gives
\[
\frac{
\|\Delta_B^{(t)}\|_{\fro}^2
}{
\|\Delta_{U_Q}^{(t)}\|_{\fro}^2
}
\le
\frac{\|S_{U_Q}^{(t)}\|_{\oper}^2}
{\sigma_{\min}^2(S_B^{(t)})}
\cdot
\frac{s_B^{(t)}}{s_{U_Q}^{(t)}}\le \sup_{0\le t\le T}
\frac{
4\lambda_1\|U_Q^{(t)}\|_{\oper}^2
+
2\lambda_2\|B^{(t)}\|_{\oper}^2
}{
2\lambda_2\sigma_{\min}^2(U^{(t)})
}\cdot
\frac{s_B^{(t)}}{s_{U_Q}^{(t)}}.
\]
Hence, if
$
\sup_{0\le t\le T}{s_B^{(t)}}/{s_{U_Q}^{(t)}}=o(1)$ and
$
\sup_{0\le t\le T}
({
4\lambda_1\|U_Q^{(t)}\|_{\oper}^2
+
2\lambda_2\|B^{(t)}\|_{\oper}^2
})/{
2\lambda_2\sigma_{\min}^2(U^{(t)})
}
=O(1)$.
Then
\[
\sup_{0\le t\le T}
\frac{
\|\Delta_B^{(t)}\|_{\fro}^2
}{
\|\Delta_{U_Q}^{(t)}\|_{\fro}^2
}
=o(1).
\]
Thus, under this relative-conditioning condition, the weighted block-dominance
assumption means that the ordinary \(B\)-block estimation error is
asymptotically negligible relative to the target-block error.
\end{remark}

\begin{corollary}[Error bound for the embedding-supported nodes $Q^c$]
\label{cor:uqc}
Under the assumptions of \Cref{cor:weighted-error-UQ}, suppose additionally
that \(\lambda_2>0\), 
 for the limiting estimator \(\widehat U_{Q^c}\) and the corresponding
Procrustes rotation \(R\),
\[
\frac{1}{k}
\bigl\|\widehat U_{Q^c}-U_{Q^c}R\bigr\|_{\fro}^2
\;\lesssim\;
\frac{
\lambda_1^2 \|U_Q\|_{\oper}^2\|D\|_{\oper}^2
+
\lambda_2^2 \|E_Q\|_{\oper}^2\|B\|_{\oper}^2
}{
\lambda_2\sigma_{\min}^2(B)
\bigl(
2\lambda_1\sigma_{\min}^2(U_Q)
+
\lambda_2\sigma_{\min}^2(B)
\bigr)
}
+
\frac{
\|E_{Q^c}\|_{\oper}^2\|B\|_{\oper}^2
}{
\sigma_{\min}^4(B)
}.
\]
\end{corollary}

The bound separates the error of \(U_{Q^c}\) into two sources. The first term
is inherited from the target-block error for \(U_Q\), and can be viewed as the
price of joint estimation through the shared embedding factor \(B\). Indeed,
under
$\sup_{0\le t\le T}{s_B^{(t)}}/{s_{U_Q}^{(t)}}=o(1)$,
the \(B\)-block error is dominated by the \(U_Q\)-block error, so the
uncertainty transmitted to \(U_{Q^c}\) through \(B\) is controlled by the same
weighted statistical error that governs target-latent recovery. The second
term is the direct contribution of embedding noise on \(Q^c\), reflecting that
the non-target latent positions are recovered only through the proxy
embeddings.
\begin{proof}[Proof of \Cref{cor:uqc}]
By \Cref{lem:inner-product-eq}~(\ref{overcross}) and \Cref{lem:iter_all},
the weighted errors for the \(U_Q\)- and \(U_{Q^c}\)-blocks satisfy
\[
\begin{aligned}
s_{U_Q}^{(t+1)}+s_{U_{Q^c}}^{(t+1)}
\le\;&
(1-\eta\rho)
\bigl(s_{U_Q}^{(t)}+s_{U_{Q^c}}^{(t)}\bigr)
\\
&+
C\eta k
\frac{
\lambda_1^2 \|U_Q\|_{\oper}^2\|D\|_{\oper}^2
+
\lambda_2^2 \|E_Q\|_{\oper}^2\|B\|_{\oper}^2
}{
2\lambda_1\sigma_{\min}^2(U_Q)
+
\lambda_2\sigma_{\min}^2(B)
}
\\
&+
C\eta k
\frac{
\lambda_2\|E_{Q^c}\|_{\oper}^2\|B\|_{\oper}^2
}{
\sigma_{\min}^2(B)
}
+
o\!\left(\eta s_{B}^{(t)}\right).
\end{aligned}
\]
Under the assumptions of \Cref{cor:weighted-error-UQ}, we have
\[
\sup_{0\le t\le T}
\frac{s_B^{(t)}}{s_{U_Q}^{(t)}}=o(1),
\]
so the \(B\)-block interaction is asymptotically negligible relative to the
target-block error. Hence the \(o(\eta s_{B}^{(t)})\) term can be absorbed
into the contraction term. Iterating the above recursion gives
\[
\begin{aligned}
s_{U_Q}^{(t+1)}+s_{U_{Q^c}}^{(t+1)}
\le\;&
(1-\eta\rho)^{t+1}
\bigl(s_{U_Q}^{(0)}+s_{U_{Q^c}}^{(0)}\bigr)
\\
&+
Ck
\frac{
\lambda_1^2 \|U_Q\|_{\oper}^2\|D\|_{\oper}^2
+
\lambda_2^2 \|E_Q\|_{\oper}^2\|B\|_{\oper}^2
}{
2\lambda_1\sigma_{\min}^2(U_Q)
+
\lambda_2\sigma_{\min}^2(B)
}
\\
&+
Ck
\frac{
\lambda_2\|E_{Q^c}\|_{\oper}^2\|B\|_{\oper}^2
}{
\sigma_{\min}^2(B)
}.
\end{aligned}
\]
Letting \(t\to\infty\) removes the optimization term. Therefore,
\[
\frac{1}{k}s_{U_{Q^c}}^{(\infty)}
\lesssim
\frac{
\lambda_1^2 \|U_Q\|_{\oper}^2\|D\|_{\oper}^2
+
\lambda_2^2 \|E_Q\|_{\oper}^2\|B\|_{\oper}^2
}{
2\lambda_1\sigma_{\min}^2(U_Q)
+
\lambda_2\sigma_{\min}^2(B)
}
+
\frac{
\lambda_2\|E_{Q^c}\|_{\oper}^2\|B\|_{\oper}^2
}{
\sigma_{\min}^2(B)
}.
\]
By definition,
\[
s_{U_{Q^c}}^{(t)}
=
2\lambda_2
\bigl\|
\Delta_{U_{Q^c}}^{(t)}(B^{(t)})^\top
\bigr\|_{\fro}^2 .
\]
Since the iterates satisfy
\(\sigma_{\min}(B^{(t)})\gtrsim \sigma_{\min}(B)\), we have
\[
s_{U_{Q^c}}^{(t)}
\gtrsim
\lambda_2\sigma_{\min}^2(B)
\|\Delta_{U_{Q^c}}^{(t)}\|_{\fro}^2 .
\]
Consequently, for the limiting estimator
\(\widehat U_{Q^c}\) and the corresponding rotation \(R\),
\[
\frac{1}{k}
\bigl\|\widehat U_{Q^c}-U_{Q^c}R\bigr\|_{\fro}^2
\lesssim
\frac{
\lambda_1^2 \|U_Q\|_{\oper}^2\|D\|_{\oper}^2
+
\lambda_2^2 \|E_Q\|_{\oper}^2\|B\|_{\oper}^2
}{
\lambda_2\sigma_{\min}^2(B)
\bigl(
2\lambda_1\sigma_{\min}^2(U_Q)
+
\lambda_2\sigma_{\min}^2(B)
\bigr)
}
+
\frac{
\|E_{Q^c}\|_{\oper}^2\|B\|_{\oper}^2
}{
\sigma_{\min}^4(B)
}.
\]
This completes the proof.
\end{proof}
\subsection{Partial-oracle analysis of the trade-off}\label{subsec:Poa}
\begin{proposition}[Partial-oracle target risk]
\label{prop:tightness_partial_oracle_short}
Suppose \Cref{asp:row_rank} holds and $Q^c=\varnothing$, so that $U=U_Q$.
Consider the partial-oracle estimator
\[
\hat U_{\mathrm{oracle}}
\in \argmin_{V\in\RR^{n\times k}}
\left\{
\lambda_1\|A-VU^\top\|_{\fro}^2
+
\lambda_2\|W-VB^\top\|_{\fro}^2
\right\}.
\]
Then
\[
\hat U_{\mathrm{oracle}}-U
=
(\lambda_1DU+\lambda_2EB)
(\lambda_1U^\top U+\lambda_2B^\top B)^{-1}.
\]
Moreover, under \Cref{asm:noise},
\begin{equation}
\frac{1}{k}\EE_{D,E}\|\hat U_{\mathrm{oracle}}-U\|_{\fro}^2
\asymp
\left(
\frac{\lambda_1\|U\|_{\oper}\sigma_1}
{\lambda_1\sigma_{\min}^2(U)+\lambda_2\sigma_{\min}^2(B)}
\right)^2
+
\left(
\frac{\lambda_2\|B\|_{\oper}\sigma_2}
{\lambda_1\sigma_{\min}^2(U)+\lambda_2\sigma_{\min}^2(B)}
\right)^2 ,
\label{eq:U_oracle_rate}
\end{equation}
where the constants depend only on the condition-number bounds in
\Cref{asp:row_rank}.
\end{proposition}
\begin{proof}[Proof of \Cref{prop:tightness_partial_oracle_short}]
Write the oracle objective as
\[
\ell_{\mathrm{oracle}}(V)
=\lambda_1\|A-VU^\top\|_{\fro}^2+\lambda_2\|W-VB^\top\|_{\fro}^2,
\qquad V\in\RR^{n\times k}.
\]
Differentiating and setting the gradient to zero gives
\[
0=\nabla_V \ell_{\mathrm{oracle}}(V)
=2\lambda_1(VU^\top-A)U+2\lambda_2(VB^\top-W)B,
\]
hence, with
\[
S:=\lambda_1U^\top U+\lambda_2 B^\top B,
\]
the minimizer satisfies the normal equation
\[
\hat U_{\mathrm{oracle}}\,S=\lambda_1AU+\lambda_2WB.
\]
Under \Cref{asp:row_rank}, $S\succ 0$, so
\[
\hat U_{\mathrm{oracle}}=(\lambda_1AU+\lambda_2WB)\,S^{-1}.
\]
Using $A=UU^\top+D$ and $W=UB^\top+E$, we obtain
\[
AU=U(U^\top U)+DU,
\qquad
WB=U(B^\top B)+EB,
\]
and therefore
\[
\hat U_{\mathrm{oracle}}-U
=\bigl(\lambda_1DU+\lambda_2EB\bigr)S^{-1}.
\]
Let $G:=\lambda_1DU+\lambda_2EB\in\RR^{n\times k}$. Then
\[
\|\hat U_{\mathrm{oracle}}-U\|_{\fro}^2
=\|GS^{-1}\|_{\fro}^2
=\tr\!\bigl(S^{-1}G^\top G S^{-1}\bigr).
\]
Taking expectation over $(D,E)$ and using that $D$ and $E$ are independent and mean-zero, the cross term vanishes:
\[
\EE_{D,E}[G^\top G]
=\lambda_1^2\,\EE[U^\top D^\top D\,U]+\lambda_2^2\,\EE[B^\top E^\top E\,B].
\]
Under \Cref{asm:noise}, we have
\[
\EE[D^\top D]=\sigma_1^2 I_n,
\qquad
\EE[E^\top E]=\sigma_2^2 I_d,
\]
so
\[
\EE_{D,E}[G^\top G]
=\lambda_1^2\sigma_1^2\,U^\top U+\lambda_2^2\sigma_2^2\,B^\top B.
\]
Consequently,
\[
\EE_{D,E}\!\big[\|\hat U_{\mathrm{oracle}}-U\|_{\fro}^2\big]
=\tr\!\Bigl(S^{-1}\bigl(\lambda_1^2\sigma_1^2\,U^\top U+\lambda_2^2\sigma_2^2\,B^\top B\bigr)S^{-1}\Bigr).
\]
Using the eigenvalue bounds $\lambda_{\min}(S)I\preceq S\preceq \lambda_{\max}(S)I$ yields the sandwich
\[
\frac{1}{\lambda_{\max}(S)^2}\,\tr\!\bigl(\lambda_1^2\sigma_1^2\,U^\top U+\lambda_2^2\sigma_2^2\,B^\top B\bigr)
\;\le\;
\EE_{D,E}\!\big[\|\hat U_{\mathrm{oracle}}-U\|_{\fro}^2\big]
\;\le\;
\frac{1}{\lambda_{\min}(S)^2}\,\tr\!\bigl(\lambda_1^2\sigma_1^2\,U^\top U+\lambda_2^2\sigma_2^2\,B^\top B\bigr).
\]
Finally, under \Cref{asp:row_rank} and fixed $k$,
\[
\tr(U^\top U)=\|U\|_{\fro}^2 \asymp \|U\|_{\oper}^2 \asymp \sigma_{\min}^2(U),
\qquad
\tr(B^\top B)=\|B\|_{\fro}^2 \asymp \|B\|_{\oper}^2 \asymp \sigma_{\min}^2(B), 
\]
and
\[
\lambda_{\min}(S)=\sigma_{\min}^2(S^{1/2})\asymp \lambda_1\sigma_{\min}^2(U)+\lambda_2\sigma_{\min}^2(B),\qquad
\lambda_{\max}(S)=\sigma_{\max}^2(S^{1/2})\asymp \lambda_1\sigma_{\min}^2(U)+\lambda_2\sigma_{\min}^2(B),
\]
with implicit constants depending only on the condition numbers in \Cref{asp:row_rank}.
Combining the above displays gives \eqref{eq:U_oracle_rate}.
\end{proof}

\clearpage
\section{Lemmas and their proofs}\label{app:sec:lemma}

\subsection{Deterministic gradient lemmas}\label{app:subsec:grad}
\begin{lemma}[Directional derivative decomposition]\label{lem:inner-product-eq}
For \Cref{alg:joint-UB-nGD}, the directional derivatives admit the decomposition:
 \begin{enumerate}[(i)]
    \item \label{eqcross} For each individual update of $U_Q,\;B$ and $U_{Q^c}$, it holds that  \begin{align*}
       - 2 \langle \nabla_{U_Q}^{(t)}\ell ,\Delta_{U_Q}^{(t)}\rangle &=  -2 s_{U_Q}^{(t)}
       -2\lambda_2  \langle \sqrt{2}U_Q^{(t)} ({\Delta_{B}^{(t)}})^\top , \sqrt{2}\Delta_{U_Q}^{(t)} {B^{(t)}}^\top \rangle
       -2\lambda_1 \big\|2(U_Q^{(t)})^\top \Delta_{U_Q}^{(t)}\big\|_{\fro}^2\\
        &\quad
       +\left[       8\lambda_1\big\langle \Delta_{U_Q}^{(t)}(\Delta_{U_Q}^{(t)})^\top,\ \Delta_{U_Q}^{(t)}(U_Q^{(t)})^\top \big\rangle 
        + 4\lambda_2 \langle \Delta_{U_Q}^{(t)}({\Delta_{B}^{(t)}})^\top ,  \Delta_{U_Q}^{(t)} {B^{(t)}}^\top  \rangle \right]\\ 
        &\quad - \langle  8\lambda_1DU_Q^{(t)}+4\lambda_2E_Q {B^{(t)}}  , \Delta_{U_Q}^{(t)}   \rangle  \\
        - 2 \langle \nabla_{B}^{(t)}\ell ,\Delta_{B}^{(t)}\rangle &= -2 s_{B}^{(t)}-4 \lambda_2\langle B^{(t)}(\Delta_U^{(t)})^\top,\Delta_B^{(t)} {U^{(t)}}^\top\rangle+4 \lambda_2\langle \Delta_B^{(t)}(\Delta_U^{(t)})^\top,\Delta_B^{(t)} {U^{(t)}}^\top\rangle\\ &\quad -4\lambda_2\langle E^\top {U^{(t)}},\Delta_B^{(t)} \rangle\\
        - 2 \langle \nabla_{U_{Q^c}}^{(t)}\ell ,\Delta_{U_{Q^c}}^{(t)}\rangle &= -2 s_{U_{Q^c}}^{(t)}-4\lambda_2  \langle U_{Q^c}^{(t)} ({\Delta_{B}^{(t)}})^\top , \Delta_{U_{Q^c}}^{(t)} {B^{(t)}}^\top \rangle+4\lambda_2 \langle \Delta_{U_{Q^c}}^{(t)}({\Delta^{(t)}_{B}})^\top ,  \Delta_{U_{Q^c}}^{(t)} {B^{(t)}}^\top  \rangle\\ &\quad -4\lambda_2  \langle E_{Q^c}{B^{(t)}},  \Delta_{U_{Q^c}}^{(t)} \rangle.
    \end{align*}
    \item \label{overcross} For any constants $\BC_1,\BC_2,\BC_3,\BC_4,\BC_5,\BC_6>0$, we have
    \begin{align*}
      - 2 \sum_{v\in\mathcal V}  
     \langle \nabla_v^{(t)} \ell, \Delta_v^{(t)} \rangle
    \le &
    -2(1-\sqrt\gamma)\sum_{v\in\mathcal V}s_v^{(t)} 
    +\frac{k}{\BC_1}\|  ( 4\lambda_1D U_Q^{(t)}+2\lambda_2 E_Q B^{(t)})(S_{U_Q}^{(t)})^{-1}\|_{\oper}^2\\ &+\frac{4k}{\BC_2}\|\lambda_2  E_{Q^c} B^{(t)} (S_{U_{Q^c}}^{(t)})^{-1}\|_{\oper}^2 +\frac{4k}{\BC_3}\|\lambda_2 E^\top U^{(t)}(S_{B}^{(t)})^{-1}\|_{\oper}^2 \\ &+(4\BC_1+\BC_4)\lambda_1
    \|\Delta_{U_Q}^{(t)}U_Q^{(t)\top}\|_{\fro}^2 +\BC_6\lambda_2
    \|\Delta_U^{(t)}B^{(t)\top}\|_{\fro}^2 \\
    &+2\BC_1\lambda_2
    \|\Delta_{U_Q}^{(t)}B^{(t)\top}\|_{\fro}^2 +2\BC_2\lambda_2
    \|\Delta_{U_{Q^c}}^{(t)}B^{(t)\top}\|_{\fro}^2 
    +(2\BC_3+\BC_5)\lambda_2
    \|\Delta_B^{(t)}U^{(t)\top}\|_{\fro}^2 \\
    &+\frac{16\lambda_1}{\BC_4}\|\Delta_{U_Q}^{(t)}\|_{\fro}^4 +4\lambda_2\!\left(\frac{1}{\BC_5}+\frac{1}{\BC_6}\right)
    \|\Delta_U^{(t)}\|_{\fro}^2
    \|\Delta_B^{(t)}\|_{\fro}^2 .
    \end{align*}
\end{enumerate}
   
\end{lemma}
\begin{proof}[Proof of \Cref{lem:inner-product-eq}]
    We split the proof into three parts.
    
    \paragraph{(1) Directional derivative with respect to $U_Q$.}
    Using the definition $D=A_Q-U_QU_Q^{\top}$ and $E=W-UB^{\top}$, the gradient $\nabla_{U_Q}^{(t)}\ell$ is given by:
    \begin{align}
         \nabla_{U_Q}^{(t)}\ell &= 4\lambda_1( U_Q^{(t)} (U_Q^{(t)})^\top U_Q^{(t)}-A_QU_Q^{(t)} ) + 2\lambda_2  (U_Q^{(t)} {B^{(t)}}^\top - W_Q) B^{(t)}\notag\\
         &= 4\lambda_1(U_Q^{(t)} (U_Q^{(t)})^{\top} -U_QU_Q^\top) U_Q^{(t)} + 
         2\lambda_2 (U_Q^{(t)} {B^{(t)}}^{\top} - U_Q B^{\top} )  {B}^{(t)}
         -(4\lambda_1 D U_Q^{(t)}
         +2\lambda_2 E_Q{B}^{(t)})\notag\\
         &:= I_1 + I_2 + I_3,\label{eq:I1I2I3}
    \end{align}
    where the gradient components $I_1$, $I_2$ and $I_3$ capture the influence from the network model, the factor model, and the noise term, respectively. 

    For the network gradient term $I_1$, note that $ U_Q U_Q^\top =U_Q R^{(t)}\bigl(U_Q R^{(t)}\bigr)^\top$, $\Delta_{U_Q}^{(t)} = U_Q^{(t)} -U_Q R^{(t)}$ and $(U_Q^{(t)})^\top\Delta_{U_Q}^{(t)}$ is symmetric from \Cref{lem:procrustes_basic}.
    It implies that
     \begin{align*}
        &-2 \langle I_1, \Delta_{U_Q}^{(t)}  \rangle \\
        =&-8\lambda_1 \langle (U_Q^{(t)} (U_Q^{(t)})^{\top} -U_QU_Q^\top) , \Delta_{U_Q}^{(t)} (U_Q^{(t)})^{\top}  \rangle\\
        &-8\lambda_1 \langle (U_Q^{(t)} (U_Q^{(t)})^{\top} -U_QU_Q^\top) , \Delta_{U_Q}^{(t)} (U_Q^{(t)})^{\top}  \rangle\\ = &-8\lambda_1 \langle U_Q^{(t)}{\Delta_{U_Q}^{(t)}}^\top
       + \Delta_{U_Q}^{(t)}\bigl(U_Q^{(t)}\bigr)^\top  - \Delta_{U_Q}^{(t)}{\Delta_{U_Q}^{(t)}}^\top, \Delta_{U_Q}^{(t)} (U_Q^{(t)})^{\top} \rangle\\
       =& -8\lambda_1\Big(
        {\big\langle U_Q^{(t)}(\Delta_{U_Q}^{(t)})^\top,\ 
        \Delta_{U_Q}^{(t)}(U_Q^{(t)})^\top \big\rangle}
        +{\big\langle \Delta_{U_Q}^{(t)}(U_Q^{(t)})^\top,\ 
        \Delta_{U_Q}^{(t)}(U_Q^{(t)})^\top \big\rangle}
        -{\big\langle \Delta_{U_Q}^{(t)}(\Delta_{U_Q}^{(t)})^\top,\ 
        \Delta_{U_Q}^{(t)}(U_Q^{(t)})^\top \big\rangle}
        \Big) \\
        =& -8\lambda_1\Big(
        {\big\|(U_Q^{(t)})^\top \Delta_{U_Q}^{(t)}\big\|_{\fro}^2}
        +{\big\|\Delta_{U_Q}^{(t)}(U_Q^{(t)})^\top\big\|_{\fro}^2}
        -{\big\langle \Delta_{U_Q}^{(t)}(\Delta_{U_Q}^{(t)})^\top,\ 
        \Delta_{U_Q}^{(t)}(U_Q^{(t)})^\top \big\rangle}
        \Big).
     \end{align*}

    For the factor gradient term $I_2$, we have
    \begin{align*}
        &-2 \langle I_2, \Delta_{U_Q}^{(t)} \rangle\\
        =&  -4 \lambda_2 \langle \; U_Q^{(t)} {B^{(t)}}^{\top} - U_Q{B}^{\top}, \Delta_{U_Q}^{(t)}{{B}^{(t)}}^\top \; \rangle \\        
        = & -  4 \lambda_2 \langle U_Q^{(t)} {B^{(t)}}^{\top} - U_QR^{(t)}{B}^{(t)\top}+U_Q R^{(t)}{B}^{(t)\top} -U_Q {B}^{\top}  , \Delta_{U_Q}^{(t)}{B^{(t)}}^\top \; \rangle \\
        =& - 4\lambda_2  \langle  (U_Q^{(t)} - U_Q R^{(t)}){B^{(t)}}^\top ,  \Delta_{U_Q}^{(t)}{B^{(t)}}^\top  \rangle \\
        &\quad - 4\lambda_2  \langle U_Q(B^{(t)}{R^{(t)} }^\top - B)^\top ,  \Delta_{U_Q}^{(t)}{B^{(t)}}^\top  \rangle\\= &-4\lambda_2  \|\Delta_{U_Q}^{(t)}{B^{(t)}}^\top\|_{\fro}^2-4\lambda_2  \langle U_Q^{(t)} ({\Delta_{B}^{(t)}})^\top , \Delta_{U_Q}^{(t)} {B^{(t)}}^\top \rangle+4\lambda_2  \langle \Delta_{U_Q}^{(t)}({\Delta_{B}^{(t)}})^\top ,  \Delta_{U_Q}^{(t)} {B^{(t)}}^\top  \rangle.
    \end{align*}

    For the noise term $I_3$, we have
       \begin{align*}
        &2 \langle I_3, \Delta_{U_Q}^{(t)} \rangle  \\
       =&8\lambda_1 \langle  DU_Q^{(t)} , \Delta_{U_Q}^{(t)}   \rangle + 4\lambda_2  \langle E_Q {B^{(t)}},  \Delta_{U_Q}^{(t)}  \rangle\\=& \langle  8\lambda_1DU_Q^{(t)}+4\lambda_2E_Q {B^{(t)}}  , \Delta_{U_Q}^{(t)}  \rangle .\end{align*}
   
    \paragraph{(2) Directional derivative with respect to $B$.}
    The gradient $ \nabla_B^{(t)} \ell$ can be decomposed as
    \begin{align*}
        \nabla_B^{(t)} \ell&= 2\lambda_2(U^{(t)}{B^{(t)}}^\top-W)^\top U^{(t)}\\
       & =2\lambda_2(U^{(t)}{B^{(t)}}^\top-UB^\top)^\top U^{(t)}+2\lambda_2E^\top U^{(t)}\\ & :=B_1+B_2.
    \end{align*}
    The inner products are
    \begin{align*}
        -2\langle B_1,\Delta_B^{(t)} \rangle &=   -4 \lambda_2\langle (U^{(t)}{B^{(t)}}^\top-UB^\top)^\top,\Delta_B^{(t)} {U^{(t)}}^\top\rangle\\
       &= -4 \lambda_2\langle ({B^{(t)}}-{BR^{(t)}}){U^{(t)}}^\top ,\Delta_B^{(t)} {U^{(t)}}^\top\rangle-4 \lambda_2\langle BR^{(t)}(U^{(t)}-UR^{(t)})^\top,\Delta_B^{(t)} {U^{(t)}}^\top\rangle\\ &= -4\lambda_2\|\Delta_B^{(t)} {U^{(t)}}^\top\|_{\fro}^2-4 \lambda_2\langle B^{(t)}(\Delta_U^{(t)})^\top,\Delta_B^{(t)} {U^{(t)}}^\top\rangle+4 \lambda_2\langle \Delta_B^{(t)}(\Delta_U^{(t)})^\top,\Delta_B^{(t)} {U^{(t)}}^\top\rangle
    \end{align*}
    and 
    \begin{align*}
      -2\langle B_2,\Delta_B^{(t)} \rangle &=   -4\lambda_2\langle E^\top {U^{(t)}},\Delta_B^{(t)} \rangle.{U^{(t)}}^\top\|_{\fro}^2.
    \end{align*}
    Thus, $-2\langle   \nabla_B^{(t)} \ell ,\Delta_B^{(t)}\rangle=-2\langle   B_1+B_2,\Delta_B^{(t)}\rangle$, which yields the stated result.
      \paragraph{(3) Directional derivative with respect to $U_{Q^c}$.}
        The gradient for $U_{Q^c}$ is 
     \begin{align*}
      \nabla^{(t)}_{U_{Q^c}}\ell(U,B) &= 2\lambda_2 (U^{(t)}_{Q^c} {B^{(t)}}^\top - W_{Q^c}) B^{(t)}\\   &= 2\lambda_2 (U^{(t)}_{Q^c} {B^{(t)}}^\top -U_{Q^c} {B}^\top) B^{(t)}-2\lambda_2 E_{Q^c} B^{(t)}\\ &:=K_1+K_2.
     \end{align*}
     Analogous to the $I_2$ term in the analysis of $U_Q$, we decompose
     \begin{align*}
        &-2\langle K_1,\Delta_{U_{Q^c}} \rangle\\ =&-4\lambda_2  \|\Delta_{U_{Q^c}}^{(t)}{B^{(t)}}^\top\|_{\fro}^2-4\lambda_2  \langle U_{Q^c}^{(t)} ({\Delta_{B}^{(t)}})^\top , \Delta_{U_{Q^c}}^{(t)} {B^{(t)}}^\top \rangle+4\lambda_2 \langle \Delta_{U_{Q^c}}^{(t)}({\Delta^{(t)}_{B}})^\top ,  \Delta_{U_{Q^c}}^{(t)} {B^{(t)}}^\top  \rangle .
     \end{align*}
     In addition, 
      \begin{align*}
            &-2 \langle K_2, \Delta_{U_{Q^c}}^{(t)} \rangle  
            = -4\lambda_2  \langle E_{Q^c}{B^{(t)}},  \Delta_{U_{Q^c}}^{(t)} \rangle.
        \end{align*}
       By $-2\langle   \nabla_{U_{Q^c}}^{(t)} \ell ,\Delta_{U_{Q^c}}^{(t)}\rangle=-2\langle   K_1+K_2,\Delta_{U_{Q^c}}^{(t)}\rangle$,  the desired conclusion follows.
     \paragraph{(4) Directional derivative with respect to overall variable.}
     \begin{align*}
         - 2 \sum_{v\in\mathcal V}  \langle \nabla_v^{(t)} \ell, \Delta_v^{(t)} \rangle &=  -2 \sum_{v\in\mathcal V} s_{v}^{(t)}
           -2\lambda_1 \big\|2(U_Q^{(t)})^\top \Delta_{U_Q}^{(t)}\big\|_{\fro}^2\\&\quad   -4\lambda_2  \big\{\langle U_Q^{(t)} ({\Delta_{B}^{(t)}})^\top , \Delta_{U_Q}^{(t)} {B^{(t)}}^\top \rangle+  \langle U_{Q^c}^{(t)} ({\Delta_{B}^{(t)}})^\top , \Delta_{U_{Q^c}}^{(t)} {B^{(t)}}^\top\rangle+\langle B^{(t)}(\Delta_U^{(t)})^\top,\Delta_B^{(t)} {U^{(t)}}^\top\rangle \big\}\\
            &\quad
           +\left\{      8\lambda_1\big\langle \Delta_{U_Q}^{(t)}(\Delta_{U_Q}^{(t)})^\top,\ \Delta_{U_Q}^{(t)}(U_Q^{(t)})^\top \big\rangle 
            + 4\lambda_2 \langle \Delta_{U_Q}^{(t)}({\Delta_{B}^{(t)}})^\top ,  \Delta_{U_Q}^{(t)} {B^{(t)}}^\top  \rangle \right.\\ 
           & \qquad \left.+4 \lambda_2\langle \Delta_B^{(t)}(\Delta_U^{(t)})^\top,\Delta_B^{(t)} {U^{(t)}}^\top\rangle +4\lambda_2 \langle \Delta_{U_{Q^c}}^{(t)}({\Delta^{(t)}_{B}})^\top ,  \Delta_{U_{Q^c}}^{(t)} {B^{(t)}}^\top  \rangle\right\}\\&\quad - \big\{8 \lambda_1 \langle  D U_Q^{(t)}, \Delta_{U_Q}^{(t)}   \rangle + 4\lambda_2 \langle E_Q{B^{(t)}},  \Delta_{U_Q}^{(t)}  \rangle\\ &\qquad +4\lambda_2\langle E^\top{U^{(t)}},\Delta_B^{(t)} \rangle+4\lambda_2  \langle E_{Q^c}{B^{(t)}},  \Delta_{U_{Q^c}}^{(t)}  \rangle \big\}\\
    &=
    -2 \sum_{v\in\mathcal V} s_{v}^{(t)}
    -2\lambda_1 \big\|2(U_Q^{(t)})^\top \Delta_{U_Q}^{(t)}\big\|_{\fro}^2 
    -8\lambda_2
    \big\langle
    B^{(t)}(\Delta_U^{(t)})^\top,\,
    \Delta_B^{(t)} U^{(t)\top}
    \big\rangle
    \\
    &\quad+8\lambda_1\big\langle
    \Delta_{U_Q}^{(t)}\Delta_{U_Q}^{(t)\top},\,
    \Delta_{U_Q}^{(t)}U_Q^{(t)\top}
    \big\rangle 
    +4\lambda_2
    \big\langle
    \Delta_B^{(t)}(\Delta_U^{(t)})^\top,\,
    \Delta_B^{(t)}U^{(t)\top}
    \big\rangle  + 4\lambda_2 \langle \Delta_{U}^{(t)}({\Delta_{B}^{(t)}})^\top ,  \Delta_{U}^{(t)} {B^{(t)}}^\top  \rangle \\
    &\quad
    -8 \lambda_1 \big\langle  D U_Q^{(t)}, \Delta_{U_Q}^{(t)}   \big\rangle
    -4\lambda_2 \big\langle E B^{(t)},  \Delta_U^{(t)}  \big\rangle
    -4\lambda_2\big\langle E^\top U^{(t)},\Delta_B^{(t)} \big\rangle.
    \end{align*}
    For the term $-8 \lambda_2\langle B^{(t)}(\Delta_U^{(t)})^\top,\Delta_B^{(t)} {U^{(t)}}^\top\rangle$, we have
    \begin{align*}
        &-8 \lambda_2\langle B^{(t)}(\Delta_U^{(t)})^\top,\Delta_B^{(t)} {U^{(t)}}^\top\rangle\\
        =&-8 \lambda_2\langle B^{(t)}(\Delta_{U_Q}^{(t)})^\top,\Delta_B^{(t)} {U_Q^{(t)}}^\top\rangle  -8 \lambda_2\langle B^{(t)}(\Delta_{U_{Q^c}}^{(t)})^\top,\Delta_B^{(t)} {U_{Q^c}^{(t)}}^\top\rangle\\
        \leq&  8 \lambda_2 \|B^{(t)}(\Delta_{U_Q}^{(t)})^\top\|_{\fro}\|\Delta_B^{(t)} {U_Q^{(t)}}^\top\|_{\fro} + 8 \lambda_2 \|B^{(t)}(\Delta_{U_{Q^c}}^{(t)})^\top\|_{\fro}\|\Delta_B^{(t)} {U_{Q^c}^{(t)}}^\top\|_{\fro}\\
        =& 4\frac{\sqrt{2\lambda_2\|B^{(t)}(\Delta_{U_Q}^{(t)})^\top\|^2_{\fro}}}{\sqrt{s_{U_Q}^{(t)}}} 
        \frac{\sqrt{2\lambda_2}\|\Delta_B^{(t)} {U_Q^{(t)}}^\top\|_{\fro}}{\sqrt{s_B^{(t)}}}
        \sqrt{s_{U_Q}^{(t)}s_B^{(t)}}
        +4\frac{\sqrt{2\lambda_2\|\Delta_B^{(t)} {U_{Q^c}^{(t)}}^\top\|^2_{\fro}}}{\sqrt{s_{B}^{(t)}}} \sqrt{s_{B}^{(t)}}\sqrt{2\lambda_2} \|B^{(t)}(\Delta_{U_{Q^c}}^{(t)})^\top\|_{\fro}\\
        = & 4 \gamma_1^{(t)} \sqrt{s_{U_Q}^{(t)}s_B^{(t)}} + 4\gamma_2^{(t)}  \sqrt{s_{U_{Q^c}}^{(t)}s_B^{(t)}}\\
        \leq & 2\sqrt{(\gamma_1^{(t)} )^2+(\gamma_2^{(t)} )^2} ( s_{U_Q}^{(t)} +  s_B^{(t)} +s_{U_{Q^c}}^{(t)})\\
        \leq & 2\sqrt{\gamma} ( s_{U_Q}^{(t)} +  s_B^{(t)} +s_{U_{Q^c}}^{(t)}),
    \end{align*}
    where  \[ \gamma_1^{(t)}=\frac{\sqrt{2\lambda_2\|B^{(t)}(\Delta_{U_Q}^{(t)})^\top\|^2_{\fro}}}{\sqrt{s_{U_Q}^{(t)}}} 
        \frac{\sqrt{2\lambda_2}\|\Delta_B^{(t)} {U_Q^{(t)}}^\top\|_{\fro}}{\sqrt{s_B^{(t)}}} ,\quad \gamma_2^{(t)} =\frac{\sqrt{2\lambda_2\|\Delta_B^{(t)} {U_{Q^c}^{(t)}}^\top\|^2_{\fro}}}{\sqrt{s_{B}^{(t)}}}. \] The the penultimate inequality is obtained by first applying the Cauchy--Schwarz inequality,
    \[
    \gamma_1 \sqrt{zx}+\gamma_2 \sqrt{zy}
    \le
    \sqrt{\gamma_1^2+\gamma_2^2}\,\sqrt{z(x+y)},
    \]
    and then the AM--GM inequality,
    \[
    \sqrt{z(x+y)}\le \frac{x+y+z}{2}.
    \]
    And the last inequality follows from \(\sqrt{\gamma_1^2+\gamma_2^2}\le \sqrt\gamma<1\) with $\lambda_1+\lambda_2=1,\;\lambda_1>\underline \lambda$ by \Cref{lem:ratio-U}.
where \(\lambda_1+\lambda_2=1\) and
\(\lambda_1\ge \underline\lambda>0\), then
\[
(\gamma_1^{(t)})^2+(\gamma_2^{(t)})^2\le \gamma<1,
\]
where
\[
\gamma
:=
1-(1-\bar\gamma_{\underline\lambda})\underline b,
\qquad
\bar\gamma_{\underline\lambda}
:=
\frac{(1-\underline\lambda)\overline{\sigma}_B^2}
{2\underline\lambda\,\underline{\sigma}_{U_Q}^2+
(1-\underline\lambda)\overline{\sigma}_B^2},\qquad \underline b :=  \frac{\underline{\sigma}_{U_Q}^2}{\underline{\sigma}_{U_Q}^2+\overline{\sigma}_{U_{Q^c}}^2} < 1.
\]
    
    For the noise term, we have
      \begin{align*}
     & -8 \lambda_1 \big\langle  D U_Q^{(t)}, \Delta_{U_Q}^{(t)}   \big\rangle
    -4\lambda_2 \big\langle E B^{(t)},  \Delta_U^{(t)}  \big\rangle
    -4\lambda_2\big\langle E^\top U^{(t)},\Delta_B^{(t)} \big\rangle\\ = & - \big\langle  8 \lambda_1D U_Q^{(t)}+4\lambda_2 E_Q B^{(t)}, \Delta_{U_Q}^{(t)}   \big\rangle
    -4\lambda_2 \big\langle E_{Q^c} B^{(t)},  \Delta_{U_{Q^c}}^{(t)}  \big\rangle
    -4\lambda_2\big\langle E^\top U^{(t)},\Delta_B^{(t)} \big\rangle\\= & - \big\langle  (8 \lambda_1D U_Q^{(t)}+4\lambda_2 E_Q B^{(t)})(S_{U_Q}^{(t)})^{-1}, \Delta_{U_Q}^{(t)}S_{U_Q}^{(t)}   \big\rangle
    -4\lambda_2 \big\langle E_{Q^c} B^{(t)} (S_{U_{Q^c}}^{(t)})^{-1},  \Delta_{U_{Q^c}}^{(t)} S_{U_{Q^c}}^{(t)} \big\rangle
    \\ & -4\lambda_2 E^\top U^{(t)}(S_{B}^{(t)})^{-1},\Delta_B^{(t)}S_{B}^{(t)} \big\rangle\\ \le &  \|  (8 \lambda_1D U_Q^{(t)}+4\lambda_2 E_Q B^{(t)})(S_{U_Q}^{(t)})^{-1}\|_{\oper}\| \Delta_{U_Q}^{(t)}S_{U_Q}^{(t)}   \|_{*}
    +\|4\lambda_2 \big\langle E_{Q^c} B^{(t)} (S_{U_{Q^c}}^{(t)})^{-1}\|_{\oper}\|  \Delta_{U_{Q^c}}^{(t)} S_{U_{Q^c}}^{(t)} \|_*
    \\ &+\|4\lambda_2 E^\top U^{(t)}(S_{B}^{(t)})^{-1}\|_{\oper}\|\Delta_B^{(t)}S_{B}^{(t)} \|_*\\ \le &  \sqrt{k}\|  (8 \lambda_1D U_Q^{(t)}+4\lambda_2 E_Q B^{(t)})(S_{U_Q}^{(t)})^{-1}\|_{\oper}\| \Delta_{U_Q}^{(t)}S_{U_Q}^{(t)}   \|_{\fro}
    +\sqrt{k}\|4\lambda_2 \big\langle E_{Q^c} B^{(t)} (S_{U_{Q^c}}^{(t)})^{-1}\|_{\oper}\|  \Delta_{U_{Q^c}}^{(t)} S_{U_{Q^c}}^{(t)} \|_{\fro}
    \\ &+\sqrt{k}\|4\lambda_2 E^\top U^{(t)}(S_{B}^{(t)})^{-1}\|_{\oper}\|\Delta_B^{(t)}S_{B}^{(t)} \|_{\fro}\\ 
        \le & \frac{k}{4\BC_1}\|  (8 \lambda_1D U_Q^{(t)}+4\lambda_2 E_Q B^{(t)})(S_{U_Q}^{(t)})^{-1}\|_{\oper}^2+\BC_1\| \Delta_{U_Q}^{(t)}S_{U_Q}^{(t)}\|^2_{\fro}+\frac{k}{4\BC_2}\|4\lambda_2  E_{Q^c} B^{(t)} (S_{U_{Q^c}}^{(t)})^{-1}\|_{\oper}^2\\ &+\BC_2\|  \Delta_{U_{Q^c}}^{(t)} S_{U_{Q^c}}^{(t)}\|_{\fro}^2 +\frac{k}{4\BC_3}\|4\lambda_2 E^\top U^{(t)}(S_{B}^{(t)})^{-1}\|_{\oper}^2+\BC_3\|\Delta_B^{(t)}S_{B}^{(t)}\|_{\fro}^2, 
    \end{align*}
     where the first inequality is from $\langle A,B \rangle \leq \|A\|_{\oper}\|B\|_{*}$, and the second inequality is from $\text{rank}(AB)\leq \min\{\text{rank}(A),\text{rank}(B)\}$ and $\|A\|_{*}\leq \sqrt{\text{rank}(A)}\|A\|_{\fro}$, and last inequality follows from Young's inequality $ab\leq a^2/(4\epsilon) + \epsilon b^2$.
     Next we control the higher-order terms. By the Cauchy--Schwarz inequality and
    $\|AB^\top\|_{\fro}\le \|A\|_{\fro}\|B\|_{\fro}$,
    \begin{align*}
    & 8\lambda_1\big\langle
    \Delta_{U_Q}^{(t)}\Delta_{U_Q}^{(t)\top},\,
    \Delta_{U_Q}^{(t)}U_Q^{(t)\top}
    \big\rangle 
    +4\lambda_2
    \big\langle
    \Delta_B^{(t)}(\Delta_U^{(t)})^\top,\,
    \Delta_B^{(t)}U^{(t)\top}
    \big\rangle 
    +4\lambda_2
    \big\langle
    \Delta_U^{(t)}(\Delta_B^{(t)})^\top,\,
    \Delta_U^{(t)}B^{(t)\top}
    \big\rangle  \\
    \le&
    8\lambda_1
    \|\Delta_{U_Q}^{(t)}\Delta_{U_Q}^{(t)\top}\|_{\fro}
    \|\Delta_{U_Q}^{(t)}U_Q^{(t)\top}\|_{\fro}
    +
    4\lambda_2
    \|\Delta_B^{(t)}(\Delta_U^{(t)})^\top\|_{\fro}
    \|\Delta_B^{(t)}U^{(t)\top}\|_{\fro} 
    +
    4\lambda_2
    \|\Delta_U^{(t)}(\Delta_B^{(t)})^\top\|_{\fro}
    \|\Delta_U^{(t)}B^{(t)\top}\|_{\fro} \\
    \le &
    8\lambda_1
    \|\Delta_{U_Q}^{(t)}\|_{\fro}^2
    \|\Delta_{U_Q}^{(t)}U_Q^{(t)\top}\|_{\fro}
    +
    4\lambda_2
    \|\Delta_B^{(t)}\|_{\fro}\|\Delta_U^{(t)}\|_{\fro}
    \|\Delta_B^{(t)}U^{(t)\top}\|_{\fro} 
    +
    4\lambda_2
    \|\Delta_U^{(t)}\|_{\fro}\|\Delta_B^{(t)}\|_{\fro}
    \|\Delta_U^{(t)}B^{(t)\top}\|_{\fro} \\
    \le&
    \frac{16\lambda_1}{\BC_4}\|\Delta_{U_Q}^{(t)}\|_{\fro}^4
    +
    4\lambda_2\!\left(\frac{1}{\BC_5}+\frac{1}{\BC_6}\right)
    \|\Delta_U^{(t)}\|_{\fro}^2
    \|\Delta_B^{(t)}\|_{\fro}^2  
    +
    {\BC_4}\lambda_1\|\Delta_{U_Q}^{(t)}U_Q^{(t)\top}\|_{\fro}^2
    +
    {\BC_5}\lambda_2\|\Delta_B^{(t)}U^{(t)\top}\|_{\fro}^2
    +
    {\BC_6}\lambda_2\|\Delta_U^{(t)}B^{(t)\top}\|_{\fro}^2 ,
    \end{align*}
    where the last inequality by Young's inequality $ab\leq a^2/(4\epsilon) + \epsilon b^2$.
    
    Combining the above results yields that:
    \begin{align*}
     - 2 \sum_{v\in\mathcal V}  
     \langle \nabla_v^{(t)} \ell, \Delta_v^{(t)} \rangle
    \le &
    -2(1-\sqrt\gamma)\sum_{v\in\mathcal V}s_v^{(t)} 
    +\frac{k}{\BC_1}\|  ( 4\lambda_1D U_Q^{(t)}+2\lambda_2 E_Q B^{(t)})(S_{U_Q}^{(t)})^{-1}\|_{\oper}^2\\ & +\frac{4k}{\BC_2}\|\lambda_2  E_{Q^c} B^{(t)} (S_{U_{Q^c}}^{(t)})^{-1}\|_{\oper}^2 +\frac{4k}{\BC_3}\|\lambda_2 E^\top U^{(t)}(S_{B}^{(t)})^{-1}\|_{\oper}^2 \\ &+(4\BC_1+\BC_4)\lambda_1
    \|\Delta_{U_Q}^{(t)}U_Q^{(t)\top}\|_{\fro}^2 +\BC_6\lambda_2
    \|\Delta_U^{(t)}B^{(t)\top}\|_{\fro}^2 \\
    &+2\BC_1\lambda_2
    \|\Delta_{U_Q}^{(t)}B^{(t)\top}\|_{\fro}^2 +2\BC_2\lambda_2
    \|\Delta_{U_{Q^c}}^{(t)}B^{(t)\top}\|_{\fro}^2 
    +(2\BC_3+\BC_5)\lambda_2
    \|\Delta_B^{(t)}U^{(t)\top}\|_{\fro}^2 \\
    &+\frac{16\lambda_1}{\BC_4}\|\Delta_{U_Q}^{(t)}\|_{\fro}^4 +4\lambda_2\!\left(\frac{1}{\BC_5}+\frac{1}{\BC_6}\right)
    \|\Delta_U^{(t)}\|_{\fro}^2
    \|\Delta_B^{(t)}\|_{\fro}^2 .
    \end{align*}

    This completes the proof.
\end{proof}

\begin{lemma}[Upper bounds on the gradient norms]\label{lem:gradient}
For \Cref{alg:joint-UB-nGD}, the gradient norms satisfy that:
\begin{enumerate}[(i)]
    \item \label{nabla:UQ} For the update of $U_Q$, we have \begin{align*}
        &\| \nabla_{U_Q}^{(t)}\ell \; (S_{U_Q}^{(t)})^{-1}\|_{\fro}^2 \\ \leq  &\frac{1}{{\,2\lambda_1\sigma_{\min}^2(U_Q^{(t)})+\lambda_2\sigma_{\min}^2(B^{(t)})\,}} \big\{24\lambda_1^2\|U_Q^{(t)}\|_{\oper}^2\big(
           \big\|U_Q^{(t)}{\Delta_{U_Q}^{(t)}}^\top\big\|_{\fro}^2
       + \big\|\Delta_{U_Q}^{(t)}\bigl(U_Q^{(t)}\bigr)^\top \big\|_{\fro}^2 +\big\| \Delta_{U_Q}^{(t)}{\Delta_{U_Q}^{(t)}}^\top\big\|_{\fro}^2 
           \big)\\ & \quad +6\lambda_2^2\|B^{(t)}\|_{\oper}^2\big(\|\Delta_{U_Q}^{(t)}{B^{(t)}}^{\top} \|_{\fro}^2 + \|U_Q^{(t)} (\Delta^{(t)}_{B})^{\top}\|_{\fro}^2 + \|\Delta^{(t)}_{U_Q} (\Delta^{(t)}_{B})^{\top} \|_{\fro}^2\big)\\ & \quad +16 \lambda_1^2k\|U_Q^{(t)}\|_{\oper}^2\| D \|_{\oper}^2+4\lambda_2^2k\|{B}^{(t)}\|_{\oper}^2\|E_Q\|_{\oper}^2\big\};
    \end{align*}
    \item \label{nabla:B}For the update of $B$, it holds that
     \begin{align*}
        \| \nabla_{B}^{(t)}\ell \;(S_{B}^{(t)})^{-1}\|_{\fro}^2  &\leq  \frac{6\lambda_2  \|U^{(t)}\|_{\oper}^2 }{{\sigma_{\min}^2(U^{(t)})\,}}
        \left\{2\| \Delta_U^{(t)} {B^{(t)}}^\top\|_{\fro}^2+ 2\| \Delta_U^{(t)} ({\Delta_B^{(t)}})^\top\|_{\fro}^2+  \| U^{(t)}(\Delta_{B}^{(t)})^\top\|_{\fro}^2 + k\| E\|_{\oper}^2 \right\};
    \end{align*}
    \item \label{nabla:UQc} For the update of $U_{Q^c}$, it holds that\begin{align*}
        \| \nabla_{U_{Q^c}}^{(t)}\ell \; (S_{U_{Q^c}}^{(t)})^{-1} \|_{\fro}^2 &\leq \frac{6\lambda_2  \|B^{(t)}\|_{\oper}^2 }{{\sigma_{\min}^2(B^{(t)})\,}}
        \left\{2\| \Delta_{U_{Q^c}}^{(t)} {B^{(t)}}^\top\|_{\fro}^2+ 2\| \Delta_{U_{Q^c}}^{(t)} ({\Delta_B^{(t)}})^\top\|_{\fro}^2+  \| U_{Q^c}^{(t)}(\Delta_{B}^{(t)})^\top\|_{\fro}^2 +k \| E_{Q^c}\|_{\oper}^2 \right\} .
    \end{align*}
\end{enumerate}
\end{lemma}
\begin{proof}[Proof of \Cref{lem:gradient}]
    We split the proof into three parts.
     \paragraph{Part 1: Norm for $U_Q$.}
     First, we prove \Cref{lem:gradient} related to $\|  \nabla_{U_Q}^{(t)}\ell \; (S_{U_Q}^{(t)})^{-1} \|_{\fro}^2$. Note that
     \begin{align*}
      \|  \nabla_{U_Q}^{(t)}\ell  \; (S_{U_Q}^{(t)})^{-1} \|_{\fro}^2\leq  3 \| I_1  \; (S_{U_Q}^{(t)})^{-1} \|_{\fro}^2+3 \| I_2  \; (S_{U_Q}^{(t)})^{-1} \|_{\fro}^2+3 \| I_3  \; (S_{U_Q}^{(t)})^{-1} \|_{\fro}^2,
     \end{align*}
     where the inequality follows from $\|A_1+A_2+A_3\|^2_{\fro}\leq 3\|A_1\|^2_{\fro}+3\|A_2\|^2_{\fro}+3\|A_3\|^2_{\fro}$, $I_1$, $I_2$ and $I_3$ are defined in \eqref{eq:I1I2I3}, and $S_{U_Q}^{(t)} := \bigl(G_{U_Q}^{(t)}\bigr)^{1/2}$ with $ G_{U_Q}^{(t)} := 4\lambda_1 U_Q^{(t)\top} U_Q^{(t)} + 2\lambda_2 B^{(t)\top} B^{(t)}$.
     For the term $\|{I}_1  \; (S_{U_Q}^{(t)})^{-1}\|_{\fro}^2$,
    \begin{align*}
        \|{I}_1 \; (S_{U_Q}^{(t)})^{-1}\|_{\fro}^2 
        &= 16\lambda_1^2 \big\|(U_Q^{(t)} (U_Q^{(t)})^{\top} -U_QU_Q^\top) U_Q^{(t)}  \; (S_{U_Q}^{(t)})^{-1}\big\|_{\fro}^2 \\
         &\leq 16\lambda_1^2 \big\|(U_Q^{(t)}{\Delta_{U_Q}^{(t)}}^\top
       + \Delta_{U_Q}^{(t)}\bigl(U_Q^{(t)}\bigr)^\top  - \Delta_{U_Q}^{(t)}{\Delta_{U_Q}^{(t)}}^\top)U_Q^{(t)}(S_{U_Q}^{(t)})^{-1}\big\|_{\fro}^2 \\&\leq 16\lambda_1^2\|U_Q^{(t)}\|_{\oper}^2\|(S_{U_Q}^{(t)})^{-1}\|_{\oper}^2 \big\|U_Q^{(t)}{\Delta_{U_Q}^{(t)}}^\top
       + \Delta_{U_Q}^{(t)}\bigl(U_Q^{(t)}\bigr)^\top  - \Delta_{U_Q}^{(t)}{\Delta_{U_Q}^{(t)}}^\top\big\|_{\fro}^2 \\
        &{\leq}  \frac{24\lambda_1^2\|U_Q^{(t)}\|_{\oper}^2}{{\,2\lambda_1\sigma_{\min}^2(U_Q^{(t)})+\lambda_2\sigma_{\min}^2(B^{(t)})\,}} \big\{
           \big\|U_Q^{(t)}{\Delta_{U_Q}^{(t)}}^\top\big\|_{\fro}^2
       + \big\|\Delta_{U_Q}^{(t)}\bigl(U_Q^{(t)}\bigr)^\top \big\|_{\fro}^2 +\big\| \Delta_{U_Q}^{(t)}{\Delta_{U_Q}^{(t)}}^\top\big\|_{\fro}^2 
           \big\}
    \end{align*}
    where the last inequality follows from $\|A_1+A_2+A_3\|^2_{\fro}\leq 3\|A_1\|^2_{\fro}+3\|A_2\|^2_{\fro}+3\|A_3\|^2_{\fro}$. Since $G_{U_Q}^{(t)}\succ 0$ and $\|(S_{U_Q}^{(t)})^{-1}\|_{\oper}
=1/\sigma_{\min}(S_{U_Q}^{(t)})=1/\sqrt{\sigma_{\min}(G_{U_Q}^{(t)})}$ and
$\lambda_{\min}(A+B)\ge \lambda_{\min}(A)+\lambda_{\min}(B)$ for $A,B\succeq 0$, we have
\[
\sigma_{\min}(G_{U_Q}^{(t)})
\ge 4\lambda_1\sigma_{\min}^2\!\bigl(U_Q^{(t)}\bigr)
+2\lambda_2\sigma_{\min}^2\!\bigl(B^{(t)}\bigr),
\]
and hence
\[
\|(S_{U_Q}^{(t)})^{-1}\|_{\oper}^2
\le \frac{1}{{\,4\lambda_1\sigma_{\min}^2(U_Q^{(t)})+2\lambda_2\sigma_{\min}^2(B^{(t)})\,}}.
\] In addition,
    \begin{align*}
        \|{I}_2\; (S_{U_Q}^{(t)})^{-1}\|_{\fro}^2 
        &= 4 \lambda_2^2 \| (U_Q^{(t)} {B^{(t)}}^{\top} - U_Q B^{\top} )  {B}^{(t)}\; (S_{U_Q}^{(t)})^{-1}\|_{\fro}^2\notag \\
        &\leq 4 \lambda_2^2\| (U_Q^{(t)} {B^{(t)}}^{\top} -U_Q R^{(t)}{B^{(t)}}^{\top}+U_Q R^{(t)}{B^{(t)}}^{\top} - U_Q B^{\top} )  {B}^{(t)}\; (S_{U_Q}^{(t)})^{-1}\|_{\fro}^2\notag\\
        &{\leq}   \frac{6\lambda_2^2\|B^{(t)}\|_{\oper}^2}{{{\,2\lambda_1\sigma_{\min}^2(U_Q^{(t)})+\lambda_2\sigma_{\min}^2(B^{(t)})\,}}} \big(\|\Delta_{U_Q}^{(t)}{B^{(t)}}^{\top} \|_{\fro}^2 + \|U_Q^{(t)} (\Delta^{(t)}_{B})^{\top}) \|_{\fro}^2 + \|\Delta^{(t)}_{U_Q} (\Delta^{(t)}_{B})^{\top}) \|_{\fro}^2\big),
    \end{align*}
    where the last inequality follows from $\|A_1+A_2+A_3\|^2_{\fro}\leq 3(\|A_1\|^2_{\fro}+\|A_2\|^2_{\fro}+\|A_3\|^2_{\fro})$. And
 \begin{align*}
        \|{I}_3\; (S_{U_Q}^{(t)})^{-1}\|_{\fro}^2 
        &=  \| ( -(4\lambda_1 D U_Q^{(t)}
         +2\lambda_2 E_Q{B}^{(t)})\; (S_{U_Q}^{(t)})^{-1}\|_{\fro}^2\notag \\
        &\leq 32 \lambda_1^2\| D U_Q^{(t)}\; (S_{U_Q}^{(t)})^{-1}\|_{\fro}^2+8\lambda_2^2\|E_Q{B}^{(t)}\; (S_{U_Q}^{(t)})^{-1}\|_{\fro}^2\notag\\
        &{\leq}   \frac{1}{{{\,2\lambda_1\sigma_{\min}^2(U_Q^{(t)})+\lambda_2\sigma_{\min}^2(B^{(t)})\,}}} \big(16 \lambda_1^2k\|U_Q^{(t)}\|_{\oper}^2\| D \|_{\oper}^2+4\lambda_2^2k\|{B}^{(t)}\|_{\oper}^2\|E_Q\|_{\oper}^2\big).
    \end{align*}
    
    Putting the bounds for $\|I_1\; (S_{U_Q}^{(t)})^{-1}\|_{\fro}^2$, $\|I_2\; (S_{U_Q}^{(t)})^{-1}\|_{\fro}^2$ and $\|I_3\; (S_{U_Q}^{(t)})^{-1}\|_{\fro}^2$ together yields
    \begin{align*}
         \|  \nabla_{U_Q}^{(t)}\ell \; (S_{U_Q}^{(t)})^{-1}\|_{\fro}^2 \leq &   \frac{1}{{\,2\lambda_1\sigma_{\min}^2(U_Q^{(t)})+\lambda_2\sigma_{\min}^2(B^{(t)})\,}} \big\{24\lambda_1^2\|U_Q^{(t)}\|_{\oper}^2\big(
           \big\|U_Q^{(t)}{\Delta_{U_Q}^{(t)}}^\top\big\|_{\fro}^2
       + \big\|\Delta_{U_Q}^{(t)}\bigl(U_Q^{(t)}\bigr)^\top \big\|_{\fro}^2 \\ & \quad+\big\| \Delta_{U_Q}^{(t)}{\Delta_{U_Q}^{(t)}}^\top\big\|_{\fro}^2 
           \big)\\ & \quad +6\lambda_2^2\|B^{(t)}\|_{\oper}^2\big(\|\Delta_{U_Q}^{(t)}{B^{(t)}}^{\top} \|_{\fro}^2 + \|U_Q^{(t)} (\Delta^{(t)}_{B})^{\top} \|_{\fro}^2 + \|\Delta^{(t)}_{U_Q} (\Delta^{(t)}_{B})^{\top} \|_{\fro}^2\big)\\ & \quad +16 \lambda_1^2k\|U_Q^{(t)}\|_{\oper}^2\| D \|_{\oper}^2+4\lambda_2^2k\|{B}^{(t)}\|_{\oper}^2\|E_Q\|_{\oper}^2\big\}.
     \end{align*}
     
      \paragraph{Part 2: Norm for $B$.}
    Analogously, by the triangle inequality, it holds that 
    \begin{align*}
        \|\nabla_{B}^{(t)}\ell \; (S_{B}^{(t)})^{-1}\|_{\fro}^2 =&  \| -2\lambda_2(W - U^{(t)}{B^{(t)}}^\top)^\top U^{(t)}  \; (S_{B}^{(t)})^{-1}\|_{\fro}^2\notag\\ 
        =&4  \lambda_2  \| (UB^\top+E - U^{(t)}{B^{(t)}}^\top)^\top U^{(t)} \; (S_{B}^{(t)})^{-1}\|_{\fro}^2\notag\\ 
        =&4\lambda_2\| \{(U-U^{(t)}{R^{(t)}}^\top)B^\top+U^{(t)}(B{R^{(t)}}- {B^{(t)}})^\top+E )\}^\top U^{(t)} \; (S_{B}^{(t)})^{-1}\|_{\fro}^2\notag\\
        \leq & \frac{6\lambda_2  \|U^{(t)}\|_{\oper}^2 }{{\sigma_{\min}^2(U^{(t)})\,}}
        \left\{2\| \Delta_U^{(t)} {B^{(t)}}^\top\|_{\fro}^2+ 2\| \Delta_U^{(t)} ({\Delta_B^{(t)}})^\top\|_{\fro}^2+  \| U^{(t)}(\Delta_{B}^{(t)})^\top\|_{\fro}^2 + k\| E\|_{\oper}^2 \right\} ,
      \end{align*}
      where in the last inequality we used $(a+b)^2 \le 2(a^2+b^2)$. 
    Moreover, since $G_B^{(t)}=2\lambda_2 U^{(t)\top}U^{(t)}$, it follows that
    $
    \sigma_{\min}(G_B^{(t)})
    =
    2\lambda_2\sigma_{\min}^2\!\bigl(U^{(t)}\bigr),
    $
    Consequently,
    \[
    \|(S_B^{(t)})^{-1}\|_{\oper}^2
    =
    \frac{1}{\sigma_{\min}(G_B^{(t)})}
    \le
    \frac{1}{2\lambda_2\,\sigma_{\min}^2\!\bigl(U^{(t)}\bigr)}.
    \]

    \paragraph{Part 3: Norm for $U_{Q^c}$.}
    For the $ \|\nabla_{{U_{Q^c}}}^{(t)}\ell\; (S_{U_{Q^c}}^{(t)})^{-1}\|_{\fro}^2$, we also have
    \begin{align*}
        \|\nabla_{U_{Q^c}}^{(t)}\ell\; (S_{U_{Q^c}}^{(t)})^{-1}\|_{\fro}^2  = &  \|2\lambda_2\big(U_{Q^c}^{(t)}{  B^{(t)}}^\top - W_{Q^c}\big)B^{(t)}\; (S_{U_{Q^c}}^{(t)})^{-1}\|_{\fro}^2 \notag\\
        \leq &4\lambda_2^2 
       \|\big(U_{Q^c}^{(t)}{  B^{(t)}}^\top-U_{Q^c}^{(t)}{ ( BR^{(t)})}^\top +U_{Q^c}^{(t)}{ ( BR^{(t)})}^\top- U_{Q^c}{  B}^\top-E_{Q^c}\big)B^{(t)}\; (S_{U_{Q^c}}^{(t)})^{-1}\|_{\fro}^2\notag\\
        \leq & \frac{6\lambda_2  \|B^{(t)}\|_{\oper}^2 }{{\sigma_{\min}^2(B^{(t)})\,}}
        \left\{2\| \Delta_{U_{Q^c}}^{(t)} {B^{(t)}}^\top\|_{\fro}^2+ 2\| \Delta_{U_{Q^c}}^{(t)} ({\Delta_B^{(t)}})^\top\|_{\fro}^2+  \| U_{Q^c}^{(t)}(\Delta_{B}^{(t)})^\top\|_{\fro}^2 +k \| E_{Q^c}\|_{\oper}^2 \right\} .
    \end{align*}
    where the last inequality follows from the fact that
    $G_{U_{Q^c}}^{(t)}=2\lambda_2 B^{(t)\top}B^{(t)}$.
    In particular,
    $\sigma_{\min}(G_{U_{Q^c}}^{(t)})=2\lambda_2\sigma_{\min}^2\!\bigl(B^{(t)}\bigr)$,
    which implies
    \[
    \|S_{U_{Q^c}}^{(t)-1}\|_{\oper}^2
    \le
    \frac{1}{2\lambda_2\sigma_{\min}^2\!\bigl(B^{(t)}\bigr)}.
    \]
    This completes the proof.
\end{proof}

\begin{lemma}[Normalized errors]\label{lem:DeltaS}
    For \Cref{alg:joint-UB-nGD}, the consecutive normalized errors satisfy:
    
      \begin{align*}
    \|\Delta_{U_Q}^{(t)}(S_{U_Q}^{(t+1)}-S_{U_Q}^{(t)})\|_{\fro}^2
    &\le
    \frac{\|\Delta_{U_Q}^{(t)}\|_{\fro}^2}{\gamma_{U_Q}^{(t)2}}
    \Big[
    32\lambda_1^2\eta^2
    (\|U_Q^{(t+1)}\|_{\oper}+\|U_Q^{(t)}\|_{\oper})^2
    \|\nabla_{U_Q}^{(t)}\ell\|_{\fro}^2 \\
    &\qquad
    +
    8\lambda_2^2\eta^2
    (\|B^{(t+1)}\|_{\oper}+\|B^{(t)}\|_{\oper})^2
    \|\nabla_{B}^{(t)}\ell\|_{\fro}^2
    \Big],\\
    \|\Delta_{B}^{(t)}(S_{B}^{(t+1)}-S_{B}^{(t)})\|_{\fro}^2
    &\le
    \frac{\|\Delta_{B}^{(t)}\|_{\fro}^2}{\gamma_{B}^{(t)2}}
    4\lambda_2^2\eta^2
    (\|U^{(t+1)}\|_{\oper}+\|U^{(t)}\|_{\oper})^2
    \|\nabla_{U}^{(t)}\ell\|_{\fro}^2,\\
    \|\Delta_{U_{Q^c}}^{(t)}(S_{U_{Q^c}}^{(t+1)}-S_{U_{Q^c}}^{(t)})\|_{\fro}^2
    &\le
    \frac{\|\Delta_{U_{Q^c}}^{(t)}\|_{\fro}^2}{\gamma_{U_{Q^c}}^{(t)2}}
    4\lambda_2^2\eta^2
    (\|B^{(t+1)}\|_{\oper}+\|B^{(t)}\|_{\oper})^2
    \|\nabla_{B}^{(t)}\ell\|_{\fro}^2 ,
    \end{align*}
    where $\gamma_v^{(t)} = \sigma_{\min}(S_v^{(t+1)})+\sigma_{\min}(S_v^{(t)})$ for $v\in\mathcal V$. 
\end{lemma}
\begin{proof}[Proof of \Cref{lem:DeltaS}]
Using $\|AB\|_{\fro}\le \|A\|_{\fro}\|B\|_{\oper}$, we have
\[
\|\Delta_v^{(t)}(S_v^{(t+1)}-S_v^{(t)})\|_{\fro}^2
\le
\|\Delta_v^{(t)}\|_{\fro}^2
\|S_v^{(t+1)}-S_v^{(t)}\|_{\oper}^2.
\]
Moreover, for positive definite matrices $A$ and $B$,
\[
\|A^{1/2}-B^{1/2}\|_{\oper}
\le
\frac{\|A-B\|_{\oper}}
{\sqrt{\lambda_{\min}(A)}+\sqrt{\lambda_{\min}(B)}}.
\]
Applying this inequality with $A=G_v^{(t+1)}$ and $B=G_v^{(t)}$, and using
$\|M\|_{\oper}\le \|M\|_{\fro}$ for any matrix $M$, yields
\[
\|S_v^{(t+1)}-S_v^{(t)}\|_{\oper}
\le
\frac{\|G_v^{(t+1)}-G_v^{(t)}\|_{\oper}}
{\sigma_{\min}(S_v^{(t+1)})+\sigma_{\min}(S_v^{(t)})}
\le
\frac{\|G_v^{(t+1)}-G_v^{(t)}\|_{\fro}}
{\sigma_{\min}(S_v^{(t+1)})+\sigma_{\min}(S_v^{(t)})}.
\]
Hence
\begin{align}\label{eq:deltaS}
\|\Delta_v^{(t)}(S_v^{(t+1)}-S_v^{(t)})\|_{\fro}^2
\le
\frac{
\|\Delta_v^{(t)}\|_{\fro}^2\,
\|G_v^{(t+1)}-G_v^{(t)}\|_{\fro}^2
}{
\bigl(\sigma_{\min}(S_v^{(t+1)})+\sigma_{\min}(S_v^{(t)})\bigr)^2
}.
\end{align}

For each $v\in\mathcal V$, using
\[
A^\top A-B^\top B=A^\top(A-B)+(A-B)^\top B,
\]
we obtain
\[
\|A^\top A-B^\top B\|_{\fro}^2
\le
(\|A\|_{\oper}+\|B\|_{\oper})^2\|A-B\|_{\fro}^2.
\]
Hence, $\|G_v^{(t+1)}-G_v^{(t)}\|_{\fro}^2$
can be controlled by the corresponding iterate differences. More precisely,
\begin{align}
\|G_{U_Q}^{(t+1)}-G_{U_Q}^{(t)}\|_{\fro}^2
&\le
2\Big\|
4\lambda_1\bigl(U_Q^{(t+1)\top}U_Q^{(t+1)}-U_Q^{(t)\top}U_Q^{(t)}\bigr)
\Big\|_{\fro}^2\notag \\
&\qquad
+2\Big\|
2\lambda_2\bigl(B^{(t+1)\top}B^{(t+1)}-B^{(t)\top}B^{(t)}\bigr)
\Big\|_{\fro}^2\notag \\
&\le
32\lambda_1^2(\|U_Q^{(t+1)}\|_{\oper}+\|U_Q^{(t)}\|_{\oper})^2
\|U_Q^{(t+1)}-U_Q^{(t)}\|_{\fro}^2\notag\\
&\qquad
+8\lambda_2^2(\|B^{(t+1)}\|_{\oper}+\|B^{(t)}\|_{\oper})^2
\|B^{(t+1)}-B^{(t)}\|_{\fro}^2\notag\\
&\le
32\lambda_1^2\eta^2(\|U_Q^{(t+1)}\|_{\oper}+\|U_Q^{(t)}\|_{\oper})^2
\|\nabla_{U_Q}^{(t)}\ell\|_{\fro}^2\notag\\
&\qquad
+8\lambda_2^2\eta^2(\|B^{(t+1)}\|_{\oper}+\|B^{(t)}\|_{\oper})^2
\|\nabla_B^{(t)}\ell\|_{\fro}^2,\label{eq:dGUQ}\\[0.5em]
\|G_{B}^{(t+1)}-G_{B}^{(t)}\|_{\fro}^2&\le
4\lambda_2^2(\|U^{(t+1)}\|_{\oper}+\|U^{(t)}\|_{\oper})^2
\|U^{(t+1)}-U^{(t)}\|_{\fro}^2\notag\\
&\le
4\lambda_2^2\eta^2(\|U^{(t+1)}\|_{\oper}+\|U^{(t)}\|_{\oper})^2
\|\nabla_U^{(t)}\ell\|_{\fro}^2,\label{eq:dGB}\\[0.5em]
\|G_{U_{Q^c}}^{(t+1)}-G_{U_{Q^c}}^{(t)}\|_{\fro}^2
&\le
4\lambda_2^2(\|B^{(t+1)}\|_{\oper}+\|B^{(t)}\|_{\oper})^2
\|B^{(t+1)}-B^{(t)}\|_{\fro}^2\notag\\&\le
4\lambda_2^2\eta^2(\|B^{(t+1)}\|_{\oper}+\|B^{(t)}\|_{\oper})^2
\|\nabla_B^{(t)}\ell\|_{\fro}^2\label{eq:dGQc}.
\end{align}
Define
\[
\gamma_v^{(t)} := \sigma_{\min}(S_v^{(t+1)})+\sigma_{\min}(S_v^{(t)}).
\]
Combining the above bounds with the generic inequality yields
\begin{align*}
\|\Delta_{U_Q}^{(t)}(S_{U_Q}^{(t+1)}-S_{U_Q}^{(t)})\|_{\fro}^2
&\le
\frac{\|\Delta_{U_Q}^{(t)}\|_{\fro}^2}{\gamma_{U_Q}^{(t)2}}
\Big[
32\lambda_1^2\eta^2
(\|U_Q^{(t+1)}\|_{\oper}+\|U_Q^{(t)}\|_{\oper})^2
\|\nabla_{U_Q}^{(t)}\ell\|_{\fro}^2 \\
&\qquad
+
8\lambda_2^2\eta^2
(\|B^{(t+1)}\|_{\oper}+\|B^{(t)}\|_{\oper})^2
\|\nabla_{B}^{(t)}\ell\|_{\fro}^2
\Big],\\
\|\Delta_{B}^{(t)}(S_{B}^{(t+1)}-S_{B}^{(t)})\|_{\fro}^2
&\le
\frac{\|\Delta_{B}^{(t)}\|_{\fro}^2}{\gamma_{B}^{(t)2}}
4\lambda_2^2\eta^2
(\|U^{(t+1)}\|_{\oper}+\|U^{(t)}\|_{\oper})^2
\|\nabla_{U}^{(t)}\ell\|_{\fro}^2,\\
\|\Delta_{U_{Q^c}}^{(t)}(S_{U_{Q^c}}^{(t+1)}-S_{U_{Q^c}}^{(t)})\|_{\fro}^2
&\le
\frac{\|\Delta_{U_{Q^c}}^{(t)}\|_{\fro}^2}{\gamma_{U_{Q^c}}^{(t)2}}
4\lambda_2^2\eta^2
(\|B^{(t+1)}\|_{\oper}+\|B^{(t)}\|_{\oper})^2
\|\nabla_{B}^{(t)}\ell\|_{\fro}^2 .
\end{align*}
This completes the proof.
\end{proof}

\begin{lemma}[Normalized gradients]\label{lem:grad-DeltaS}
    For \Cref{alg:joint-UB-nGD}, the consecutive normalized gradients satisfy:
    \begin{align*}
    \eta^2\big\|\nabla_{U_Q}^{(t)}\ell\,(S_{U_Q}^{(t)})^{-2}(S_{U_Q}^{(t+1)}-S_{U_Q}^{(t)})\big\|_{\fro}^2&\le
    \frac{\eta^4\|\nabla_{U_Q}^{(t)}\ell\|_{\fro}^2}
    {\sigma_{\min}(S_{U_Q}^{(t)})^4\,\gamma_{U_Q}^{(t)2}}
    \Big[
    32\lambda_1^2
    (\|U_Q^{(t+1)}\|_{\oper}+\|U_Q^{(t)}\|_{\oper})^2
    \|\nabla_{U_Q}^{(t)}\ell\|_{\fro}^2 \\
    &\qquad
    +
    8\lambda_2^2
    (\|B^{(t+1)}\|_{\oper}+\|B^{(t)}\|_{\oper})^2
    \|\nabla_{B}^{(t)}\ell\|_{\fro}^2
    \Big],\\[0.6em]
    \eta^2\big\|\nabla_{B}^{(t)}\ell\,(S_{B}^{(t)})^{-2}(S_{B}^{(t+1)}-S_{B}^{(t)})\big\|_{\fro}^2
    &\le
    \frac{\eta^4\|\nabla_{B}^{(t)}\ell\|_{\fro}^2}
    {\sigma_{\min}(S_{B}^{(t)})^4\,\gamma_{B}^{(t)2}}
    4\lambda_2^2
    (\|U^{(t+1)}\|_{\oper}+\|U^{(t)}\|_{\oper})^2
    \|\nabla_{U}^{(t)}\ell\|_{\fro}^2,\\[0.6em]
    \eta^2\big\|\nabla_{U_{Q^c}}^{(t)}\ell\,(S_{U_{Q^c}}^{(t)})^{-2}(S_{U_{Q^c}}^{(t+1)}-S_{U_{Q^c}}^{(t)})\big\|_{\fro}^2
    &\le
    \frac{\eta^4\|\nabla_{U_{Q^c}}^{(t)}\ell\|_{\fro}^2}
    {\sigma_{\min}(S_{U_{Q^c}}^{(t)})^4\,\gamma_{U_{Q^c}}^{(t)2}}
    4\lambda_2^2
    (\|B^{(t+1)}\|_{\oper}+\|B^{(t)}\|_{\oper})^2
    \|\nabla_{B}^{(t)}\ell\|_{\fro}^2 ,
    \end{align*}
    where
    \(
    \gamma_v^{(t)}=\sigma_{\min}(S_v^{(t+1)})+\sigma_{\min}(S_v^{(t)})
    \) for $v\in\mathcal V$.
\end{lemma}
\begin{proof}[Proof of \Cref{lem:grad-DeltaS}]
Using $\|ABC\|_{\fro}\le \|A\|_{\fro}\|B\|_{\oper}\|C\|_{\oper}$, we have
\[
\big\|\nabla_v^{(t)}\ell\,(S_v^{(t)})^{-2}(S_v^{(t+1)}-S_v^{(t)})\big\|_{\fro}^2
\le
\|\nabla_v^{(t)}\ell\|_{\fro}^2
\|(S_v^{(t)})^{-2}\|_{\oper}^2
\|S_v^{(t+1)}-S_v^{(t)}\|_{\oper}^2.
\]
Moreover,
\[
\|(S_v^{(t)})^{-2}\|_{\oper}
=
\|(S_v^{(t)})^{-1}\|_{\oper}^2
=
\frac{1}{\sigma_{\min}(S_v^{(t)})^2}.
\]
Hence, 
multiply both sides by $\eta^2$ gives
\[
\eta^2
\big\|\nabla_v^{(t)}\ell\,(S_v^{(t)})^{-2}(S_v^{(t+1)}-S_v^{(t)})\big\|_{\fro}^2
\le
\frac{\eta^2\|\nabla_v^{(t)}\ell\|_{\fro}^2}
{\sigma_{\min}(S_v^{(t)})^4}
\|S_v^{(t+1)}-S_v^{(t)}\|_{\oper}^2.
\]
By the proof of \Cref{lem:DeltaS},
\[
\|S_v^{(t+1)}-S_v^{(t)}\|_{\oper}^2
\le
\frac{\|G_v^{(t+1)}-G_v^{(t)}\|_{\fro}^2}{\gamma_v^{(t)2}},
\qquad
\gamma_v^{(t)}:=\sigma_{\min}(S_v^{(t+1)})+\sigma_{\min}(S_v^{(t)}).
\]
Therefore,
\[
\eta^2
\big\|\nabla_v^{(t)}\ell\,(S_v^{(t)})^{-2}(S_v^{(t+1)}-S_v^{(t)})\big\|_{\fro}^2
\le
\frac{\eta^2\|\nabla_v^{(t)}\ell\|_{\fro}^2}
{\sigma_{\min}(S_v^{(t)})^4\,\gamma_v^{(t)2}}
\|G_v^{(t+1)}-G_v^{(t)}\|_{\fro}^2.
\]
Applying the corresponding bounds for $\|G_v^{(t+1)}-G_v^{(t)}\|_{\fro}^2$ established in the proof of \Cref{lem:DeltaS}, we obtain
\begin{align*}
&\eta^2\big\|\nabla_{U_Q}^{(t)}\ell\,(S_{U_Q}^{(t)})^{-2}(S_{U_Q}^{(t+1)}-S_{U_Q}^{(t)})\big\|_{\fro}^2\\
&\le
\frac{\eta^4\|\nabla_{U_Q}^{(t)}\ell\|_{\fro}^2}
{\sigma_{\min}(S_{U_Q}^{(t)})^4\,\gamma_{U_Q}^{(t)2}}
\Big[
32\lambda_1^2
(\|U_Q^{(t+1)}\|_{\oper}+\|U_Q^{(t)}\|_{\oper})^2
\|\nabla_{U_Q}^{(t)}\ell\|_{\fro}^2 \\
&\qquad
+
8\lambda_2^2
(\|B^{(t+1)}\|_{\oper}+\|B^{(t)}\|_{\oper})^2
\|\nabla_{B}^{(t)}\ell\|_{\fro}^2
\Big],\\
&\eta^2\big\|\nabla_{B}^{(t)}\ell\,(S_{B}^{(t)})^{-2}(S_{B}^{(t+1)}-S_{B}^{(t)})\big\|_{\fro}^2\\
&\le
\frac{\eta^4\|\nabla_{B}^{(t)}\ell\|_{\fro}^2}
{\sigma_{\min}(S_{B}^{(t)})^4\,\gamma_{B}^{(t)2}}
4\lambda_2^2
(\|U^{(t+1)}\|_{\oper}+\|U^{(t)}\|_{\oper})^2
\|\nabla_{U}^{(t)}\ell\|_{\fro}^2,\\
&\eta^2\big\|\nabla_{U_{Q^c}}^{(t)}\ell\,(S_{U_{Q^c}}^{(t)})^{-2}(S_{U_{Q^c}}^{(t+1)}-S_{U_{Q^c}}^{(t)})\big\|_{\fro}^2\\
&\le
\frac{\eta^4\|\nabla_{U_{Q^c}}^{(t)}\ell\|_{\fro}^2}
{\sigma_{\min}(S_{U_{Q^c}}^{(t)})^4\,\gamma_{U_{Q^c}}^{(t)2}}
4\lambda_2^2
(\|B^{(t+1)}\|_{\oper}+\|B^{(t)}\|_{\oper})^2
\|\nabla_{B}^{(t)}\ell\|_{\fro}^2 .
\end{align*}
This completes the proof.
\end{proof}

\begin{lemma}[Cross terms]\label{lem:cross-term-blockwise}
    For \Cref{alg:joint-UB-nGD}, the cross terms satisfy:
    \begin{align*}
    & 2\Big\langle (\Delta_{U_Q}^{(t)}-\eta\nabla_{U_Q}^{(t)}\ell(S_{U_Q}^{(t)})^{-2} )S_{U_Q}^{(t)},
    (\Delta_{U_Q}^{(t)}-\eta\nabla_{U_Q}^{(t)}\ell(S_{U_Q}^{(t)})^{-2} )(S_{U_Q}^{(t+1)}-S_{U_Q}^{(t)})\Big\rangle \\
    \le &
    \Bigg\{
    2\|S_{U_Q}^{(t)}\|_{\oper}\|\Delta_{U_Q}^{(t)}\|_{\fro}^2
    +2\eta^2\|\nabla_{U_Q}^{(t)}\ell\|_{\fro}^2
    \|(S_{U_Q}^{(t)})^{-1}\|_{\oper}\|(S_{U_Q}^{(t)})^{-2}\|_{\oper} \\
    &\qquad
    +2\eta\|\Delta_{U_Q}^{(t)}\|_{\fro}\Big(
    \|S_{U_Q}^{(t)}\|_{\oper}\|(S_{U_Q}^{(t)})^{-2}\|_{\oper}\|\nabla_{U_Q}^{(t)}\ell\|_{\fro}
    +\|\nabla_{U_Q}^{(t)}\ell(S_{U_Q}^{(t)})^{-1}\|_{\fro}
    \Big)
    \Bigg\} \\
    &\qquad\times
    \frac{\eta}{\gamma_{U_Q}^{(t)}}
    \Bigg[
    32\lambda_1^2(\|U_Q^{(t+1)}\|_{\oper}+\|U_Q^{(t)}\|_{\oper})^2
    \|\nabla_{U_Q}^{(t)}\ell\|_{\fro}^2 +
    8\lambda_2^2(\|B^{(t+1)}\|_{\oper}+\|B^{(t)}\|_{\oper})^2
    \|\nabla_B^{(t)}\ell\|_{\fro}^2
    \Bigg]^{1/2}, \\
    & 2\Big\langle (\Delta_{B}^{(t)}-\eta\nabla_{B}^{(t)}\ell(S_{B}^{(t)})^{-2} )S_{B}^{(t)},
    (\Delta_{B}^{(t)}-\eta\nabla_{B}^{(t)}\ell(S_{B}^{(t)})^{-2} )(S_{B}^{(t+1)}-S_{B}^{(t)})\Big\rangle \\
    \le&
    \Bigg\{
    2\|S_{B}^{(t)}\|_{\oper}\|\Delta_{B}^{(t)}\|_{\fro}^2
    +2\eta^2\|\nabla_{B}^{(t)}\ell\|_{\fro}^2
    \|(S_{B}^{(t)})^{-1}\|_{\oper}\|(S_{B}^{(t)})^{-2}\|_{\oper} \\
    &\qquad
    +2\eta\|\Delta_B^{(t)}\|_{\fro}\Big(
    \|S_{B}^{(t)}\|_{\oper}\|(S_{B}^{(t)})^{-2}\|_{\oper}\|\nabla_{B}^{(t)}\ell\|_{\fro}
    +\|\nabla_{B}^{(t)}\ell(S_{B}^{(t)})^{-1}\|_{\fro}
    \Big)
    \Bigg\} \\
    &\qquad\times
    \frac{2\lambda_2\eta}{\gamma_{B}^{(t)}}
    (\|U^{(t+1)}\|_{\oper}+\|U^{(t)}\|_{\oper})
    \|\nabla_U^{(t)}\ell\|_{\fro},\\
    & 2\Big\langle (\Delta_{U_{Q^c}}^{(t)}-\eta\nabla_{U_{Q^c}}^{(t)}\ell(S_{U_{Q^c}}^{(t)})^{-2} )S_{U_{Q^c}}^{(t)},
    (\Delta_{U_{Q^c}}^{(t)}-\eta\nabla_{U_{Q^c}}^{(t)}\ell(S_{U_{Q^c}}^{(t)})^{-2} )(S_{U_{Q^c}}^{(t+1)}-S_{U_{Q^c}}^{(t)})\Big\rangle \\
    \le&
    \Bigg\{
    2\|S_{U_{Q^c}}^{(t)}\|_{\oper}\|\Delta_{U_{Q^c}}^{(t)}\|_{\fro}^2
    +2\eta^2\|\nabla_{U_{Q^c}}^{(t)}\ell\|_{\fro}^2
    \|(S_{U_{Q^c}}^{(t)})^{-1}\|_{\oper}\|(S_{U_{Q^c}}^{(t)})^{-2}\|_{\oper} \\
    &\qquad
    +2\eta \|\Delta_{U_{Q^c}}^{(t)}\|_{\fro}\Big(
    \|S_{U_{Q^c}}^{(t)}\|_{\oper}\|(S_{U_{Q^c}}^{(t)})^{-2}\|_{\oper}\|\nabla_{U_{Q^c}}^{(t)}\ell\|_{\fro}
    +\|\nabla_{U_{Q^c}}^{(t)}\ell(S_{U_{Q^c}}^{(t)})^{-1}\|_{\fro}
    \Big)
    \Bigg\} \\
    &\qquad\times
    \frac{2\lambda_2\eta}{\gamma_{U_{Q^c}}^{(t)}}
    (\|B^{(t+1)}\|_{\oper}+\|B^{(t)}\|_{\oper})
    \|\nabla_B^{(t)}\ell\|_{\fro},
    \end{align*}
    where \(
    \gamma_v^{(t)} := \sigma_{\min}(S_v^{(t+1)})+\sigma_{\min}(S_v^{(t)})
    \) for $v\in\mathcal V$.
\end{lemma}

\begin{proof}
To facilitate the analysis, we decompose the term as follows.
\begin{align*}
   & 2\langle (\Delta_v^{(t)}-\eta\nabla_v\ell(S_{v}^{(t)})^{-2} )S_v^{(t)}, (\Delta_v^{(t)}-\eta\nabla_v\ell(S_{v}^{(t)})^{-2} )(S_v^{(t+1)}-S_v^{(t)})\rangle \\
   =& 2\langle \Delta_v^{(t)}S_v^{(t)}, \Delta_v^{(t)}(S_v^{(t+1)}-S_v^{(t)})\rangle +2\langle \eta\nabla_v\ell(S_{v}^{(t)})^{-1}, \eta\nabla_v\ell(S_{v}^{(t)})^{-2} (S_v^{(t+1)}-S_v^{(t)})\rangle \\ & \quad +2\langle  \Delta_v^{(t)}S_v^{(t)}, \eta\nabla_v\ell(S_{v}^{(t)})^{-2} (S_v^{(t+1)}-S_v^{(t)})\rangle +2\langle \eta\nabla_v\ell(S_{v}^{(t)})^{-1}, \Delta_v^{(t)}(S_v^{(t+1)}-S_v^{(t)})\rangle \\ \le & 2\| \Delta_v^{(t)}S_v^{(t)}\|_{\fro} \|\Delta_v^{(t)}(S_v^{(t+1)}-S_v^{(t)})\|_{\fro}+2\| \eta\nabla_v\ell(S_{v}^{(t)})^{-1}\|_{\fro}\| \eta\nabla_v\ell(S_{v}^{(t)})^{-2} (S_v^{(t+1)}-S_v^{(t)})\|_{\fro}\\ & \quad +2\|  \Delta_v^{(t)}S_v^{(t)}\|_{\fro}\|\eta\nabla_v\ell(S_{v}^{(t)})^{-2} (S_v^{(t+1)}-S_v^{(t)})\|_{\fro}+2\| \eta\nabla_v\ell(S_{v}^{(t)})^{-1} \|_{\fro}\|\Delta_v^{(t)}(S_v^{(t+1)}-S_v^{(t)})\|_{\fro}
\end{align*}

By  \eqref{eq:deltaS} of the proof of \Cref{lem:DeltaS},
\begin{align*}
&2\| \Delta_v^{(t)}S_v^{(t)}\|_{\fro} \|\Delta_v^{(t)}(S_v^{(t+1)}-S_v^{(t)})\|_{\fro} \\ \le&   
\frac{2\|S_v^{(t)}\|_{\oper}
\|\Delta_v^{(t)}\|_{\fro}^2\,
\|G_v^{(t+1)}-G_v^{(t)}\|_{\fro}
}{
\gamma_v^{(t)}
}.
\end{align*}
For positive definite $A$ and $B$,
\[
\|A^{1/2}-B^{1/2}\|_{\fro}
\le
\frac{\|A-B\|_{\fro}}
{\sqrt{\lambda_{\min}(A)}+\sqrt{\lambda_{\min}(B)}}.
\]
Applying this with $A=G_v^{(t+1)}$ and $B=G_v^{(t)}$, and noting that
$S_v^{(t)}=(G_v^{(t)})^{1/2}$, we obtain
\[
\|S_v^{(t+1)}-S_v^{(t)}\|_{\fro}
\le
\frac{\|G_v^{(t+1)}-G_v^{(t)}\|_{\fro}}
{\sigma_{\min}(S_v^{(t+1)})+\sigma_{\min}(S_v^{(t)})}
=
\frac{\|G_v^{(t+1)}-G_v^{(t)}\|_{\fro}}{\gamma_v^{(t)}}.
\]
Moreover,
\begin{align*}
    &2\| \eta\nabla_v\ell(S_{v}^{(t)})^{-1}\|_{\fro}\| \eta\nabla_v\ell(S_{v}^{(t)})^{-2} (S_v^{(t+1)}-S_v^{(t)})\|_{\fro}\\ \le & 2\eta^2\| \nabla_v\ell\|_{\fro}^2\|(S_{v}^{(t)})^{-1}\|_{\oper}\| S_{v}^{(t)})^{-2}\|_{\oper} \|S_v^{(t+1)}-S_v^{(t)}\|_{\fro}\\ \le & 2\eta^2\| \nabla_v\ell\|_{\fro}^2\|(S_{v}^{(t)})^{-1}\|_{\oper}\| (S_{v}^{(t)})^{-2}\|_{\oper} \frac{\|G_v^{(t+1)}-G_v^{(t)}\|_{\fro}}{\gamma_v^{(t)}}.
\end{align*}
Then,
\begin{align*}
    &2\|  \Delta_v^{(t)}S_v^{(t)}\|_{\fro}\|\eta\nabla_v\ell(S_{v}^{(t)})^{-2} (S_v^{(t+1)}-S_v^{(t)})\|_{\fro}+2\| \eta\nabla_v\ell(S_{v}^{(t)})^{-1} \|_{\fro}\|\Delta_v^{(t)}(S_v^{(t+1)}-S_v^{(t)})\|_{\fro}\\ \le & 2\eta\|S_v^{(t)}\|_{\oper}\|  \Delta_v^{(t)}\|_{\fro}\|\nabla_j\ell\|_{\fro}\|(S_{v}^{(t)})^{-2}\|_{\oper}\| (S_v^{(t+1)}-S_v^{(t)})\|_{\fro} +2\eta\| \nabla_v\ell(S_{v}^{(t)})^{-1} \|_{\fro}\|\Delta_v^{(t)}(S_v^{(t+1)}-S_v^{(t)})\|_{\fro}\\ \le & 2\eta\big(\|S_v^{(t)}\|_{\oper}\|(S_{v}^{(t)})^{-2}\|_{\oper}\|\nabla_v\ell\|_{\fro}+\| \nabla_v\ell(S_{v}^{(t)})^{-1} \|_{\fro}\big)\frac{\|\Delta_v^{(t)}\|_{\fro}\|G_v^{(t+1)}-G_v^{(t)}\|_{\fro}}{\gamma_v^{(t)}}.
\end{align*}
Hence, we conclude that the above inequalities
\begin{align*}
     & 2\langle (\Delta_v^{(t)}-\eta\nabla_v\ell(S_{v}^{(t)})^{-2} )S_v^{(t)}, (\Delta_v^{(t)}-\eta\nabla_v\ell(S_{v}^{(t)})^{-2} )(S_v^{(t+1)}-S_v^{(t)})\rangle \\ \le & \big\{2\|S_v^{(t)}\|_{\oper}
\|\Delta_v^{(t)}\|_{\fro}^2+2\eta^2\| \nabla_v\ell\|_{\fro}^2\|(S_{v}^{(t)})^{-1}\|_{\oper}\| (S_{v}^{(t)})^{-2}\|_{\oper}\\ & \quad +2\eta\|\Delta_v^{(t)}\|_{\fro}\big(\|S_v^{(t)}\|_{\oper}\|(S_{v}^{(t)})^{-2}\|_{\oper}\|\nabla_v\ell\|_{\fro}+\| \nabla_v\ell(S_{v}^{(t)})^{-1} \|_{\fro}\big)\big\} \frac{\|G_v^{(t+1)}-G_v^{(t)}\|_{\fro}}{\gamma_v^{(t)}}.
\end{align*}
By \eqref{eq:dGUQ}, \eqref{eq:dGB}, and \eqref{eq:dGQc}, we obtain the individual error
for the $U_Q$-block,
\begin{align*}
& 2\Big\langle (\Delta_{U_Q}^{(t)}-\eta\nabla_{U_Q}^{(t)}\ell(S_{U_Q}^{(t)})^{-2} )S_{U_Q}^{(t)},
(\Delta_{U_Q}^{(t)}-\eta\nabla_{U_Q}^{(t)}\ell(S_{U_Q}^{(t)})^{-2} )(S_{U_Q}^{(t+1)}-S_{U_Q}^{(t)})\Big\rangle \\
&\le
\Bigg\{
2\|S_{U_Q}^{(t)}\|_{\oper}\|\Delta_{U_Q}^{(t)}\|_{\fro}^2
+2\eta^2\|\nabla_{U_Q}^{(t)}\ell\|_{\fro}^2
\|(S_{U_Q}^{(t)})^{-1}\|_{\oper}\|(S_{U_Q}^{(t)})^{-2}\|_{\oper} \\
&\qquad
+2\eta\|\Delta_{U_Q}^{(t)}\|_{\fro}\Big(
\|S_{U_Q}^{(t)}\|_{\oper}\|(S_{U_Q}^{(t)})^{-2}\|_{\oper}\|\nabla_{U_Q}^{(t)}\ell\|_{\fro}
+\|\nabla_{U_Q}^{(t)}\ell(S_{U_Q}^{(t)})^{-1}\|_{\fro}
\Big)
\Bigg\} \\
&\qquad\times
\frac{\eta}{\gamma_{U_Q}^{(t)}}
\Bigg[
32\lambda_1^2(\|U_Q^{(t+1)}\|_{\oper}+\|U_Q^{(t)}\|_{\oper})^2
\|\nabla_{U_Q}^{(t)}\ell\|_{\fro}^2 +
8\lambda_2^2(\|B^{(t+1)}\|_{\oper}+\|B^{(t)}\|_{\oper})^2
\|\nabla_B^{(t)}\ell\|_{\fro}^2
\Bigg]^{1/2}.
\end{align*}

For the $B$-block, analoguously, we have
\begin{align*}
& 2\Big\langle (\Delta_{B}^{(t)}-\eta\nabla_{B}^{(t)}\ell(S_{B}^{(t)})^{-2} )S_{B}^{(t)},
(\Delta_{B}^{(t)}-\eta\nabla_{B}^{(t)}\ell(S_{B}^{(t)})^{-2} )(S_{B}^{(t+1)}-S_{B}^{(t)})\Big\rangle \\
&\le
\Bigg\{
2\|S_{B}^{(t)}\|_{\oper}\|\Delta_{B}^{(t)}\|_{\fro}^2
+2\eta^2\|\nabla_{B}^{(t)}\ell\|_{\fro}^2
\|(S_{B}^{(t)})^{-1}\|_{\oper}\|(S_{B}^{(t)})^{-2}\|_{\oper} \\
&\qquad
+2\eta\|\Delta_B^{(t)}\|_{\fro}\Big(
\|S_{B}^{(t)}\|_{\oper}\|(S_{B}^{(t)})^{-2}\|_{\oper}\|\nabla_{B}^{(t)}\ell\|_{\fro}
+\|\nabla_{B}^{(t)}\ell(S_{B}^{(t)})^{-1}\|_{\fro}
\Big)
\Bigg\} \\
&\qquad\times
\frac{2\lambda_2\eta}{\gamma_{B}^{(t)}}
(\|U^{(t+1)}\|_{\oper}+\|U^{(t)}\|_{\oper})
\|\nabla_U^{(t)}\ell\|_{\fro}.
\end{align*}

For the $U_{Q^c}$-block, analogously, we have
\begin{align*}
& 2\Big\langle (\Delta_{U_{Q^c}}^{(t)}-\eta\nabla_{U_{Q^c}}^{(t)}\ell(S_{U_{Q^c}}^{(t)})^{-2} )S_{U_{Q^c}}^{(t)},
(\Delta_{U_{Q^c}}^{(t)}-\eta\nabla_{U_{Q^c}}^{(t)}\ell(S_{U_{Q^c}}^{(t)})^{-2} )(S_{U_{Q^c}}^{(t+1)}-S_{U_{Q^c}}^{(t)})\Big\rangle \\
&\le
\Bigg\{
2\|S_{U_{Q^c}}^{(t)}\|_{\oper}\|\Delta_{U_{Q^c}}^{(t)}\|_{\fro}^2
+2\eta^2\|\nabla_{U_{Q^c}}^{(t)}\ell\|_{\fro}^2
\|(S_{U_{Q^c}}^{(t)})^{-1}\|_{\oper}\|(S_{U_{Q^c}}^{(t)})^{-2}\|_{\oper} \\
&\qquad
+2\eta \|\Delta_{U_{Q^c}}^{(t)}\|_{\fro}\Big(
\|S_{U_{Q^c}}^{(t)}\|_{\oper}\|(S_{U_{Q^c}}^{(t)})^{-2}\|_{\oper}\|\nabla_{U_{Q^c}}^{(t)}\ell\|_{\fro}
+\|\nabla_{U_{Q^c}}^{(t)}\ell(S_{U_{Q^c}}^{(t)})^{-1}\|_{\fro}
\Big)
\Bigg\} \\
&\qquad\times
\frac{2\lambda_2\eta}{\gamma_{U_{Q^c}}^{(t)}}
(\|B^{(t+1)}\|_{\oper}+\|B^{(t)}\|_{\oper})
\|\nabla_B^{(t)}\ell\|_{\fro}.
\end{align*}
This completes the proof.
\end{proof}

\begin{lemma}[Contraction bound for the interaction ratio]\label{lem:ratio-U}
Assume that $U_Q^{(t)}$ is full rank and that $\lambda_1>0$. 
 Define
    \[
    \gamma_1^{(t)}
    :=
   \frac{\sqrt{2\lambda_2\|B^{(t)}(\Delta_{U_Q}^{(t)})^\top\|^2_{\fro}}}{\sqrt{s_{U_Q}^{(t)}}} 
    \frac{\sqrt{2\lambda_2}\|\Delta_B^{(t)} {U_Q^{(t)}}^\top\|_{\fro}}{\sqrt{s_B^{(t)}}},\quad \gamma_2^{(t)}
    :=\frac{\sqrt{2\lambda_2\|\Delta_B^{(t)} {U_{Q^c}^{(t)}}^\top\|^2_{\fro}}}{\sqrt{s_{B}^{(t)}}}.
    \]
     On the event $\mathcal G_t$ in \eqref{event:norm-bound} and \Cref{asp:local-well-conditioned}, for $v\in\mathcal{V}$, define
\begin{align*}
\underline{\sigma}_{v}=\underline l_{v}\sigma_{\min}(v), \quad\overline{\sigma}_{v}=\overline l_{v}\|v\|_{\oper}.
\end{align*}
let
\[
\sigma_{\min}(U_Q^{(t)})\ge \underline{\sigma}_{U_Q}>0,
\qquad
\|B^{(t)}\|_{\oper}\le \overline{\sigma}_B<\infty.
\]
If, in addition, \(\lambda_1+\lambda_2=1\) and
\(\lambda_1\ge \underline\lambda>0\), then
\[
(\gamma_1^{(t)})^2+(\gamma_2^{(t)})^2\le \gamma<1,
\]
where
\[
\gamma
:=
1-(1-\bar\gamma_{\underline\lambda})\underline b,
\qquad
\bar\gamma_{\underline\lambda}
:=
\frac{(1-\underline\lambda)\overline{\sigma}_B^2}
{2\underline\lambda\,\underline{\sigma}_{U_Q}^2+
(1-\underline\lambda)\overline{\sigma}_B^2},\qquad \underline b :=  \frac{\underline{\sigma}_{U_Q}^2}{\underline{\sigma}_{U_Q}^2+\overline{\sigma}_{U_{Q^c}}^2} < 1.
\]
\end{lemma}
\begin{proof}[Proof of \Cref{lem:ratio-U}]

Define
\[
a_t
:=
\frac{2\lambda_2\|B^{(t)}(\Delta_{U_Q}^{(t)})^\top\|_{\fro}^2}{s_{U_Q}^{(t)}},
\qquad
b_t
:=
\frac{2\lambda_2\|\Delta_B^{(t)}U_Q^{(t)\top}\|_{\fro}^2}{s_B^{(t)}},
\qquad
c_t
:=
\frac{2\lambda_2\|\Delta_B^{(t)}U_{Q^c}^{(t)\top}\|_{\fro}^2}{s_B^{(t)}}.
\]
Then, by definition,
\[
(\gamma_1^{(t)})^2+(\gamma_2^{(t)})^2
=
a_t b_t + c_t.
\]
Moreover, by the definition of \(s_B^{(t)}\),
\[
s_B^{(t)}
=
2\lambda_2\|\Delta_B^{(t)}U^{(t)\top}\|_{\fro}^2
=
2\lambda_2\|\Delta_B^{(t)}U_Q^{(t)\top}\|_{\fro}^2
+
2\lambda_2\|\Delta_B^{(t)}U_{Q^c}^{(t)\top}\|_{\fro}^2,
\]
and when $s_B^{(t)}\neq 0$, hence
\[
b_t+c_t=1.
\]
Therefore,
\[
(\gamma_1^{(t)})^2+(\gamma_2^{(t)})^2
=
a_t b_t + (1-b_t)
=
1-(1-a_t)b_t.
\]
By \Cref{asp:local-well-conditioned}, we have
\[
\sigma_{\min}(U_Q^{(t)})
\ge
\underline l_{U_Q}\sigma_{\min}(U_Q),\qquad \sigma_{\min}(B^{(t)})\ge \underline l_B\sigma_{\min}(B),
\qquad
\sigma_{\min}(U_{Q^c}^{(t)})
\ge \underline l_{U_{Q^c}}\sigma_{\min}(U_{Q^c}).
\]
and
\[
\|U_Q^{(t)}\|_{\oper}
\le
\overline l_{U_Q}\|U_Q\|_{\oper},\qquad \|B^{(t)}\|_{\oper}\le \overline l_B\|B\|_{\oper},
\qquad
\|U_{Q^c}^{(t)}\|_{\oper}
\le \overline l_{U_{Q^c}}\|U_{Q^c}\|_{\oper}.
\]
Define
\begin{align*}
\underline{\sigma}_{v}=\underline l_{v}\sigma_{\min}(v), \quad\overline{\sigma}_{v}=\overline l_{v}\|v\|_{\oper}.
\end{align*}
Consequently, on the event \(\mathcal G_t\),
\[
a_t
=
\frac{2\lambda_2\|B^{(t)}(\Delta_{U_Q}^{(t)})^\top\|_{\fro}^2}{s_{U_Q}^{(t)}}
\le \bar\gamma_1,
\]
where
\[
\bar\gamma_1
:=
\frac{\lambda_2 \overline{\sigma}_B^2}
{2\lambda_1 \underline{\sigma}_{U_Q}^2+\lambda_2 \overline{\sigma}_B^2}
<1.
\]
Thus,
\[
(\gamma_1^{(t)})^2+(\gamma_2^{(t)})^2
=
1-(1-a_t)b_t
\le
1-(1-\bar\gamma_1)b_t.
\]

Similarly, on the event \(\mathcal G_t\), there exists a constant \( 0<\underline b :=  \frac{\underline{\sigma}_{U_Q}^2}{\underline{\sigma}_{U_Q}^2+\overline{\sigma}_{U_{Q^c}}^2} < 1\) such that
\[
b_t
=
\frac{\|\Delta_B^{(t)}U_Q^{(t)\top}\|_{\fro}^2}
{\|\Delta_B^{(t)}U_Q^{(t)\top}\|_{\fro}^2+\|\Delta_B^{(t)}U_{Q^c}^{(t)\top}\|_{\fro}^2}
\ge \underline b
\]
uniformly over \(t\). 
Then
\[
(\gamma_1^{(t)})^2+(\gamma_2^{(t)})^2
\le
1-(1-\bar\gamma_1)\underline b
.
\]
Since \(\lambda_1+\lambda_2=1\), we can write
\[
\bar\gamma_1
=
\frac{(1-\lambda_1)\overline{\sigma}_B^2}
{2\lambda_1\underline{\sigma}_{U_Q}^2+
(1-\lambda_1)\overline{\sigma}_B^2}.
\]
The right-hand side is decreasing in \(\lambda_1\). Hence, under
\(\lambda_1\ge \underline\lambda\),
\[
\bar\gamma_1
\le
\frac{(1-\underline\lambda)\overline{\sigma}_B^2}
{2\underline\lambda\,\underline{\sigma}_{U_Q}^2+
(1-\underline\lambda)\overline{\sigma}_B^2}
=:\bar\gamma_{\underline\lambda}.
\]
Therefore, since \(0<\bar\gamma_{\underline\lambda}<1\) and \(\underline b<1\), it follows that
\[
(\gamma_1^{(t)})^2+(\gamma_2^{(t)})^2
\le
1-(1-\bar\gamma_{\underline\lambda})\underline b
=:\gamma<1.
\]
This completes the proof.

\end{proof}

\clearpage
\subsection{Error recursion lemmas}\label{app:subsec:recusion}

\begin{lemma}[One-step update for weighted aggregate error]\label{lem:iter_all}
 For $v\in\mathcal{V}$ or $v=U$, define
\begin{align*}
\underline{\sigma}_{v}=\underline l_{v}\sigma_{\min}(v), \quad\overline{\sigma}_{v}=\overline l_{v}\|v\|_{\oper}.
\end{align*}
    Define step size scaling factor by
  \begin{align}\label{eq:c_eta_1}
           \eta&:= {\delta}\min\left\{  \frac{\lambda_1\underline{\sigma}_{U_Q}^2+\lambda_2\underline{\sigma}_{B}^2}{\overline \sigma^2_{U_Q}(\lambda_1\vee \lambda_2  )},\;\frac{\underline{\sigma}_{U}^2 }{\overline{\sigma}_U^2} ,\;,\frac{\underline{\sigma}_{B}^2 }{\overline{\sigma}_B^2}\right\},
    \end{align}
    for some constant $\delta>0$. Assume that \Cref{asp:row_rank}--\Cref{asp:local-well-conditioned} hold, 
    under events ${\cG}_t$ and $\cF_{t}$ defined in \eqref{event:norm-bound} and \eqref{event:iter}, if $\delta$ is sufficiently small, then event $\cF_{t+1}$ also holds.
\end{lemma}
\begin{proof}[Proof of \Cref{lem:iter_all}]
    Under event ${\mathcal{G}}_t$, for all $s\leq t$,
      \begin{align*}
s_{U_Q}^{(s)}+s_{B}^{(s)}+s_{U_{Q^c}}^{(s)} \leq \tau,
     \end{align*}
     where $\tau:=\frac{4\lambda_1c_1}{\kappa_{U_Q}^4} \|U_Q\|_{\oper}^4+\frac{2\lambda_2c_1}{\kappa_{U_Q}^4} \|U_Q\|_{\oper}^2\|B\|_{\oper}^2+\frac{2\lambda_2c_2}{\kappa_{U_{Q^c}}^4} \|U_{Q^c}\|_{\oper}^2\|B\|_{\oper}^2+\frac{2\lambda_2c_3}{\kappa_{B}^4} \|B\|_{\oper}^2\|U\|_{\oper}^2$.
The term can be decomposed as follows:
\begin{align}
   \sum_{v\in\mathcal{V}} s_v^{(t+1)} & = \sum_{v\in\mathcal{V}}\|\Delta_v^{(t+1)} S_v^{(t+1)} \|_{\fro}^2 \notag\\
 &\leq\sum_{v\in\mathcal V}
\bigl\|
\bigl({v}^{(t+1)}-v R^{(t)}\bigr) S_v^{(t+1)}
\bigr\|_{\fro}^2\notag \\ &=\sum_{v\in\mathcal V}
\bigl\|
\bigl({v}^{(t)}-\eta\nabla_v\ell(S_{v}^{(t)})^{-2}-v R^{(t)}\bigr) S_v^{(t+1)}
\bigr\|_{\fro}^2 \notag \\&=\sum_{v\in\mathcal V}
\bigl\|
\bigl({v}^{(t)}-\eta\nabla_v\ell(S_{v}^{(t)})^{-2}-v R^{(t)}\bigr) (S_v^{(t)}+S_v^{(t+1)}-S_v^{(t)})
\bigr\|_{\fro}^2 \notag \\
    &= \sum_{v\in\mathcal V}\big(s_v^{(t)} - 2 \eta \langle \nabla_v^{(t)} \ell, \Delta_v^{(t)} \rangle + \eta^2\|\nabla_v^{(t)} \ell\; (S_v^{(t)})^{-1}\|_{\fro}^2+\|\Delta_v^{(t)}(S_v^{(t+1)}-S_v^{(t)})\|_{\fro}^2\notag \\&\quad -2\eta\langle \Delta_v^{(t)}(S_v^{(t+1)}-S_v^{(t)}),\nabla_v\ell(S_{v}^{(t)})^{-2}(S_v^{(t+1)}-S_v^{(t)})\rangle +\eta^2\|\nabla_v\ell(S_{v}^{(t)})^{-2}(S_v^{(t+1)}-S_v^{(t)})\|_{\fro}^2\notag\\ & \quad +2\langle (\Delta_v^{(t)}-\eta\nabla_v\ell(S_{v}^{(t)})^{-2} )S_v^{(t)}, (\Delta_v^{(t)}-\eta\nabla_v\ell(S_{v}^{(t)})^{-2} )(S_v^{(t+1)}-S_v^{(t)})\rangle \big).
    \label{eq:iter-s}
\end{align}
We next bound each term separately. By \Cref{lem:inner-product-eq}~(\ref{overcross}),
\begin{align*}
   - 2 \sum_{v\in\mathcal V}  
     \langle \nabla_v^{(t)} \ell, \Delta_v^{(t)} \rangle
    \le &
    -2(1-\sqrt\gamma)\sum_{v\in\mathcal V}s_v^{(t)} 
    +\frac{k}{\BC_1}\|  ( 4\lambda_1D U_Q^{(t)}+2\lambda_2 E_Q B^{(t)})(S_{U_Q}^{(t)})^{-1}\|_{\oper}^2+\frac{4k}{\BC_2}\|\lambda_2  E_{Q^c} B^{(t)} S_{U_{Q^c}}^{-1}\|_{\oper}^2\\ & +\frac{4k}{\BC_3}\|\lambda_2 E^\top U^{(t)}(S_{B}^{(t)})^{-1}\|_{\oper}^2 +(4\BC_1+\BC_4)\lambda_1
    \|\Delta_{U_Q}^{(t)}U_Q^{(t)\top}\|_{\fro}^2 +\BC_6\lambda_2
    \|\Delta_U^{(t)}B^{(t)\top}\|_{\fro}^2 \\
    &+2\BC_1\lambda_2
    \|\Delta_{U_Q}^{(t)}B^{(t)\top}\|_{\fro}^2 +2\BC_2\lambda_2
    \|\Delta_{U_{Q^c}}^{(t)}B^{(t)\top}\|_{\fro}^2 
    +(2\BC_3+\BC_5)\lambda_2
    \|\Delta_B^{(t)}U^{(t)\top}\|_{\fro}^2 \\
    &+\frac{16\lambda_1}{\BC_4}\|\Delta_{U_Q}^{(t)}\|_{\fro}^4 +4\lambda_2\!\left(\frac{1}{\BC_5}+\frac{1}{\BC_6}\right)
    \|\Delta_U^{(t)}\|_{\fro}^2
    \|\Delta_B^{(t)}\|_{\fro}^2 ,
\end{align*}
for some arbitrary constants $\BC_1,\BC_2,\ldots,\BC_6$ and  the positive constant $\gamma<1$.
 Using \Cref{lem:gradient}~(\ref{nabla:UQ})-(\ref{nabla:UQc}), we obtain
\begin{align*}
&\eta^2\sum_{v\in\mathcal V} \big\|\nabla_v^{(t)} \ell\,(S_v^{(t)})^{-1}\big\|_{\fro}^2\\
\;\le\;&
\eta^2\Bigg[
{
\frac{1}{
2\lambda_1\sigma_{\min}^2\!\big(U_Q^{(t)}\big)+\lambda_2\sigma_{\min}^2\!\big(B^{(t)}\big)
}
}
\Bigg\{
24\lambda_1^2\|U_Q^{(t)}\|_{\oper}^2\Big(
\big\|U_Q^{(t)}{\Delta_{U_Q}^{(t)}}^\top\big\|_{\fro}^2
+\big\|\Delta_{U_Q}^{(t)}\big(U_Q^{(t)}\big)^\top\big\|_{\fro}^2
+\big\|\Delta_{U_Q}^{(t)}{\Delta_{U_Q}^{(t)}}^\top\big\|_{\fro}^2
\Big)
\nonumber\\[-2pt]
&\quad
+6\lambda_2^2\big\|B^{(t)}\big\|_{\oper}^2\Big(
\big\|\Delta_{U_Q}^{(t)}{B^{(t)}}^\top\big\|_{\fro}^2
+\big\|U_Q^{(t)}(\Delta_B^{(t)})^\top\big\|_{\fro}^2
+\big\|\Delta_{U_Q}^{(t)}(\Delta_B^{(t)})^\top\big\|_{\fro}^2
\Big)
\nonumber\\[-2pt]
&\quad
+16\lambda_1^2k\big\|U_Q^{(t)}\big\|_{\oper}^2\|D\|_{\oper}^2
+4\lambda_2^2k\big\|B^{(t)}\big\|_{\oper}^2\|E_Q\|_{\oper}^2
\Bigg\}
\nonumber\\[4pt]
&\;+\;
{\frac{6\lambda_2\|U^{(t)}\|_{\oper}^2}{\sigma_{\min}^2\!\big(U^{(t)}\big)}}
\Bigg\{
2\big\|\Delta_U^{(t)}{B^{(t)}}^\top\big\|_{\fro}^2
+2\big\|\Delta_U^{(t)}(\Delta_B^{(t)})^\top\big\|_{\fro}^2
+\big\|U^{(t)}(\Delta_B^{(t)})^\top\big\|_{\fro}^2
+k\|E\|_{\oper}^2
\Bigg\}
\nonumber\\[4pt]
&\;+\;
{\frac{6\lambda_2\|B^{(t)}\|_{\oper}^2}{\sigma_{\min}^2\!\big(B^{(t)}\big)}}
\Bigg\{
2\big\|\Delta_{U_{Q^c}}^{(t)}{B^{(t)}}^\top\big\|_{\fro}^2
+2\big\|\Delta_{U_{Q^c}}^{(t)}(\Delta_B^{(t)})^\top\big\|_{\fro}^2
+\big\|U_{Q^c}^{(t)}(\Delta_B^{(t)})^\top\big\|_{\fro}^2
+k\|E_{Q^c}\|_{\oper}^2
\Bigg\}
\Bigg].
\end{align*}
Applying the results of \Cref{lem:DeltaS},
\begin{align*}
    \|\Delta_v^{(t)}(S_v^{(t+1)}-S_v^{(t)})\|_{\fro}^2 \le & \frac{\|\Delta_{U_Q}^{(t)}\|_{\fro}^2}{\gamma_{U_Q}^{(t)2}}
\Big[
32\lambda_1^2\eta^2
(\|U_Q^{(t+1)}\|_{\oper}+\|U_Q^{(t)}\|_{\oper})^2
\|\nabla_{U_Q}^{(t)}\ell\|_{\fro}^2 \\
&\qquad
+
8\lambda_2^2\eta^2
(\|B^{(t+1)}\|_{\oper}+\|B^{(t)}\|_{\oper})^2
\|\nabla_{B}^{(t)}\ell\|_{\fro}^2
\Big]\\
&+
\frac{\|\Delta_{B}^{(t)}\|_{\fro}^2}{\gamma_{B}^{(t)2}}
4\lambda_2^2\eta^2
(\|U^{(t+1)}\|_{\oper}+\|U^{(t)}\|_{\oper})^2
\|\nabla_{U}^{(t)}\ell\|_{\fro}^2\\
&+
\frac{\|\Delta_{U_{Q^c}}^{(t)}\|_{\fro}^2}{\gamma_{U_{Q^c}}^{(t)2}}
4\lambda_2^2\eta^2
(\|B^{(t+1)}\|_{\oper}+\|B^{(t)}\|_{\oper})^2
\|\nabla_{B}^{(t)}\ell\|_{\fro}^2
\end{align*}
with $\gamma_v^{(t)} := \sigma_{\min}(S_v^{(t+1)})+\sigma_{\min}(S_v^{(t)}),\;v\in\mathcal V$.
From the proof of \Cref{lem:gradient}  together with the \Cref{lem:norm}, we can deduce that 
$ \|\Delta_v^{(t)}(S_v^{(t+1)}-S_v^{(t)})\|_{\fro}^2$ can be bounded by $\eta^2\sum_{v\in\mathcal V} \big\|\nabla_v^{(t)} \ell\,(S_v^{(t)})^{-1}\big\|_{\fro}^2$.
By \Cref{lem:grad-DeltaS}, we obtain
\begin{align*}
&\eta^2\sum_{v\in\mathcal V}
\|\nabla_v^{(t)}\ell\,(S_v^{(t)})^{-2}(S_v^{(t+1)}-S_v^{(t)})\|_{\fro}^2 \\
&\le
\frac{\eta^4\|\nabla_{U_Q}^{(t)}\ell\|_{\fro}^2}
{\sigma_{\min}(S_{U_Q}^{(t)})^4\,\gamma_{U_Q}^{(t)2}}
\Big[
32\lambda_1^2
(\|U_Q^{(t+1)}\|_{\oper}+\|U_Q^{(t)}\|_{\oper})^2
\|\nabla_{U_Q}^{(t)}\ell\|_{\fro}^2 \\
&\hspace{5em}
+
8\lambda_2^2
(\|B^{(t+1)}\|_{\oper}+\|B^{(t)}\|_{\oper})^2
\|\nabla_{B}^{(t)}\ell\|_{\fro}^2
\Big]\\
&\quad+
\frac{\eta^4\|\nabla_{B}^{(t)}\ell\|_{\fro}^2}
{\sigma_{\min}(S_{B}^{(t)})^4\,\gamma_{B}^{(t)2}}
\,
4\lambda_2^2
(\|U^{(t+1)}\|_{\oper}+\|U^{(t)}\|_{\oper})^2
\|\nabla_{U}^{(t)}\ell\|_{\fro}^2\\
&\quad+
\frac{\eta^4\|\nabla_{U_{Q^c}}^{(t)}\ell\|_{\fro}^2}
{\sigma_{\min}(S_{U_{Q^c}}^{(t)})^4\,\gamma_{U_{Q^c}}^{(t)2}}
\,
4\lambda_2^2
(\|B^{(t+1)}\|_{\oper}+\|B^{(t)}\|_{\oper})^2
\|\nabla_{B}^{(t)}\ell\|_{\fro}^2 .
\end{align*}
Hence the term $\eta^2\sum_{v\in\mathcal V}\|\nabla_v\ell(S_{v}^{(t)})^{-2}(S_v^{(t+1)}-S_v^{(t)})\|_{\fro}^2$ is also dominated by term $\eta^2\sum_{v\in\mathcal V} \big\|\nabla_v^{(t)} \ell\,(S_v^{(t)})^{-1}\big\|_{\fro}^2$. Moreover, using the inequality
\begin{align*}
   & -2\eta\sum_{v\in\mathcal V}\langle \Delta_v^{(t)}(S_v^{(t+1)}-S_j^{(t)}),\nabla_v\ell(S_{v}^{(t)})^{-2}(S_v^{(t+1)}-S_v^{(t)})\rangle\\ \leq &\eta^2\sum_{v\in\mathcal V}\|\nabla_v\ell(S_{v}^{(t)})^{-2}(S_v^{(t+1)}-S_v^{(t)})\|_{\fro}^2+ \|\Delta_v^{(t)}(S_v^{(t+1)}-S_v^{(t)})\|_{\fro}^2,
\end{align*}
we conclude that the term $-2\eta\sum_{v\in\mathcal V}\langle \Delta_v^{(t)}(S_v^{(t+1)}-S_v^{(t)}),\nabla_v\ell(S_{v}^{(t)})^{-2}(S_v^{(t+1)}-S_v^{(t)})\rangle$ is also dominated by $\eta^2\sum_{v\in\mathcal V} \big\|\nabla_v^{(t)} \ell\,(S_v^{(t)})^{-1}\big\|_{\fro}^2$. Using \Cref{lem:cross-term-blockwise},
$ 2\sum_{v\in\mathcal V}\langle (\Delta_v^{(t)}-\eta\nabla_v\ell(S_{v}^{(t)})^{-2} )S_v^{(t)}, (\Delta_v^{(t)}-\eta\nabla_v\ell(S_{v}^{(t)})^{-2} )(S_v^{(t+1)}-S_v^{(t)})\rangle \big)$ is also dominated by $\eta^2\sum_{v\in\mathcal V} \big\|\nabla_v^{(t)} \ell\,(S_v^{(t)})^{-1}\big\|_{\fro}^2$.
Hence, \eqref{eq:iter-s} can be simplified as
\begin{align*}
      &\sum_{v\in\mathcal{V}} s_v^{(t+1)} \\ \le & 
     \sum_{v\in\mathcal V}\big(s_v^{(t)} - 2 \eta \langle \nabla_v^{(t)} \ell, \Delta_v^{(t)} \rangle + C\eta^2\|\nabla_v^{(t)} \ell\; (S_v^{(t)})^{-1}\|_{\fro}^2\big)\\
     \le &\bigr(1-2\eta(1-\sqrt\gamma)\bigr)\sum_{v\in\mathcal V}s_v^{(t)} +
\frac{\eta k}{\BC_1}\|  ( 4\lambda_1D U_Q^{(t)}+2\lambda_2 E_Q B^{(t)})(S_{U_Q}^{(t)})^{-1}\|_{\oper}^2+\frac{4\eta k}{\BC_2}\|\lambda_2  E_{Q^c} B^{(t)} (S_{U_{Q^c}}^{(t)})^{-1}\|_{\oper}^2\\
&+\frac{4\eta k}{\BC_3}\|\lambda_2 E^\top U^{(t)}(S_{B}^{(t)})^{-1}\|_{\oper}^2 +(4\BC_1+\BC_4)\lambda_1\eta
    \|\Delta_{U_Q}^{(t)}U_Q^{(t)\top}\|_{\fro}^2 +\BC_6\lambda_2\eta
    \|\Delta_U^{(t)}B^{(t)\top}\|_{\fro}^2 \\
    & +2\BC_1\lambda_2\eta
    \|\Delta_{U_Q}^{(t)}B^{(t)\top}\|_{\fro}^2 +2\BC_2\lambda_2\eta
    \|\Delta_{U_{Q^c}}^{(t)}B^{(t)\top}\|_{\fro}^2 
    +(2\BC_3+\BC_5)\lambda_2\eta
    \|\Delta_B^{(t)}U^{(t)\top}\|_{\fro}^2 \\
    &+\frac{16\lambda_1\eta}{\BC_4}\|\Delta_{U_Q}^{(t)}\|_{\fro}^4 +4\lambda_2\eta\!\left(\frac{1}{\BC_5}+\frac{1}{\BC_6}\right)
    \|\Delta_U^{(t)}\|_{\fro}^2
    \|\Delta_B^{(t)}\|_{\fro}^2\\
&
+C\eta^2\Bigg[
\frac{1}{
2\lambda_1\sigma_{\min}^2\!\big(U_Q^{(t)}\big)+\lambda_2\sigma_{\min}^2\!\big(B^{(t)}\big)
}
\Bigg\{
24\lambda_1^2\|U_Q^{(t)}\|_{\oper}^2\Big(
\|U_Q^{(t)}{\Delta_{U_Q}^{(t)}}^\top\|_{\fro}^2
+\|\Delta_{U_Q}^{(t)}(U_Q^{(t)})^\top\|_{\fro}^2 \\
&\hspace{5em}
+\|\Delta_{U_Q}^{(t)}{\Delta_{U_Q}^{(t)}}^\top\|_{\fro}^2
\Big)
+6\lambda_2^2\|B^{(t)}\|_{\oper}^2\Big(
\|\Delta_{U_Q}^{(t)}{B^{(t)}}^\top\|_{\fro}^2
+\|U_Q^{(t)}(\Delta_B^{(t)})^\top\|_{\fro}^2 \\
&\hspace{5em}
+\|\Delta_{U_Q}^{(t)}(\Delta_B^{(t)})^\top\|_{\fro}^2
\Big)
+16\lambda_1^2k\|U_Q^{(t)}\|_{\oper}^2\|D\|_{\oper}^2
+4\lambda_2^2k\|B^{(t)}\|_{\oper}^2\|E_Q\|_{\oper}^2
\Bigg\}\\
&
+\frac{6\lambda_2\|U^{(t)}\|_{\oper}^2}{\sigma_{\min}^2\!\big(U^{(t)}\big)}
\Bigg\{
2\|\Delta_U^{(t)}{B^{(t)}}^\top\|_{\fro}^2
+2\|\Delta_U^{(t)}(\Delta_B^{(t)})^\top\|_{\fro}^2
+\|U^{(t)}(\Delta_B^{(t)})^\top\|_{\fro}^2
+k\|E\|_{\oper}^2
\Bigg\}\\
&
+\frac{6\lambda_2\|B^{(t)}\|_{\oper}^2}{\sigma_{\min}^2\!\big(B^{(t)}\big)}
\Bigg\{
2\|\Delta_{U_{Q^c}}^{(t)}{B^{(t)}}^\top\|_{\fro}^2
+2\|\Delta_{U_{Q^c}}^{(t)}(\Delta_B^{(t)})^\top\|_{\fro}^2
+\|U_{Q^c}^{(t)}(\Delta_B^{(t)})^\top\|_{\fro}^2
+k\|E_{Q^c}\|_{\oper}^2
\Bigg\}
\Bigg]\\
&\le
\left[1-\eta \left\{2(1-\sqrt\gamma)-\BC_1-\BC_4/4-\frac{4\|\Delta_{U_Q}^{(t)}\|_{\fro}^4}{\BC_4\|\Delta_{U_Q}^{(t)}U_Q^{(t)\top}\|_{\fro}^2}\right\}\right]4\lambda_1\|\Delta_{U_Q}^{(t)}U_Q^{(t)\top}\|_{\fro}^2\\  
&\quad
+\left[1-\eta \left\{2(1-\sqrt\gamma)-\BC_1-\BC_6/2-\frac{2\|\Delta_{U_Q}^{(t)}\|_{\fro}^2\|\Delta_B^{(t)}\|_{\fro}^2}{\BC_6\|\Delta_{U_Q}^{(t)}B^{(t)\top}\|_{\fro}^2}
\right\}\right]2\lambda_2\|\Delta_{U_Q}^{(t)}B^{(t)\top}\|_{\fro}^2 \\
&\quad
+\left[1-\eta \left\{2(1-\sqrt\gamma)-\BC_2-\BC_6/2-\frac{2\|\Delta_{U_{Q^c}}^{(t)}\|_{\fro}^2\|\Delta_B^{(t)}\|_{\fro}^2}{\BC_6\|\Delta_{U_{Q^c}}^{(t)}B^{(t)\top}\|_{\fro}^2}
\right\}\right]2\lambda_2\|\Delta_{U_{Q^c}}^{(t)}B^{(t)\top}\|_{\fro}^2 \\
&\quad
+\left[1-\eta \left\{2(1-\sqrt\gamma)-\BC_3-\BC_5/2-\frac{2\|\Delta_U^{(t)}\|_{\fro}^2\|\Delta_B^{(t)}\|_{\fro}^2}{\BC_5\|\Delta_B^{(t)}U^{(t)\top}\|_{\fro}^2}\right\}\right]2\lambda_2\|\Delta_B^{(t)}U^{(t)\top}\|_{\fro}^2
\\& \quad+\frac{\eta k}{\BC_1}\|  ( 4\lambda_1D U_Q^{(t)}+2\lambda_2 E_Q B^{(t)})(S_{U_Q}^{(t)})^{-1}\|_{\oper}^2+\frac{4\eta k}{\BC_2}\|\lambda_2  E_{Q^c} B^{(t)} (S_{U_{Q^c}}^{(t)})^{-1}\|_{\oper}^2\\
&\quad+\frac{4\eta k}{\BC_3}\|\lambda_2 E^\top U^{(t)}(S_{B}^{(t)})^{-1}\|_{\oper}^2
+C\eta^2 \cdot\text{high order coefficient} ,
\end{align*}
where
\[
\gamma
:=
1-(1-\bar\gamma_{\underline\lambda})\underline b,
\qquad
\bar\gamma_{\underline\lambda}
:=
\frac{(1-\underline\lambda)\overline{\sigma}_B^2}
{2\underline\lambda\,\underline{\sigma}_{U_Q}^2+
(1-\underline\lambda)\overline{\sigma}_B^2},\qquad \underline b :=  \frac{\underline{\sigma}_{U_Q}^2}{\underline{\sigma}_{U_Q}^2+\overline{\sigma}_{U_{Q^c}}^2} < 1.
\]
Thus,
\[
0<1-\sqrt{\gamma}
=
1-
\sqrt{
\frac{
2\underline\lambda\,\underline{\sigma}_{U_Q}^2
\overline{\sigma}_{U_{Q^c}}^2
+
(1-\underline\lambda)\overline{\sigma}_B^2
\underline{\sigma}_{U_Q}^2
+
(1-\underline\lambda)\overline{\sigma}_B^2
\overline{\sigma}_{U_{Q^c}}^2
}{
\bigl(
2\underline\lambda\,\underline{\sigma}_{U_Q}^2
+(1-\underline\lambda)\overline{\sigma}_B^2
\bigr)
\bigl(
\underline{\sigma}_{U_Q}^2+\overline{\sigma}_{U_{Q^c}}^2
\bigr)
}
}<1.
\]
For a fixed, sufficiently small constant $ c_p$, set 
\begin{align*}
    \BC_1=\BC_2=\BC_3=c_p(1-\sqrt\gamma),
\end{align*}
and 
\begin{align*}
   \BC_4=c_p\sqrt \tau/\|U_Q\|_{\oper}^2, \quad\BC_5=c_p\sqrt \tau/\|B\|_{\oper}^2,\quad\BC_6=c_p\sqrt \tau/\|U\|_{\oper}^2.
\end{align*}
And we also require that $\eta$ is a sufficiently small constant such that the second-order term can be dominated by the first-order term. Therefore, there exists a constant $C$, we have
\begin{align*}
     & \sum_{v\in\mathcal{V}} s_v^{(t+1)} \\ \le&
\bigr(1-\eta(1-\sqrt\gamma)\bigr)\sum_{v\in\mathcal V}s_v^{(t)} 
\notag\\ &+C\eta k \left\{\|  ( 4\lambda_1D U_Q^{(t)}+2\lambda_2 E_Q B^{(t)})(S_{U_Q}^{(t)})^{-1}\|_{\oper}^2+\|\lambda_2  E_{Q^c} B^{(t)} (S_{U_{Q^c}}^{(t)})^{-1}\|_{\oper}^2+\|\lambda_2 E^\top U^{(t)}(S_{B}^{(t)})^{-1}\|_{\oper}^2\right\}.
\end{align*}
By \Cref{asp:local-well-conditioned}, we have
\[
\sigma_{\min}(U_Q^{(t)})
\ge
\underline l_{U_Q}\sigma_{\min}(U_Q),\qquad\sigma_{\min}(B^{(t)})\ge \underline l_B\sigma_{\min}(B),
\qquad
\sigma_{\min}(U_{Q^c}^{(t)})
\ge \underline l_{U_{Q^c}}\sigma_{\min}(U_{Q^c}).
\]
and
\[
\|U_Q^{(t)}\|_{\oper}
\le
\overline l_{U_Q}\|U_Q\|_{\oper},\qquad\|B^{(t)}\|_{\oper}\le \overline l_B\|B\|_{\oper},
\qquad
\|U_{Q^c}^{(t)}\|_{\oper}
\le \overline l_{U_{Q^c}}\|U_{Q^c}\|_{\oper}.
\]
We can simplify as
\begin{align}
      \sum_{v\in\mathcal{V}} s_v^{(t+1)} &\le
\bigr(1-\eta(1-\sqrt\gamma)\bigr)\sum_{v\in\mathcal V}s_v^{(t)} 
\notag\\ &\quad +C\eta k \left\{\frac{\lambda_1^2 \|U_Q\|_{\oper}^2\|D\|_{\oper}^2+\lambda_2^2\|
E_Q\|_{\oper}^2\|B\|_{\oper}^2 }{2\lambda_1\sigma_{\min}^2(U_Q)+\lambda_2\sigma_{\min}^2(B)}  +\frac{\lambda_2\|E_{Q^c}\|_{\oper}^2\| B\|_{\oper}^2 }{\sigma_{\min}^2(B)}+ \frac{\lambda_2\|E\|_{\oper}^2\| U\|_{\oper}^2 }{\sigma_{\min}^2(U)} \right\}\label{eq:contraction}.
\end{align}
Thus, the above implies event $\cF_{t+1}\setminus \cF_t$ holds.
        Consequently, event $\cF_{t+1}$ holds, which completes the proof.

\end{proof}

\begin{lemma}[Norm bounds for componentwise errors] \label{lem:norm}  
  Under \Cref{asp:row_rank}, \Cref{asp:init}, \Cref{asp:operator} and events $\cG_t$, $\cF_{t+1}$, defined in \eqref{event:norm-bound}, \eqref{event:iter}, respectively, if $\eta$ in \eqref{eq:c_eta_1} is sufficiently small, then event $\cG_{t+1}$ also holds. 
\end{lemma}
\begin{proof}[Proof of \Cref{lem:norm}]
     On the event ${\mathcal{G}}_t$, the following error bounds hold:
    \begin{align*}
    &s_{U_Q}^{(t)}+s_{B}^{(t)}+s_{U_{Q^c}}^{(t)}\\ =&4\lambda_1\|\Delta_{U_Q}^{(t)}(U_Q^{(t)})^\top \|_{\fro}^2+2\lambda_2\|\Delta_{U_Q}^{(t)}(B^{(t)})^\top \|_{\fro}^2+2\lambda_2\|\Delta_{U_{Q^c}}^{(t)} (B^{(t)})^\top \|_{\fro}^2+ 2\lambda_2\|\Delta_{B}^{(t)}(U^{(t)})^\top \|_{\fro}^2\\ \leq & \frac{4\lambda_1c_1}{\kappa_{U_Q}^4} \|U_Q\|_{\oper}^4+\frac{2\lambda_2c_1}{\kappa_{U_Q}^4} \|U_Q\|_{\oper}^2\|B\|_{\oper}^2+\frac{2\lambda_2c_2}{\kappa_{U_{Q^c}}^4} \|U_{Q^c}\|_{\oper}^2\|B\|_{\oper}^2+\frac{2\lambda_2c_3}{\kappa_{B}^4} \|B\|_{\oper}^2\|U\|_{\oper}^2\\ :=&\tau.  
    \end{align*}
By \eqref{eq:contraction}, there exists a constant $C$ such that 
     \begin{align*}
      \sum_{v\in\mathcal{V}} s_v^{(t+1)}  \le &
\bigr(1-\eta(1-\sqrt\gamma)\bigr)\sum_{v\in\mathcal V}s_v^{(t)} 
\notag\\ &\quad   +C\eta k \left\{\frac{\lambda_1^2 \|U_Q\|_{\oper}^2\|D\|_{\oper}^2+\lambda_2^2\|
E_Q\|_{\oper}^2\|B\|_{\oper}^2 }{2\lambda_1\sigma_{\min}^2(U_Q)+\lambda_2\sigma_{\min}^2(B)}  +\frac{\lambda_2\|E_{Q^c}\|_{\oper}^2\| B\|_{\oper}^2 }{\sigma_{\min}^2(B)}+ \frac{\lambda_2\|E\|_{\oper}^2\| U\|_{\oper}^2 }{\sigma_{\min}^2(U)} \right\}\\ \le & \tau \left[\bigr(1-\eta(1-\sqrt\gamma)\bigr)\right. \\ &\quad  \left.+\frac{C\eta k}{\tau} \left\{\frac{\lambda_1^2 \|U_Q\|_{\oper}^2\|D\|_{\oper}^2+\lambda_2^2\|
E_Q\|_{\oper}^2\|B\|_{\oper}^2 }{2\lambda_1\sigma_{\min}^2(U_Q)+\lambda_2\sigma_{\min}^2(B)}  +\frac{\lambda_2\|E_{Q^c}\|_{\oper}^2\| B\|_{\oper}^2 }{\sigma_{\min}^2(B)}+ \frac{\lambda_2\|E\|_{\oper}^2\| U\|_{\oper}^2 }{\sigma_{\min}^2(U)} \right\}\right].
\end{align*}
 For fixed parameters $\lambda_1,\lambda_2$, and sufficiently small constant $\eta$ given in \eqref{eq:c_eta_1}, by \Cref{asp:operator}, the following inequality holds
 \begin{align*}
     \tau \geq  C_1\left\{\frac{\lambda_1^2 \|U_Q\|_{\oper}^2\|D\|_{\oper}^2+\lambda_2^2\|
E_Q\|_{\oper}^2\|B\|_{\oper}^2 }{2\lambda_1\sigma_{\min}^2(U_Q)+\lambda_2\sigma_{\min}^2(B)}  +\frac{\lambda_2\|E_{Q^c}\|_{\oper}^2\| B\|_{\oper}^2 }{\sigma_{\min}^2(B)}+ \frac{\lambda_2\|E\|_{\oper}^2\| U\|_{\oper}^2 }{\sigma_{\min}^2(U)}\right\},
 \end{align*}
 for a sufficiently large constant $C_1$, we deduce that
 \begin{align*}
      \sum_{v\in\mathcal{V}} s_v^{(t+1)}  \le \tau.
 \end{align*}
This completes the proof.
\end{proof}

\clearpage

\subsection{Auxiliary lemmas}\label{app:subsec:ineq}
\begin{lemma}[Block matrix inequality]\label{lm:rank} 
    Consider a $2$-by-$2$ block matrix $A$ of arbitrary sizes:
    \begin{align*}
        A= \begin{pmatrix}
          A_{11} &  A_{12} \\
             A_{21} &  A_{22}
        \end{pmatrix}.
    \end{align*}
    The following statements hold:
     \begin{enumerate}[(i)]
        \item The operator norm of \(A\) satisfies $ \| A\|_{\oper} \leq  \|A_{11}\|_{\oper}+\|A_{12}\|_{\oper}+\|A_{21}\|_{\oper}+\|A_{22}\|_{\oper}$. Moreover, if $A_{12}=A_{21}^\top$, then $ \| A\|_{\oper} \leq  \|A_{11}\|_{\oper}+\|A_{12}\|_{\oper}+\|A_{22}\|_{\oper}$. 
        \item The rank of \(A\) is bounded by $\rank(A) \leq \rank(A_{11})+\rank(A_{12})+\rank(A_{21})+\rank(A_{22})$. 
    \end{enumerate}
\end{lemma}
\begin{proof}[Proof of Lemma~\ref{lm:rank}]
    For the item (i), we decompose \(A\) into four matrices, each containing only one nonzero block:
    \[
    A = A^{(11)} + A^{(12)} + A^{(21)} + A^{(22)},
    \]
    where
    \[
    A^{(11)} = 
    \begin{pmatrix}
    A_{11} & 0\\
    0 & 0
    \end{pmatrix}, \quad
    A^{(12)} =
    \begin{pmatrix}
    0 & A_{12}\\
    0 & 0
    \end{pmatrix}, \quad
    A^{(21)} =
    \begin{pmatrix}
    0 & 0\\
    A_{21} & 0
    \end{pmatrix}, \quad
    A^{(22)} =
    \begin{pmatrix}
    0 & 0\\
    0 & A_{22}
    \end{pmatrix}.
    \]
    By definition, the operator norm is
    \(
    \|A\|_{\oper} = \sup_{\|x\|_2 = 1} \|A x\|_2.
    \)
    Using the decomposition of \(A\),
    \[
    A x = A^{(11)}x + A^{(12)}x + A^{(21)}x + A^{(22)}x.
    \]
    Applying the triangle inequality for the Euclidean norm, we obtain
    \[
    \|A x\|_2 \leq \|A^{(11)}x\|_2 + \|A^{(12)}x\|_2 + \|A^{(21)}x\|_2 + \|A^{(22)}x\|_2.
    \]
    Taking the supremum over all unit vectors \(x\),
    \[
    \|A\|_{\oper} \leq \|A^{(11)}\|_{\oper} + \|A^{(12)}\|_{\oper} + \|A^{(21)}\|_{\oper} + \|A^{(22)}\|_{\oper}.
    \]
    Finally, since each \(A^{(ij)}\) acts nontrivially only on a subset of coordinates, its operator norm equals that of its nonzero block:
    \(
    \|A^{(ij)}\|_{\oper} = \|A_{ij}\|_{\oper}.
    \)
    Combining the above yields
    \[
    \|A\|_{\oper} \leq \|A_{11}\|_{\oper} + \|A_{12}\|_{\oper} + \|A_{21}\|_{\oper} + \|A_{22}\|_{\oper}.
    \]
    If $A_{12}=A_{21}^\top$, then \(\|A^{(12)}+A^{(21)}\|_{\oper}\leq  \|A_{12}\|_{\oper}\) and
    \begin{align*}
    \|A\|_{\oper} \leq \|A^{(11)}\|_{\oper} + \|A^{(12)}+A^{(21)}\|_{\oper} + \|A^{(22)}\|_{\oper} \leq  \|A_{11}\|_{\oper}+\|A_{12}\|_{\oper}+\|A_{22}\|_{\oper}.
    \end{align*}
    
    To prove the term (ii), by the subadditivity of the matrix rank, for any matrices \(X\) and \(Y\) of the same size,
    \(
    \mathrm{rank}(X + Y) \leq \mathrm{rank}(X) + \mathrm{rank}(Y).
    \)
    Applying this property repeatedly to \(A = A^{(11)} + A^{(12)} + A^{(21)} + A^{(22)}\) gives
    \[
    \mathrm{rank}(A) \leq \mathrm{rank}(A^{(11)}) + \mathrm{rank}(A^{(12)}) + \mathrm{rank}(A^{(21)}) + \mathrm{rank}(A^{(22)}).
    \]
    Since each \(A^{(ij)}\) has only one nonzero block, its rank equals that of the block itself, i.e.,
    \[
    \mathrm{rank}(A^{(ij)}) = \mathrm{rank}(A_{ij}).
    \]
    Thus,
    \[
    \mathrm{rank}(A) \leq \mathrm{rank}(A_{11}) + \mathrm{rank}(A_{12}) + \mathrm{rank}(A_{21}) + \mathrm{rank}(A_{22}),
    \]
    which completes the proof.
\end{proof}

\begin{lemma}[Procrustes alignment: basic identities and norm relations]\label{lem:procrustes_basic}
    Let $U,\hat U\in\R^{n\times k}$ and let
    \[
    R\in\argmin_{Q\in\mathcal O(k)}\|\hat U-UQ\|_{\fro}
    \qquad\text{and}\qquad
    \Delta:=\hat U-UR .
    \]
    Assume $G:=U^\top U$ and define the (orthogonal) projector onto $\Col(U)$ by \(P:=U G^\dagger U^\top\), where $G^\dagger$ denotes the Moore--Penrose pseudoinverse (if $U$ has full column rank then $G^\dagger=G^{-1}$).
    Set the orthogonal decomposition \(E:=(I-P)\Delta\), \(A:=G^\dagger U^\top\Delta\) so that $\Delta=UA+E$ and $U^\top E=0$.
    Then the following hold:    
    \begin{enumerate}[(i)]
    \item (First-order optimality)
    The matrix $U^\top \hat U R^\top$ is symmetric positive semidefinite and hence, $U^\top\Delta$ is symmetric.
    
    \item (Identities)
    Let $M:=U^\top\Delta$. Then
    \begin{align*}
        \tr(M^2) &=\|M\|_{\fro}^2\\
        \langle U\Delta^\top U,\Delta\rangle
        &=\langle \Delta U^\top,\,U\Delta^\top\rangle
        =\|M\|_{\fro}^2\\
        M + M^{\top}
        &=
        -\,\Delta^\top \Delta
        + \,\Big(\hat U^\top \hat U-U^\top U\Big)
    \end{align*}
    
    \item (Norm decomposition)
    \(\|U\Delta^\top\|_{\fro}^2
    =\|U^\top\Delta\|_{\fro}^2+\|E\,G^{1/2}\|_{\fro}^2\geq \|U^\top\Delta\|_{\fro}^2\) with equality if $G\succ 0$ and $E=0$.
    \end{enumerate}
\end{lemma}
\begin{proof}[Proof of \Cref{lem:procrustes_basic}]
    Expanding the objective,
    \[
    \|\hat U-UQ\|_{\fro}^2=\|\hat U\|_{\fro}^2+\|U\|_{\fro}^2-2\tr(Q^\top U^\top \hat U),
    \]
    so $R$ maximizes $\tr(Q^\top C)$ over $Q\in\mathcal O(k)$, where $C:=U^\top\hat U$.
    
    \paragraph{(i) First-order optimality.}
    Take an SVD $C=V\Sigma W^\top$ with $\Sigma\succeq 0$. A standard von Neumann trace inequality argument yields an optimizer
    \[
    R=VW^\top,\qquad\text{so that}\qquad
    CR^\top = V\Sigma V^\top.
    \]
    Hence, $U^\top\hat U R^\top = CR^\top$ is symmetric positive semidefinite. Moreover,
    \[
    U^\top\Delta
    =U^\top\hat U - U^\top U R
    =CR^\top R - GR
    =(CR^\top - G)R,
    \]
    and since $CR^\top$ and $G$ are symmetric, $(CR^\top-G)$ is symmetric. Right-multiplication by the fixed $R$ preserves symmetry of the $k\times k$ matrix in the $UR$-frame; equivalently $(UR)^\top\Delta$ is symmetric.
    
    \paragraph{(ii) Identities.}
    Let $M:=U^\top\Delta$. Using $\langle X,Y\rangle=\tr(X^\top Y)$ and cyclicity of trace,
    \[
    \langle U\Delta^\top U,\Delta\rangle
    =\tr\!\big((U\Delta^\top U)^\top\Delta\big)
    =\tr\!\big(U^\top\Delta\,U^\top\Delta\big)
    =\tr(M^2).
    \]
    Similarly,
    \[
    \langle \Delta U^\top,\,U\Delta^\top\rangle
    =\tr\!\big((\Delta U^\top)^\top(U\Delta^\top)\big)
    =\tr\!\big(U\Delta^\top U\Delta^\top\big)
    =\tr\!\big((\Delta^\top U)^2\big)
    =\tr(M^2).
    \]
    By (i), $M$ is symmetric (in the aligned frame), hence $\tr(M^2)=\tr(M^\top M)=\|M\|_{\fro}^2\ge 0$.
    Expand the Gram matrix of $\hat U=U+\Delta$ and rearranging gives the last identity.
    
    \paragraph{(iii) Norm decomposition.}
    With $P:=UG^\dagger U^\top$, set $E:=(I-P)\Delta$ and $A:=G^\dagger U^\top\Delta$, so $\Delta=UA+E$ and $U^\top E=0$. Then
    \[
    \|U\Delta^\top\|_{\fro}^2
    =\tr\!\big((U\Delta^\top)^\top(U\Delta^\top)\big)
    =\tr\!\big(\Delta\,G\,\Delta^\top\big)
    =\tr\!\big(G\,\Delta^\top\Delta\big).
    \]
    Moreover,
    \[
    \Delta^\top\Delta=(UA+E)^\top(UA+E)=A^\top G A + E^\top E,
    \]
    since $U^\top E=0$ implies $(UA)^\top E=A^\top U^\top E=0$. Therefore,
    \[
    \|U\Delta^\top\|_{\fro}^2
    =\tr(GA^\top G A)+\tr(GE^\top E)
    =\|GA\|_{\fro}^2+\|EG^{1/2}\|_{\fro}^2.
    \]
    Finally, $U^\top\Delta=GA$, so $\|GA\|_{\fro}^2=\|U^\top\Delta\|_{\fro}^2$, giving
    \[
    \|U\Delta^\top\|_{\fro}^2
    =\|U^\top\Delta\|_{\fro}^2+\|EG^{1/2}\|_{\fro}^2
    \ge \|U^\top\Delta\|_{\fro}^2,
    \]
    as claimed.
\end{proof}

\begin{lemma}[Optimal weighted fraction]\label{lem:opt}
    Let $c_1,c_2,d_1,d_2 > 0$ and define $F:[0,1]\to(0,\infty)$ by
    \[
    F(\lambda)
    = \frac{\lambda^2 c_1 + (1-\lambda)^2 c_2}{\bigl(\lambda d_1 + (1-\lambda)d_2\bigr)^2},
    \qquad \lambda\in[0,1].
    \]
    Then $F$ attains its unique global minimum on $[0,1]$ and the minimal value, given by:
    \[
    \lambda^* \;=\; \frac{c_2 d_1}{c_1 d_2 + c_2 d_1},\qquad
    F(\lambda^*) \;=\; \frac{c_1 c_2}{c_1 d_2^2 + c_2 d_1^2}.
    \]
\end{lemma}

\begin{proof}[Proof of \Cref{lem:opt}]
    Define the numerator and denominator
    \[
    N(\lambda) := \lambda^2 c_1 + (1-\lambda)^2 c_2,
    \qquad
    D_0(\lambda) := \lambda d_1 + (1-\lambda)d_2.
    \]
    Then we can write
    \(
    F(\lambda) = {N(\lambda)}/{D_0(\lambda)^2}
    \).
    Since $d_1,d_2>0$, we have $D_0(\lambda)>0$ for all $\lambda\in[0,1]$, hence $F$ is continuously differentiable on $[0,1]$.
    
    We first compute the derivative. Differentiating $N$ and $D_0$ gives
    \[
    N'(\lambda) = 2(c_1 + c_2)\lambda - 2c_2,
    \qquad
    D_0'(\lambda) = d_1 - d_2.
    \]
    Using the quotient rule,
    \[
    F'(\lambda)
    = \frac{N'(\lambda)D_0(\lambda)^2 - N(\lambda)\cdot 2D_0(\lambda)D_0'(\lambda)}{D_0(\lambda)^4}
    = \frac{G(\lambda)}{D_0(\lambda)^4},
    \]
    where
    \(
    G(\lambda) := N'(\lambda)D_0(\lambda)^2 - 2N(\lambda)D_0(\lambda)D_0'(\lambda).
    \)
    Since $D_0(\lambda)>0$ on $[0,1]$, the sign of $F'(\lambda)$ is the sign of $G(\lambda)$.
    
    A direct expansion and simplification of $G(\lambda)$ yields a linear function of $\lambda$:
    \(
    G(\lambda)
    = 2\bigl(c_1 d_2 + c_2 d_1\bigr)\lambda - 2c_2 d_1.
    \)
    Thus
    \[
    F'(\lambda)
    = \frac{2\bigl(c_1 d_2 + c_2 d_1\bigr)\lambda - 2c_2 d_1}{D_0(\lambda)^4}.
    \]
    The critical points of $F$ inside $(0,1)$ are the solutions of $F'(\lambda)=0$, i.e.\ the solutions of $G(\lambda)=0$:
    \[
    2\bigl(c_1 d_2 + c_2 d_1\bigr)\lambda - 2c_2 d_1 = 0
    \quad\Longrightarrow\quad
    \lambda^* = \frac{c_2 d_1}{c_1 d_2 + c_2 d_1}.
    \]
    Because $c_1,c_2,d_1,d_2>0$, we have $0<\lambda^*<1$, so $\lambda^*$ lies in the interior of $[0,1]$.
    
    Next, we show that $\lambda^*$ is the unique global minimizer. Since $D_0(\lambda)^4>0$ on $[0,1]$, the sign of $F'(\lambda)$ is the sign of the linear function
    \(
    G(\lambda) = 2\bigl(c_1 d_2 + c_2 d_1\bigr)\bigl(\lambda - \lambda^*\bigr).
    \)
    Because $c_1 d_2 + c_2 d_1 > 0$, it follows that
    \[
    F'(\lambda) < 0 \quad \text{for } \lambda \in [0,\lambda^*), 
    \qquad
    F'(\lambda) > 0 \quad \text{for } \lambda \in (\lambda^*,1].
    \]
    Hence $F$ is strictly decreasing on $[0,\lambda^*)$ and strictly increasing on $(\lambda^*,1]$, which implies that $F$ attains its unique global minimum on $[0,1]$ at $\lambda^*$.
    
    It remains to compute $F(\lambda^*)$. Substituting $\lambda^*$ into the expression for $F$ and simplifying, we obtain
    \[
    F(\lambda^*)
    = \frac{c_1 c_2}{c_1 d_2^2 + c_2 d_1^2}.
    \]
    This completes the proof.
\end{proof}

\clearpage

\section{Additional simulation details and results}\label{app:sec:simulation}

\paragraph{Organization.}
This appendix
gives the full data-generating settings and additional analyses, including
optimization behavior, downstream network analysis, sample-size scaling,
binary-network robustness, model misspecification, and baseline comparisons.
The corresponding simulation parameters are summarized in
\Cref{tab:simulation-settings}.

\begin{table}[!ht]
\centering \small
\caption{Simulation settings used for the two main studies and appendix checks.}
\label{tab:simulation-settings}
\begin{tabularx}{0.98\textwidth}{l l l l c X}
    \toprule
    Role & $n_Q$ & $|Q^c|$ &$d$& Repetitions & Purpose \\
    \midrule
    Study 1, SNR weighting & 80 & 160&60 & 100 & Vary $(\sigma_A,\sigma_W)$ across network-strong and proxy-strong regimes. \\
    Study 1, convergence & 150 & 300 & 80&20 & Track objective and latent-error gaps for fixed $\lambda_2=0.5$. \\
    Study 2, extra nodes & 30 & $n_R^{(1)}$ & 120& 100 & Compare informative and null proxies in a hard target-subset design. \\
     Appendix, downstream & 90 & $n_R^{(2)}$ &80 & 100 & Evaluate target-node clustering and full-network imputation. \\
    Appendix, sample size & 60, 120, 240 & $2n_Q$ &60 &100 & Check that estimation errors decrease with target-network size. \\
    Appendix, binary robustness & 100 & 200 & 80& 50 & Evaluate held-out prediction for edge densities 5 percent and 20 percent. \\
    \bottomrule
\end{tabularx}
\end{table}

\paragraph{Data-generating model.}
For the continuous experiments, latent rows are generated from three cluster centers with Gaussian perturbations.
The embedding loading matrix $B\in\RR^{d\times k}$ has Gaussian columns normalized to have scale $\sqrt d$.
Given $U$ and $B$, the target network and proxy embedding are generated as
\[
    A_Q = U_Q U_Q^\top + \sigma_A D,\qquad
    W = U B^\top + \sigma_W E,
\]
where $D$ is symmetrized Gaussian noise and $E$ has independent Gaussian entries. 

For Study 1, we consider three signal regimes: \((\sigma_A,\sigma_W)=(0.2,1.2)\) for the network-strong/proxy-weak regime, \((0.5,0.5)\) for the balanced regime, and \((1.2,0.2)\) for the network-weak/proxy-strong regime.
We use the tuning grid
\[
\lambda_2\in\{0.0,0.01,0.03,0.1,0.2,0.4,0.6,0.8,1.0\}.
\]
Each setting is repeated 100 times.
For the optimization experiment, we set \((\sigma_A,\sigma_W)=(0.25,0.25)\) in the high-SNR scenario and \((0.6,0.6)\) in the moderate-SNR scenario. Each setting is repeated 20 times. 

For Study 2, we also consider a null-proxy experiment, where \(W\) is generated from independent latent factors while \(A_Q\) is still generated from \(U_QU_Q^\top\).
{We vary the number of embedding-only nodes over
\[
n_R^{(1)} \in \{0,100,200,300,400,500\},
\]
set \((\sigma_A,\sigma_{W_Q},\sigma_{W_{Q^c}})=(1.4,1.6,0.08)\)}, and use the tuning grid
\[
\lambda_2\in\{0,0.02,0.05,0.10,0.20,0.40,0.60,0.80,1\},
\qquad
\lambda_1=1-\lambda_2.
\]

For the downstream community-recovery study, we consider an informative-proxy setting with imbalanced target communities.
The noise levels are \((\sigma_A,\sigma_W)=(0.70,0.50)\).
We vary the number of embedding-only nodes over
$n_R^{(2)}\in\{0,100,200\}$.
The tuning grid is
\[
\lambda_2\in\{0.0,0.02,0.05,0.1,0.2,0.4,0.6,0.8,1.0\},
\qquad
\lambda_1=1-\lambda_2.
\]
Each setting is repeated 100 times, with \(10\%\) of target nodes held out for validation.

For the sample-size study, we use an informative proxy with balanced target communities.
We set rank 
\((\sigma_A,\sigma_W)=(0.55,0.55)\).
The tuning grid is
\[
\lambda_2\in\{0.0,0.02,0.05,0.1,0.2,0.4,0.6,0.8,1.0\},
\qquad
\lambda_1=1-\lambda_2.
\]
Each setting is repeated 100 times, with \(10\%\) of target nodes held out for validation.

For the binary-network robustness study, we generate binary target networks from
\[
A_{ij}\mid U \sim \operatorname{Bernoulli}
\left\{\operatorname{logit}^{-1}(\alpha+\beta u_i^\top u_j)\right\}.
\]
We set \(\beta=1\) and calibrate \(\alpha\) to match the target edge density.
We consider two density levels, \(0.20\) and \(0.05\), corresponding to moderate and sparse binary networks.
The proxy is informative, with \(\sigma_A=\sigma_W=0.50\).
The tuning grid is
\[
\lambda_2\in\{0.0,0.02,0.05,0.1,0.2,0.4,0.6,0.8,1.0\},
\qquad
\lambda_1=1-\lambda_2.
\]
For each density level, we run 50 repetitions and hold out \(10\%\) of target nodes for validation. The estimator is still the same least-squares PLANE procedure, so the binary experiment is interpreted only as a robustness check.

\paragraph{Tuning and evaluation.}
For continuous networks, PLANE-CV selects $\lambda_2$ by the unweighted held-out squared prediction error on 10 percent of target-network entries.
For the binary robustness experiment, the held-out selection criterion is average precision.
PLANE-oracle selects $\lambda_2$ by the true relative Procrustes error of $\widehat U_Q$.
In all cases, lambda values are not compared by the weighted training objective, since a common rescaling of $(\lambda_1,\lambda_2)$ changes the loss scale but not the optimizer.

The relative error of $U_Q$ is
\[
    \frac{\dist(\widehat U_Q,U_Q)}{\|U_Q\|_{\fro}},
\]
where the distance is minimized over orthogonal rotations.
We also report the relative target-network error
\[
    \frac{\|\widehat U_Q\widehat U_Q^\top-U_QU_Q^\top\|_{\fro}}
         {\|U_QU_Q^\top\|_{\fro}}.
\]
 For each replicate \(r\) and iteration \(t\), we define the objective gap as
\[
\mathrm{Gap}_{r,t}
=
\mathcal L_{r,t}-\mathcal L_{r,T},
\]
where \(\mathcal L_{r,T}\) is the final objective value using \eqref{eq:loss} of the same optimization run, and we report \(\mathrm{Gap}_t\) averaged over replicates.
For adjusted Rand index for target-node clustering, we take the estimated target latent
positions
\[
    \widehat U_Q = \widehat U_{1:n_Q},
\]
apply \(K\)-means clustering with \(K\) equal to the true number of simulated
communities. The Adjusted Rand index measures how well the estimated target
latent positions recover the underlying community structure among the target
nodes.
We report the full-network imputation error (relative full-network
error), defined as
\[
    \frac{
    \bigl\|
        \widehat U \widehat U^\top - U U^\top
    \bigr\|_{\mathrm F}
    }{
    \bigl\|
        U U^\top
    \bigr\|_{\mathrm F}
    }.
\]
For the binary network, to evaluate out-of-sample edge prediction, we randomly hold out a fraction of
target--target edges from the upper triangular off-diagonal entries and then
symmetrize the holdout mask. The model is fitted only on the remaining observed
network entries. After fitting, we compute the continuous fitted target network
\[
\widehat A_Q=\widehat U_Q\widehat U_Q^\top
\]
and use its held-out entries as edge scores. Held-out AUC is the ROC AUC for
predicting \(A_{Q,ij}=1\) from \(\widehat A_{Q,ij}\), and held-out average
precision is the corresponding average precision.

{\subsection{Misspecification robustness}
To evaluate the robustness of PLANE beyond the data-generating model used in the main simulations, we considered several misspecified simulation settings. The baseline simulation used \(n_Q=100\) target nodes, \(200\) proxy-only nodes, latent rank \(r=3\), \(80\) embedding features, and three latent communities. The target network and proxy embedding were generated from the same latent positions with Gaussian network noise level \(0.35\), embedding noise level \(0.55\), and proxy strength \(1.0\). We then modified this baseline setting to create three departures from the working model: 
(i) a nonlinear proxy embedding, replacing the baseline embedding by
\[
W
=
s_X \tanh\{1.4 UB^\top / s_X\}+\sigma_W E,
\]
where \(s_X=\mathrm{sd}(UB^\top)\) is the empirical scale of the baseline embedding mean;
(ii) a hub-residual target network, replacing the target-network mean by
\[
A_Q = U_QU_Q^\top + R_{\mathrm{hub}}+\sigma_A D,
\]
where \(R_{\mathrm{hub}}\) is a sparse hub-like residual matrix not explained by the low-rank component; 
(iii) heavy-tailed target-network noise, replacing the Gaussian noise \(D\) by standardized \(t_3\) noise; and 
we compared the network-only estimator \((\lambda_2=0)\), PLANE-CV selected by held-out observed MSE, and a PLANE oracle selected by the relative error to the true target-network mean. Performance was summarized by held-out observed MSE, relative target-network error, and latent-position recovery error across the 50 replicates.

\Cref{fig:simulation-missspecification} shows that PLANE remains stable under several departures from the simulation model. 
Compared with the network-only estimator, PLANE-CV achieves similar or lower held-out observed MSE across the baseline and misspecified settings, indicating that the cross-validation procedure does not overuse the embedding information when the model is perturbed. 
For target-network estimation, PLANE also remains competitive in most settings and shows a clear improvement under the heavy-tailed noise case, while the hub-residual setting is the most challenging because it introduces target-network structure that is not explained by the low-rank latent model. 
The clustering ARI is broadly similar across methods, suggesting that the main gains are in denoising and network estimation rather than changing the recovered community partition. Across these misspecified settings, PLANE remains competitive with or improves over the network-only estimator, with the largest attenuation occurring under the hub-residual network setting where the target network contains structure outside the low-rank working model.
\begin{figure}[!ht]
    \centering
    \includegraphics[width=1\textwidth]{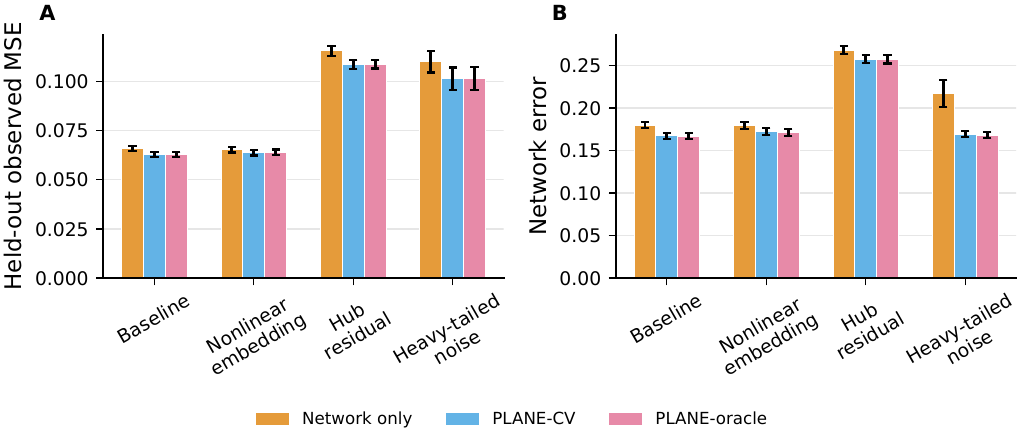}
       \caption{{Robustness to model misspecification.
We compare the network-only estimator, PLANE-CV, and PLANE-oracle under the baseline simulation setting and three departures from the working model: nonlinear proxy embedding, hub-residual target network, and heavy-tailed target-network noise.
(A) Held-out observed MSE on masked target-network entries.
(B) Relative error of the estimated target-network mean.
Bars show means over 50 simulation replicates and error bars show \(\pm 2\) SE.}}
    \label{fig:simulation-missspecification}
\end{figure}

\subsection[Baseline comparison with Q-only and embedding-only estimators]{Baseline comparison with \(Q\)-only and embedding-only estimators}
To examine whether the improvement of PLANE is driven by the proposed use of
auxiliary embedding-only nodes rather than by a generic low-rank factorization,
we considered an additional benchmark simulation. We compared the full PLANE
estimator with a \(Q\)-only PLANE variant \cite{zhang2022joint}, which fits the
same objective after removing the auxiliary nodes \(Q^c\), and with
network-only and embedding-only estimators. The network-only estimator
corresponds to \(\lambda_2=0\), using only the target network, whereas the
embedding-only estimator corresponds to \(\lambda_2=1\), using only the proxy
embedding. The \(Q\)-only PLANE variants use only the target genes \(Q\), we compute the rank-\(k\) truncated SVD of the
proxy embedding matrix \(W=\widetilde U \Sigma \widetilde V^\top\), and set
\(\widehat U=\widetilde U\Sigma^{1/2}\) and
\(\widehat B=\widetilde V\Sigma^{1/2}\). And it ignores the target
network, it requires stronger balanced/pervasive factor conditions than PLANE to
identify the target latent positions. For each variant, we report both CV selection of
\(\lambda_2\) and an oracle choice minimizing the true latent-position error.

We used a challenging weak-target setting with \(n_Q=25\) target nodes and
\(200\) auxiliary embedding-only nodes. The latent rank was \(r=3\), the
embedding dimension was \(200\), and the target network was observed with
Gaussian noise level \(0.75\). To make the target-only information limited, the
target embedding noise was set to \(1.20\), while the auxiliary embedding-only
nodes were generated with a much smaller embedding noise level \(0.20\). We also
weakened the target latent and loading signals by using community signal
\(0.45\), latent noise \(0.35\), and loading scale \(0.35\). This setting
represents a regime where the target set alone is too small and noisy to stably
identify the shared latent embedding structure, but a large auxiliary embedding
sample can help stabilize estimation.

\Cref{fig:simulation-baseline} reports the benchmark comparison under two
auxiliary-information regimes: informative and null. In the informative regime,
the auxiliary embedding is generated from the same latent structure as the target
nodes; in the null regime, it is generated from unrelated latent positions.
Across 50 replicates, tuning was performed over the same \(\lambda_2\) grid as
in the main simulations, and performance was summarized by latent-position
recovery error and target-network estimation error. Under the informative proxy regime, full PLANE gives the best latent-position
and target-network estimation errors, showing that auxiliary embedding-only
nodes improve estimation beyond the \(Q\)-only fit. Under the null regime, this
advantage is attenuated and the embedding-only estimator performs poorly,
confirming that gains arise from informative auxiliary embeddings rather than
from adding unrelated proxy nodes.

\begin{figure}[!ht]
    \centering
    \includegraphics[width=1\textwidth]{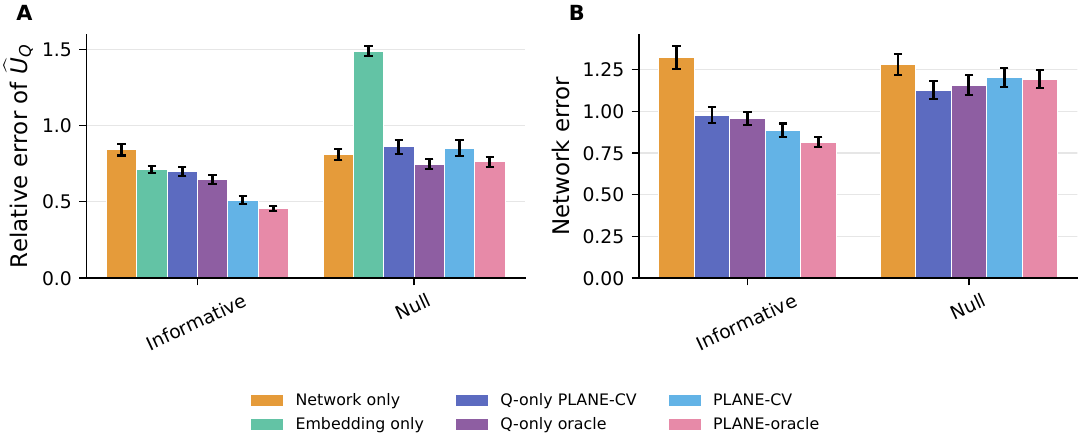}
\caption{{
Comparison with \(Q\)-only and embedding-only baselines.
We compare the network-only estimator, the embedding-only estimator, \(Q\)-only
PLANE, and full PLANE under two proxy-embedding regimes: informative and null.
The \(Q\)-only PLANE variants fit the same objective after removing the
auxiliary embedding-only nodes \(Q^c\), while
the PLANE use both the target genes \(Q\) and auxiliary
embedding-only genes \(Q^c\). 
(A) Relative error of the estimated target latent positions \(\widehat U_Q\).
(B) Relative error of the estimated target-network mean.
Bars show means over 50 simulation replicates and error bars show \(\pm 2\) SE.
}}
    \label{fig:simulation-baseline}
\end{figure}

}

\subsection{Additional simulation diagnostics and robustness checks}
\paragraph{Optimizer behavior.}
As a complement to Study 1, we examine the convergence behavior of the
optimization algorithm. \Cref{fig:simulation-1-convergence} Panels A--B fix
\(\lambda_2=0.5\) and show that both the objective gap and the excess latent
error decrease approximately linearly on the log scale before reaching a
numerical or statistical floor. This pattern is consistent with the contraction
result in \Cref{thm:error_algor}.

The objective gap in Panel A is computed relative to the final iterate rather
than the exact global optimum. Consequently, the curve may bend downward near
the last few iterations as the iterates approach this reference value.
Moreover, the PLANE criterion is nonconvex due to the quadratic latent network
term \(\|A_Q-U_QU_Q^\top\|_{\mathrm F}^2\) and its coupling with the embedding
reconstruction term \(\|W-UB^\top\|_{\mathrm F}^2\). In addition, the normalized
gradient update with line search induces iteration-dependent effective step
sizes. Therefore, different phases of the algorithm may exhibit different
rates of decrease, which can further lead to deviations from an exact
log-linear pattern.

\begin{figure}[!ht]
    \centering
    \includegraphics[width=0.90\textwidth]{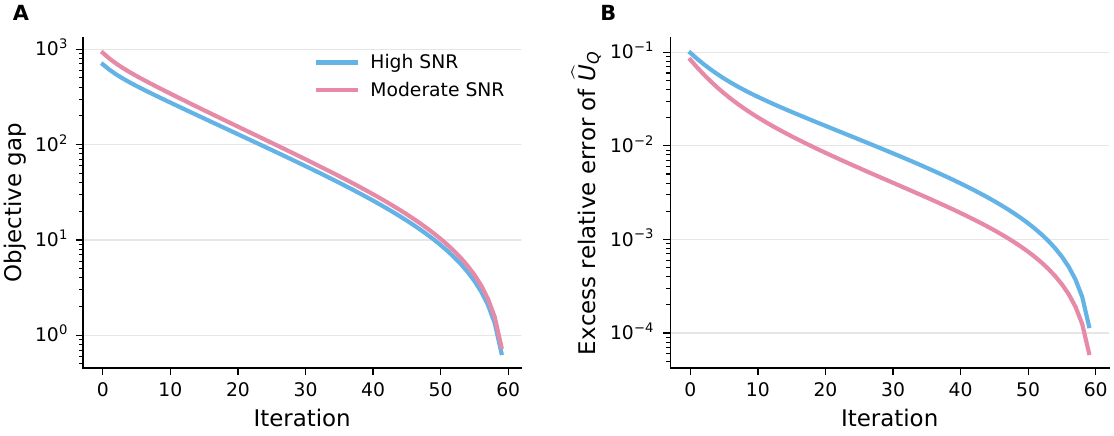}
        \caption{Panels D--E display normalized-gradient optimization trajectories at \(\lambda_2=0.5\)}
    \label{fig:simulation-1-convergence}
\end{figure}
\paragraph{Downstream network analysis.}
This study verifies that the latent recovery gain is visible in network-analysis tasks.
In the imbalanced target-subset design with $|Q^c|=200$, PLANE-CV improves the mean adjusted Rand index from $0.598$ for network-only to $0.675$, and reduces full-network imputation error from $0.940$ to $0.125$.
Thus, the benefit of estimating the shared latent representation is not limited to the Procrustes metric used in the theory.


\begin{figure}[!ht]
    \centering
    \includegraphics[width=0.90\textwidth]{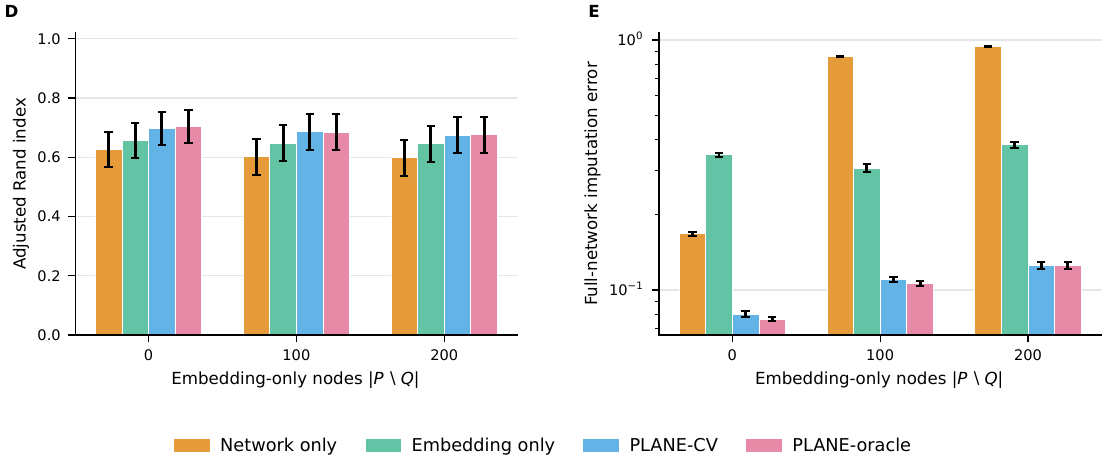}
        \caption{Downstream clustering and network imputation.
    Bars compare network-only, embedding-only, PLANE-CV, and PLANE-oracle estimators.
    Error bars are $\pm$ two standard errors over 100 repetitions.}
    \label{fig:simulation-study-downstream}
\end{figure}
\paragraph{Sample-size study.}
\Cref{fig:appendix-simulation-sample-size} gives the sample-size check.
PLANE-CV has decreasing mean relative latent error as $n_Q$ increases, from $0.103$ at $n_Q=60$ to $0.068$ at $n_Q=240$, and remains close to the oracle selector.
The target-network error shows the same monotone trend.

\begin{figure}[!ht]
    \centering
    \includegraphics[width=0.90\textwidth]{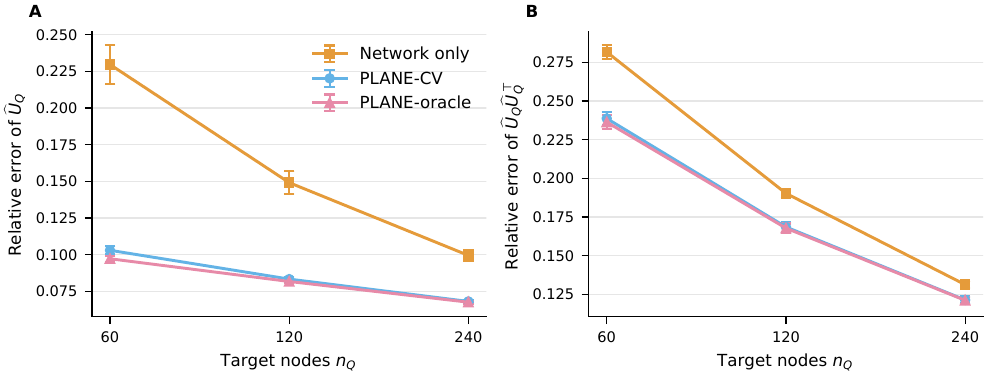}
    \caption{Sample-size scaling under the continuous Gaussian simulation.
    Curves show mean errors with $\pm$ two standard errors over 100 repetitions.}
    \label{fig:appendix-simulation-sample-size}
\end{figure}

\paragraph{Binary-network robustness.}\Cref{fig:appendix-simulation-binary} reports the binary-network robustness study. 
At 5 percent density, PLANE-CV improves held-out average precision from $0.110$ to $0.192$ and AUC from $0.608$ to $0.734$ relative to network-only.
At 20 percent density, average precision improves from $0.368$ to $0.431$ and AUC from $0.692$ to $0.741$.
These results support practical robustness, but they are not used as evidence for the Gaussian theory.

\begin{figure}[!ht]
    \centering
    \includegraphics[width=0.90\textwidth]{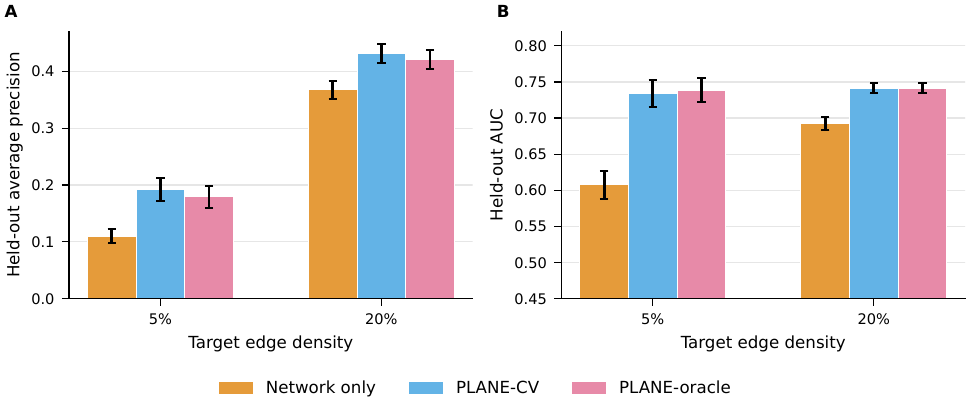}
    \caption{Binary-network robustness.
    PLANE is fit as a least-squares working estimator and evaluated by held-out edge prediction.
    Bars show means with $\pm$ two standard errors over 50 repetitions.}
    \label{fig:appendix-simulation-binary}
\end{figure}
\newpage
\section{Real-data application details and supplementary results}\label{app:sec:application}

\paragraph{Data construction.}
\Cref{tab:application-design} summarizes the analysis choices used for \Cref{sec:application}. { We construct the target graph from signed differential-expression signatures.
For each perturbation gene, perturbed cells are compared with control cells using a two-sample \(t\)-test implemented in \texttt{crispyx} \citep{du2026crispyx}.
A response gene is kept in the signed perturbation signature if its BH-adjusted differential-expression \(p\)-value is at most \(0.001\) and its absolute log-fold change is at least \(0.5\), with the sign indicating up- or down-regulation. 
For each pair of perturbation genes, we count the number of same-direction overlapping response genes and test this overlap using a hypergeometric test. 
We include a binary edge only when the same direction overlap is at least five and the BH-adjusted edge \(p\)-value is at most
\(0.01\). 
Thus, the binary target network encodes statistically significant shared perturbation effects rather than raw continuous expression similarity.
Sensitivity to using a continuous weighted target network is reported in \Cref{fig:appendix-application-weighted-network}.}

The response-signature matrix used to construct the graph is retained separately in the processed-data output, so the target adjacency \(A_Q\) is distinguishable from the signed differential-expression signatures. The tuning parameter is selected by held-out target-network criterion: $10\%$ of target genes are held out as nodes, all incident edges are masked during fitting,  and \(\lambda_2\) minimizes
held-out mean squared error. We select \(\lambda_2\) from the pre-specified grid $\{0.0, 0.001, 0.003, 0.01, 0.03, 0.1, 0.2, 0.4, 0.6, 0.8, 1.0\}$.

{For Panels B--C in \Cref{fig:application-biology}, candidate embedding-only genes were ranked once by their
maximum cosine similarity to the target-gene embeddings. Specifically, for each
candidate gene \(g\), we computed
\[
s_g=\max_{h\in Q}\frac{\langle W_g,W_h\rangle}{\|W_g\|_2\|W_h\|_2},
\]
using the standardized proxy embedding \(W\). For each
\(m\in\{0,50,200,500,800\}\), \(Q^c\) was formed by taking the top \(m\)
ranked genes, so the extra-node sets are nested. For each \(m\), we repeated
the validation procedure over 20 random subsamples of the held-out target
network, selected \(\lambda_2\) by held-out target-network MSE over the tuning
grid, and report the mean and standard error across repetitions.}
\begin{table}[!ht]
\centering \small
\caption{Design summary for the real-data application.}
\label{tab:application-design}
\begin{tabularx}{0.98\textwidth}{l X}
    \toprule
    Component & Choice in the real-data application \\
    \midrule
    Perturbational data & CRISPRa Perturb-seq single-gene perturbations in K562 cells \citep{norman2019exploring}. \\
    Target genes & \(n_Q=99\) perturbation genes after differential-expression filtering and embedding alignment. \\
    Target network \(A_Q\) & Binary signed-DE-overlap graph; edges require sufficient same-direction DE overlap and BH-adjusted overlap evidence. \\
    Proxy matrix \(W\) & PRESAGE-derived concatenation of GO biological-process (genes=25606, dim=128), TF-target (genes=28962, dim=128), STRING (genes=19485, dim=128), GenePT text (genes=93800, dim=1536), ESM protein-language (genes=19790, dim=5120), and GenePT protein embeddings (genes=133736, dim=3072), after source-wise standardization and dimension-normalizing each source block.
    \\
    Embedded universe & \(n_P=1099\) genes, including 1000 embedding-only genes selected near the target set in embedding space. \\
    Embedding dimension & \(d=10112\) standardized and source-scaled embedding coordinates. \\
    Estimator & Rank \(k=8\) PLANE with \(\lambda_1=1-\lambda_2\). \\
    Selection rule & Hold out 10 percent of target genes as nodes and select \(\lambda_2\) by unweighted held-out target-network MSE. \\
    Selected fit & \((\lambda_1,\lambda_2)=(0.40,0.60)\), final refit convergence after 257 iterations. \\
    \bottomrule
\end{tabularx}
\end{table}

{\subsection{Low-rank structure of biological embeddings}
To assess the effective dimensionality of the proxy embeddings, we performed a source-wise singular-value gap analysis. 
For each embedding source, we extracted its corresponding feature block from the processed concatenated PRESAGE embedding matrix. 
After centering and scaling the features as in the main analysis, we computed the singular values of this source-specific gene-by-feature matrix, denoted by
\(
\sigma_1 \geq \sigma_2 \geq \cdots .
\)
We then examined the adjacent singular-value gap ratio \(\sigma_j/\sigma_{j+1}\) for the leading components. 
Large ratios indicate potential spectral elbows, whereas ratios close to one indicate that the remaining components form a relatively flat residual tail.

As shown in \Cref{fig:eigengap-total}, the eigengap curve of the concatenated
proxy embedding drops after the leading components and becomes nearly flat,
supporting rank \(k=8\) as a parsimonious working latent dimension for the
real-data analysis. \Cref{fig:eigengap} shows the eigengap diagnostics for the six proxy embedding sources. 
Across sources, the gap ratios are elevated only among the first few components and become close to one after approximately 15 components. 
This suggests that each high-dimensional embedding block contains a relatively small number of dominant low-rank directions, followed by a stable residual spectrum. 
The observed elbows, ranging from roughly 6 to 15 components across sources, provide empirical support for modeling the proxy embeddings through a low-dimensional shared latent representation.
 \begin{figure}[H]
        \centering
    \includegraphics[width=0.7\linewidth]{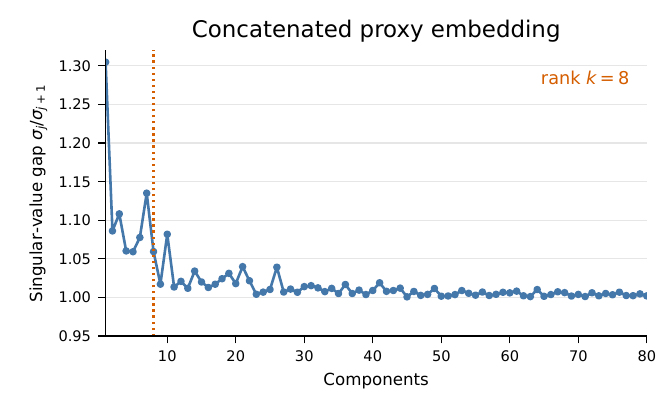}
     \caption{{Eigengap diagnostic for the concatenated proxy embedding.
The plot shows the adjacent singular-value ratio \(\sigma_j/\sigma_{j+1}\) over candidate ranks \(j\).
The eigengap curve flattens after the leading components, supporting the use of rank \(k=8\), marked by the vertical dotted line, as a parsimonious latent-factor dimension for the real-data analysis.}}
        \label{fig:eigengap-total}
    \end{figure}
    \begin{figure}
        \centering
    \includegraphics[width=1\linewidth]{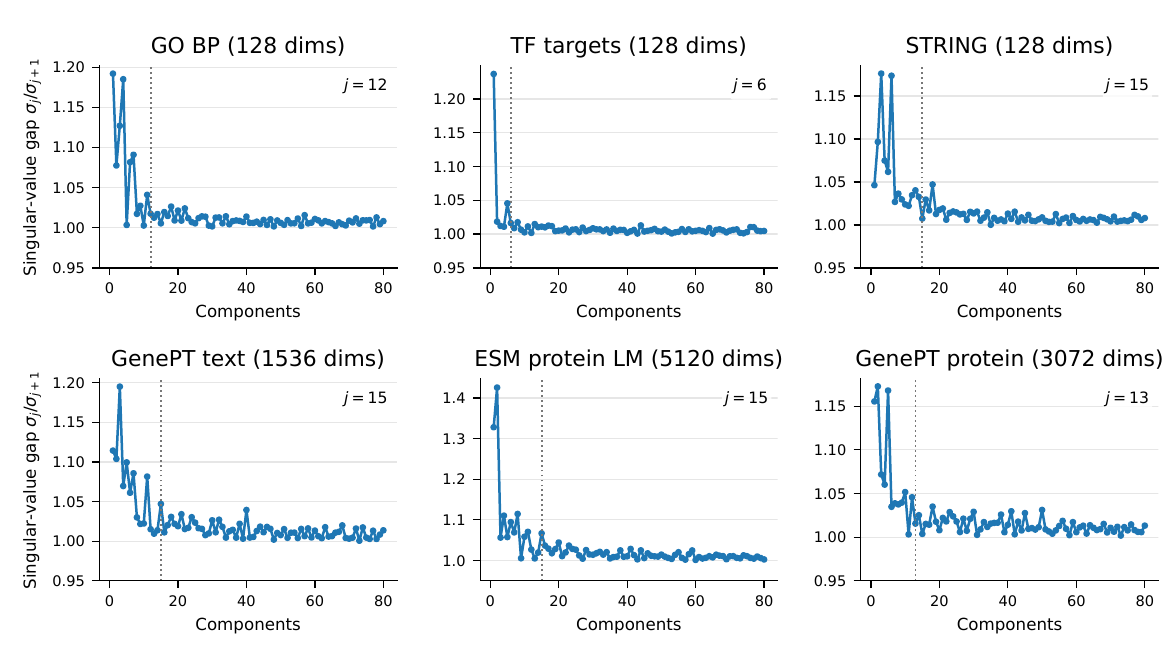}
        \caption{{Eigengap diagnostics for proxy embeddings. For each embedding source, we plot the adjacent singular-value ratio \(\sigma_j/\sigma_{j+1}\) over candidate ranks \(j\).
        Although several embedding sources have thousands of input dimensions, their eigengap curves flatten within the leading 15 components, indicating that the dominant spectral structure is much lower-dimensional than the raw feature space. We therefore use rank \(k=8\) as a parsimonious latent-factor approximation that captures the leading structure while avoiding overfitting to the high-dimensional residual tail.}}
        \label{fig:eigengap}
    \end{figure}
\subsection{Sensitivity to continuous target networks} To assess the sensitivity of PLANE to the binarization of the target network, we repeated the real-data analysis using a weighted target network constructed directly from the differential-expression response signatures. For each pair of target genes, we computed the number of same-direction differentially expressed features shared by their perturbation signatures. This overlap was converted to a hypergeometric enrichment \(p\)-value, adjusted across all gene pairs by the Benjamini-Hochberg procedure, and transformed to a continuous edge score \(-\log_{10}(q)\). The resulting scores were scaled by their 98th percentile and truncated to the interval \([0,1]\), giving a weighted target-network matrix with zero diagonal.

We then applied the same PLANE fitting and cross-validation procedure as in the main real-data analysis shown in \Cref{fig:appendix-application-weighted-network}, replacing the binary target network by this continuous target network. 
The embedding weight \(\lambda_2\) was selected by held-out target-network MSE over repeated node holdout splits. The weighted analysis selected a nonzero embedding weight and reduced the held-out MSE by 7.7\% relative to the network-only fit. Notably, the selected \(\lambda_2\) was the same as that selected in the main binary-network analysis, suggesting that the tuning of the embedding contribution is stable to the choice of target-network construction. The final PLANE fit recovered the same broad block structure as the continuous target network, and its overall modular pattern was also similar to that obtained in the binary target-network analysis. Together, these results indicate that the main real-data conclusions are not driven solely by the binarization threshold used to construct the original target network.

\begin{figure}[!ht]
    \centering
    \includegraphics[width=1\textwidth]{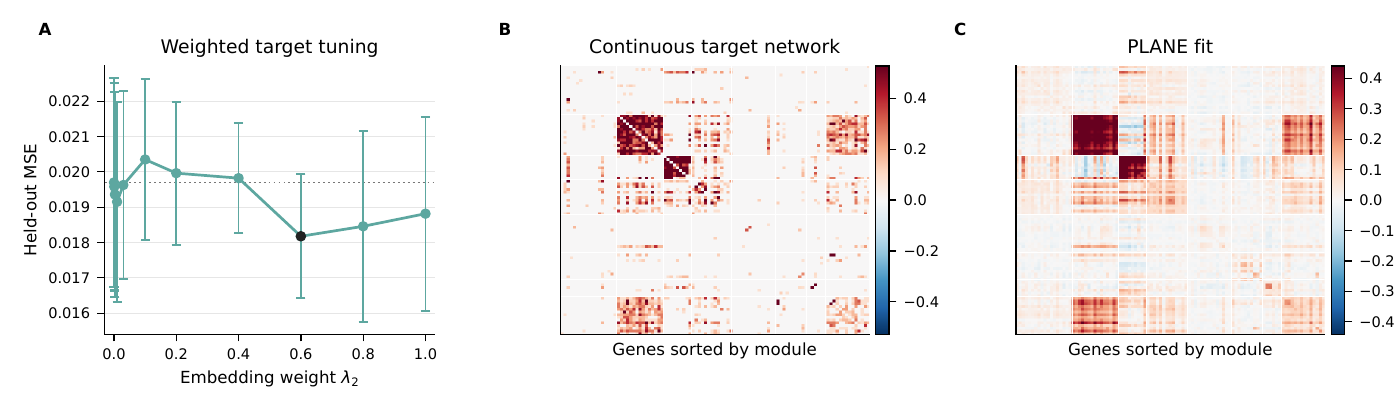}

\caption{{
Continuous-target sensitivity analysis.
(A) Held-out MSE for PLANE fitted to a weighted target network as a function of the embedding weight \(\lambda_2\). Points show the mean across repeated node holdout splits and error bars show \(\pm 2\) SE. The dotted line denotes the network-only fit \((\lambda_2=0)\), and the black point marks the selected PLANE fit. The reported gain is the relative reduction in held-out MSE compared with the network-only fit.
(B) Continuous target network constructed from continuous same-direction differential-expression overlap scores. Genes are ordered by the modules inferred from the PLANE fit.
(C) Final PLANE-fitted target network using the selected embedding weight.
}}
    \label{fig:appendix-application-weighted-network}
\end{figure}
}
\subsection{Additional diagnostics and biological interpretation}
{\paragraph{{Ordering sensitivity for embedding-only genes}}
Because PLANE uses proxy embeddings from genes outside the observed target
network, we examined whether the real-data results depend on the ordering or
selection of these embedding-only genes. We repeated the analysis using random
subsamples of embedding-only genes with varying sizes, and for each setting
recorded the held-out target-network error and the CV-selected embedding weight
\(\lambda_2\). This experiment tests whether the observed benefit of proxy
embeddings is stable to perturbations of the auxiliary node set.

The results are shown in \Cref{fig:appendix-application-ordering}. The held-out
relative error remains stable across random subsamples of embedding-only genes
and improves when larger auxiliary sets are included. At the same time, the
CV-selected embedding weight \(\lambda_2\) increases with the number of
embedding-only genes, indicating that cross-validation places more weight on the
proxy-embedding channel as more auxiliary information becomes available. These
patterns suggest that the real-data conclusions are not driven by a specific
ordering or subset of embedding-only genes.
}
\begin{figure}[!ht]
    \centering
    \includegraphics[width=1\textwidth]{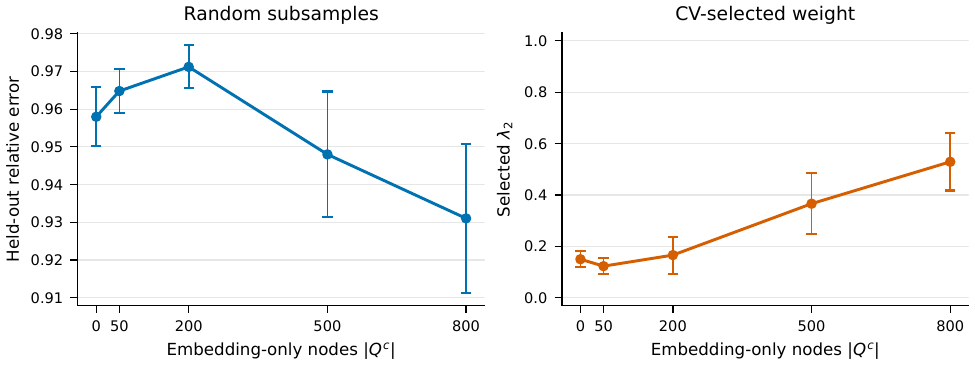}

\caption{{Sensitivity to the ordering of embedding-only nodes.
We evaluate PLANE under $10$ random subsamples of embedding-only genes to assess
whether the real-data results depend on the ordering or selection of proxy-only
nodes. Left: held-out relative error as a function of the number of
embedding-only genes included. Right: corresponding CV-selected embedding weight
\(\lambda_2\). Points show means across repeated subsamples and error bars show
\(\pm 2\) SE. The held-out error remains stable, and improves when more
embedding-only genes are included, while the selected \(\lambda_2\) increases
with proxy-only sample size, supporting robustness to the node-ordering choice.}}
    \label{fig:appendix-application-ordering}
\end{figure}
\paragraph{Optimization diagnostics.}
\Cref{fig:appendix-application-diagnostics} reports numerical diagnostics for the lambda grid and the selected final refit.
All grid fits satisfy the configured final relative objective-change tolerance \(10^{-5}\), and the selected final refit also meets this tolerance. The convergence criterion is
defined as the final relative objective change,
\[
\frac{|L^{(t)}-L^{(t-1)}|}
{\max\{|L^{(t-1)}|,1\}},
\]
where \(L^{(t)}\) denotes the objective value after iteration \(t\). Panel A shows that the selected PLANE-CV fit and the final refit exhibit
steady decreases in the objective gap, indicating stable numerical behavior.
Panels B--C show that the selected embedding weight \(\lambda_2=0.6\) lies in
a region with moderate iteration counts and that its final relative objective
change is below the prescribed tolerance.

\begin{figure}[!ht]
    \centering
    \includegraphics[width=0.96\textwidth]{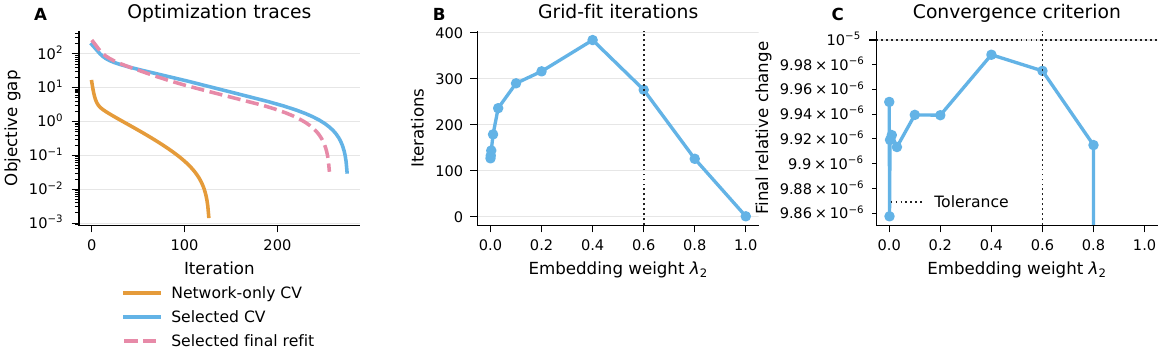}
    \caption{Optimization and selection diagnostics for the real-data application.
    Panel A shows objective-gap traces for the network-only CV fit, the selected PLANE-CV fit, and the selected final refit.
    Panel B shows the number of iterations across the lambda grid.
    Panel C shows the final relative objective change together with the configured tolerance.}
    \label{fig:appendix-application-diagnostics}
\end{figure}

\paragraph{Fitted-network displays.}
\Cref{fig:appendix-application-network} compares the observed target graph with the fitted weighted target network after ordering genes by final module labels.
The module-level average fitted weights show how the latent representation organizes the target genes into block-like neighborhoods, while still allowing between-module structure.

\begin{figure}[!ht]
    \centering
    \includegraphics[width=0.46\textwidth]{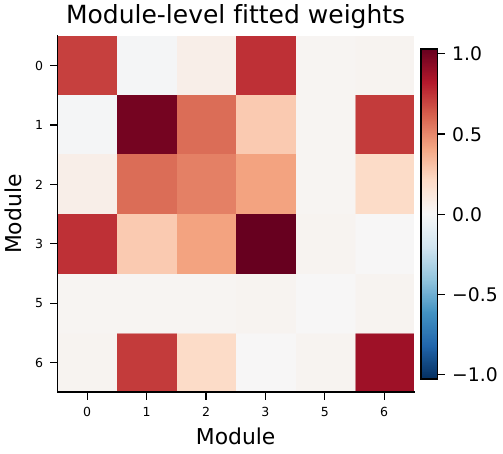}
    \caption{Supplementary fitted-network  over module pairs.}
    \label{fig:appendix-application-network}
\end{figure}
\paragraph{GO analysis.} \Cref{fig:appendix-application-go} summarizes the results of GO analysis. The GO analysis shows that the target modules capture 
regulatory and stress-response programs, whereas the embedding-expanded graph
recovers broader functional neighborhoods. The ChEA
overlay further separates these signals: the weighted target graph slightly
improves overlap with known regulatory edges compared with the network-only
graph (62/500 versus 57/500), whereas the expanded graph has no ChEA hits
among its top 1000 edges, suggesting functional similarity rather than direct
TF--target regulation.   
\begin{figure}[H]
    \centering
    \includegraphics[width=0.9\textwidth]{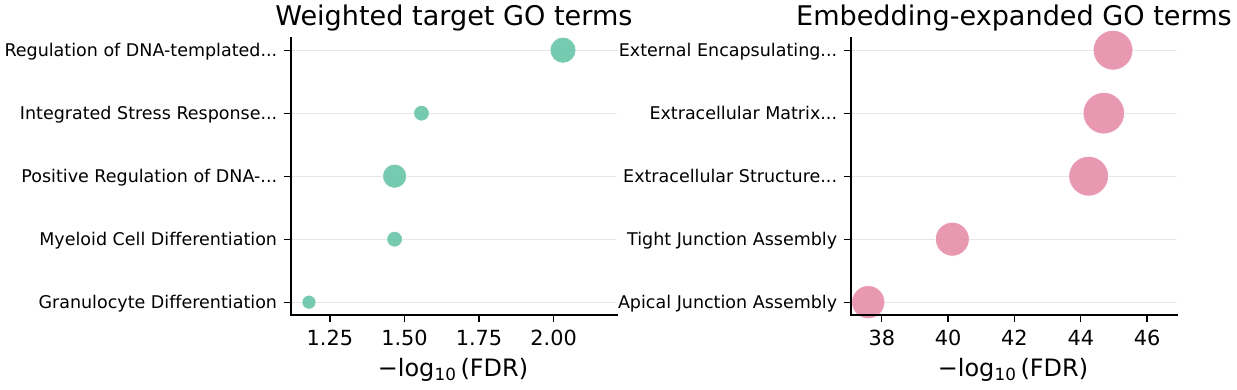}
    \caption{Leading GO enrichments for modules on $\hat U_Q$ and $\hat U$ by PLANE-CV.}
    \label{fig:appendix-application-go}
\end{figure}

\end{document}